%% file: main.tex
\documentclass{lmcs} %%% last changed 2014-08-20

\keywords{Parikh automata, Infinite Words, Model Checking}

\usepackage{amsmath}
\usepackage{amssymb}
\usepackage{amsthm}
\usepackage{caption}
\usepackage{hyperref}
\usepackage{microtype}
\usepackage{makecell}
\usepackage{mathrsfs}
\usepackage[all,defaultlines=4]{nowidow}
\usepackage{pifont}
\usepackage{tikz}
\usepackage{subcaption}
\usepackage{xcolor}
\usepackage{xspace}

\usepackage[nameinlink]{cleveref}

\usetikzlibrary{arrows,arrows.meta,automata,calc,decorations.pathreplacing,fit,calligraphy,positioning,shapes}

\crefname{lem}{Lemma}{Lemmas}
\crefname{thm}{Theorem}{Theorems}
\crefname{cor}{Corollary}{Corollaries}
\crefname{prop}{Proposition}{Propositions}
\crefname{defi}{Definition}{Definitions}
\crefname{clm}{Claim}{Claims}
\crefname{exa}{Example}{Examples}
\crefname{obs}{Observation}{Observations}

\theoremstyle{plain}\newtheorem*{remark*}{Remark}

\def\eg{{\em e.g.}}
\def\ie{{\em i.e.}}

\renewcommand{\epsilon}{\varepsilon}
\renewcommand{\phi}{\varphi}

\newcommand{\Amc}{\ensuremath{\mathcal{A}}\xspace}
\newcommand{\Bmc}{\ensuremath{\mathcal{B}}\xspace}

\newcommand{\Emc}{\ensuremath{\mathcal{E}}\xspace}
\newcommand{\Fmc}{\ensuremath{\mathcal{F}}\xspace}

\newcommand{\Kmc}{\ensuremath{\mathcal{K}}\xspace}
\newcommand{\Lmc}{\ensuremath{\mathcal{L}}\xspace}
\newcommand{\Mmc}{\ensuremath{\mathcal{M}}\xspace}

\newcommand{\Omc}{\ensuremath{\mathcal{O}}\xspace}
\newcommand{\Pmc}{\ensuremath{\mathcal{P}}\xspace}

\newcommand{\Tmc}{\ensuremath{\mathcal{T}}\xspace}

\newcommand{\Nbb}{\ensuremath{\mathbb{N}}\xspace}

\newcommand{\Zbb}{\ensuremath{\mathbb{Z}}\xspace}

\newcommand{\Nbbinfty}{\ensuremath{\Nbb_\infty}\xspace}

\newcommand{\bbf}{\ensuremath{\mathbf{b}}\xspace}
\newcommand{\ebf}{\ensuremath{\mathbf{e}}\xspace}
\newcommand{\ibf}{\ensuremath{\mathbf{i}}\xspace}
\newcommand{\pbf}{\ensuremath{\mathbf{p}}\xspace}
\newcommand{\ubf}{\ensuremath{\mathbf{u}}\xspace}
\newcommand{\vbf}{\ensuremath{\mathbf{v}}\xspace}

\newcommand{\0}{\ensuremath{\mathbf{0}}\xspace}
\newcommand{\1}{\ensuremath{\mathbf{1}}\xspace}
\newcommand{\Inf}{\ensuremath{\mathsf{Inf}}}

\newcommand{\LPAPA}{\ensuremath{\Lmc_\mathsf{PA, PA}^\omega}\xspace}
\newcommand{\LPAReg}{\ensuremath{\Lmc_\mathsf{PA, Reg}^\omega}\xspace}
\newcommand{\LRegPA}{\ensuremath{\Lmc_\mathsf{Reg, PA}^\omega}\xspace}
\newcommand{\LRegReg}{\ensuremath{\Lmc_\mathsf{Reg, Reg}^\omega}\xspace}

\newcommand{\xmark}{\ding{55}}
\newcommand{\cmark}{\ding{51}}%
\newcommand{\NO}{\textcolor{red}{\xmark}}
\newcommand{\YES}{\textcolor{green!60!black}{\cmark}}

\newcommand{\Isf}{\ensuremath{\mathsf{I}}\xspace}

\renewcommand{\P}{\texttt{P}\xspace}
\newcommand{\NP}{\texttt{NP}\xspace}
\newcommand{\coNP}{\texttt{coNP}\xspace}

\newcommand{\coNEXP}{\texttt{coNEXP}\xspace}

\newenvironment{claimproof}[1][\proofname]{%
	\begin{proof}[#1]%
	}{%
	\end{proof}%
}

\begin{document}

% If the title is longer than 55 characters, then specify a shorter running title as the optional argument to \title. The running title should be roughyl at most 55 characters:
\title[Parikh Automata on Finite and Infinite Words]{Parikh Automata on Finite and Infinite Words}	%optional

% affiliations are numbered automatically with a, b, c (see below)
% use the optional argument to indicate the affiliation(s) of each author
% omit the argument if there is only one author, or only one affiliation
\author[M. Grobler]{Mario Grobler\lmcsorcid{0000-0001-8103-6440}}[a]
\author[L. Sabellek]{Leif Sabellek\lmcsorcid{0000-0001-8051-5749}}[b]
\author[S. Siebertz]{Sebastian Siebertz\lmcsorcid{0000-0002-6347-1198}}[a]

% affiliation 1 (automatically numbered a)
\address{University of Bremen}	%optional
% write emails for all authors having that affiliation
\email{grobler@uni-bremen.de, siebertz@uni-bremen.de}  %optional

% affiliation 2 (automatically numbered b)
\address{CONTACT Software}	%optional
\email{leif.sabellek@contact-software.com}  %optional

%% etc.

%% required for running head on odd and even pages, use suitable
%% abbreviations in case of long titles and many authors:

%%%%%%%%%%%%%%%%%%%%%%%%%%%%%%%%%%%%%%%%%%%%%%%%%%%%%%%%%%%%%%%%%%%%%%%%%%%

%% the abstract has to PRECEDE the command \maketitle:
%% be sure not to issue the \maketitle command twice!
%
\begin{abstract}
 We study Parikh automata on finite and infinite words. First we establish some results for Parikh automata on finite words. Following, we present several definitions of Parikh automata on infinite words. We consider the deterministic as well as the non-deterministic variants and study closure properties, expressiveness, and common decision problems with applications to model checking. Furthermore, we compare our models to other models with counting mechanisms operating on infinite words.
\end{abstract}
\maketitle

\input{1_auto_introduction}

\input{2_auto_prelims}

\input{3_auto_finite}

\input{4_auto_infinite}

\input{4.1_auto_nondet}

\input{4.2_auto_det}

\input{4.3_auto_concl}

\bibliographystyle{alphaurl}
\bibliography{lit}

\end{document}

%% file: 1_auto_introduction.tex
\section{Introduction}

Finite automata find numerous applications in formal language theory, logic, verification, and many more, in particular due to their good closure properties and algorithmic properties. To enrich this spectrum of applications even more, it has been a fruitful direction to add features to finite automata to capture also situations beyond the regular realm. 

One such possible extension of finite automata with counting mechanisms has been introduced by Greibach in her study of blind and partially blind (one-way) multicounter machines~\cite{greibach}. Blind multicounter machines are generalized by weighted automata as introduced in~\cite{groupautomata}. 
Parikh automata (PA) were introduced by Klaedtke and Ruess in~\cite{klaedtkeruess}. A~Parikh automaton is a non-deterministic finite automaton 
that is additionally equipped with a semi-linear set~$C$, and every transition is equipped with a $d$-tuple of non-negative integers. 
Whenever an input word is read, $d$ counters are initialized with the values $0$ and every time a transition is used, the counters are incremented by the values in the tuple of the transition accordingly. 
An input word is accepted if the~PA ends in an accepting state and additionally, the resulting $d$-tuple of counter values lies in~$C$. 
Klaedtke and Ruess showed that~PA are equivalent to weighted automata over the group $(\Zbb^k, +, \0)$, and hence equivalent to Greibach's blind multicounter machines, as well as to reversal bounded multicounter machines~\cite{bakerbook,ibarra}.
Recently it was shown that these models can be translated into each other using only logarithmic space \cite{logspaceconversion}.
In this work we call the class of languages recognized by any of these models \emph{Parikh recognizable}. 
Klaedtke and Ruess~\cite{klaedtkeruess} showed that the class of Parikh recognizable languages is precisely the class of languages definable in weak existential monadic second-order logic of one successor extended with linear cardinality constraints. 
The class of Parikh recognizable languages contains all regular languages, but also many more, even languages that are not context-free, \eg, the language $\{a^nb^nc^n \mid n \in \mathbb{N}\}$. On the other hand, the language of palindromes is context-free, but not Parikh recognizable.
On finite words, blind multicounter automata, Parikh automata and related models have been investigated extensively, extending~\cite{greibach,klaedtkeruess} for example by affine~PA and~PA on letters \cite{cadilhac2011expressiveness,DBLP:journals/ita/CadilhacFM12},
bounded~PA~\cite{DBLP:journals/ijfcs/CadilhacFM12}, two-way~PA~\cite{FiliotGM19},~PA with a pushdown stack \cite{karianto2004parikh} as well as a combination of both \cite{DBLP:conf/fossacs/DartoisFT19}, history-deterministic~PA~\cite{erlich2022history}, automata and grammars with valences~\cite{valence,valenceRevisited}, and several algorithmic applications, \eg, in the context of path logics for querying graphs~\cite{emptynp}.

In the well-studied realm of verification of reactive systems, automata-related approaches provide a powerful framework to tackle important problems such as model checking problems \cite{Baier2008,Clarke,ClarkeHandbook}.
However, computations of such systems are generally represented as infinite objects, as we often expect them to not terminate (but rather interact with the environment).
Hence, automata processing infinite words are well-suited to approach the model checking problem. One common approach is the following: assume we are given a system, \eg, represented as a Kripke structure $\Kmc$, and a specification represented as an automaton~$\Amc$ (or any formalism that can be translated into one) accepting all correct computations.
Then we can verify the correctness of the system by solving the inclusion problem, that is, answering the question whether every computation of $\Kmc$ is also a computation of~$\Amc$.
However, solving the inclusion problem generally requires to complement the automaton~$\Amc$, which is often expensive or not even possible. Hence, another common counterexample-driven approach is the following.
Let~$\Amc$ be an automaton accepting all incorrect computations.
Then we can verify that the system has no incorrect computations by solving the intersection-emptiness problem of~$\Kmc$ and~$\Amc$, that is, answering the question whether their sets of computations are disjoint.
Büchi automata, recognizing $\omega$-regular languages, provide a natural extension of classical (finite word) automata to infinite words, enjoying closure under the Boolean operations and decidable decision problems. However, being one of the most basic models, their expressiveness is quite limited.

Let us consider two examples. In a three-user setting in an operating system we would like to ensure that none of the users gets a lot more resources than the other two. 
One way to represent a such a specification is considering the $\omega$-language 
\[\{\alpha \in \{a,b,c\}^\omega \mid \text{ there are infinitely many prefixes $w$ of $\alpha$ with $|w|_a = |w|_b = |w|_c$}\},\] 
stating that there are infinitely many moments where the resources are distributed equally.
Similarly, one could provide a set of unwanted computations via the $\omega$-language
\[\{\alpha \in \{a,b,c\}^\omega \mid \text{ there are infinitely many prefixes $w$ of $\alpha$ with $|w|_a > |w|_b + |w|_c$}\},\]
stating that one user gets more resources than the other two users combined infinitely often. 
As another example, consider a classical producer-consumer setting, where a producer continuously produces a good, and a consumer consumes these goods continuously. 
We can model this setting as an infinite word and ask that at no time the consumer has consumed more than the producer has produced at this time. 
Bad computations can be modeled via the $\omega$-language 
\[\{\alpha \in \{p,c\}^\omega \mid \text{there is a prefix $w$ of $\alpha$ with $|w|_c > |w|_p$}\}.\] 
Such specifications are not $\omega$-regular, as these require to ``count arbitrarily''.
This motivates the study of finite automata with counting mechanisms on infinite words.
Fernau and Stiebe initiated this study and introduced blind counter automata on infinite words~\cite{blindcounter}, extending Greibachs blind multicounter machines.
Independently, Klaedtke and Ruess proposed possible extensions of Parikh automata to infinite words. This line of research was recently picked up by Guha et al.~\cite{infiniteZimmermann}. 

Guha et al.~\cite{infiniteZimmermann} introduced \emph{safety, reachability, Büchi- and co-Büchi Parikh automata}.
These models provide natural generalizations of automata models with Parikh conditions on infinite words. 
One shortcoming of safety, reachability and co-Büchi Parikh automata is that they do not generalize Büchi automata, that is, they cannot recognize all $\omega$-regular languages. 
The non-emptiness problem, which is highly relevant for model checking applications, is undecidable for safety and co-Büchi Parikh automata. Furthermore, none of these models is closed under the $\omega$-operation, meaning that for every model there is a Parikh recognizable (finite word) language $L$ such that $L^\omega$ is not recognizable by any of these models. 
Guha et al.\ raised the question whether (appropriate variants of) Parikh automata on infinite words have the same expressiveness as blind counter automata on infinite words. 

Büchi's famous theorem states that $\omega$-regular languages are characterized as languages of the form $\bigcup_i U_i V_i^\omega$, where the $U_i$ and $V_i$ are regular languages~\cite{buechi}.  
As a consequence of the theorem, many properties of $\omega$-regular languages are inherited from regular languages. 
For example, algorithms deciding non-emptiness for finite word automata can be used to solve the non-emptiness problem for Büchi automata as well.
In their systematic study of blind counter automata operating on infinite words, Fernau and Stiebe~\cite{blindcounter} considered the class $\Kmc_*$, the class of $\omega$-languages of the form $\bigcup_i U_i V_i^\omega$ for Parikh recognizable languages $U_i$ and $V_i$.
They proved that the class of $\omega$-languages recognizable by blind counter automata is a proper subset of the class~$\Kmc_*$.
They posed as an open problem to provide automata models that capture classes of $\omega$-languages of the form $\bigcup_i U_i V_i^\omega$ where~$U_i$ and~$V_i$ are described by a certain mechanism. 
In this work we propose \emph{reachability-regular Parikh automata}, \emph{limit Parikh automata}, \emph{strong reset Parikh automata}, and \emph{weak reset Parikh automata} as new automata models and study their deterministic and non-deterministic variants. First, we focus on their non-deterministic variants.

We pick up the question of Fernau and Stiebe~\cite{blindcounter} to consider classes of $\omega$-languages of the form $\bigcup_i U_i V_i^\omega$ where $U_i$ and $V_i$ are described by a certain mechanism. 
We define the four classes $\LRegReg$, $\LPAReg$, $\LRegPA$ and $\LPAPA$ of $\omega$-languages of the form $\bigcup_iU_iV_i^\omega$, where the $U_i,V_i$ are regular or Parikh recognizable languages of finite words, respectively. By Büchi's theorem the class $\LRegReg$ is the class of $\omega$-regular languages. 

We show that the newly introduced (non-deterministic) reachability-regular Parikh automata, which are a small modification of reachability Parikh automata (as introduced by Guha et al.~\cite{infiniteZimmermann}) capture exactly the class $\LPAReg$. 
This model turns out to be equivalent to (non-deterministic) limit Parikh automata. This model was hinted at in the concluding remarks of~\cite{klaedtkeruess}. 

Fully resolving the classification of the above mentioned classes we introduce 
weak and strong reset Parikh automata, whose non-deterministic variants turn out to be equivalent. In contrast to all other Parikh models, these are closed under the $\omega$-operation, while maintaining all algorithmic properties of~PA (in particular, non-emptiness is $\NP$-complete and hence decidable). 
We show that the class of $\omega$-languages recognized by reset~PA is a strict superclass of $\LPAPA$. 
We show that appropriate graph-theoretic restrictions of reset~PA exactly capture the classes $\LPAPA$ and $\LRegPA$, yielding the first automata characterizations for these classes. On our way, we study the closure properties of the newly introduced models and obtain the following hierarchy.
\begin{align*}
	\text{reachability~PA } & \subsetneq \text{ reachability-regular~PA } = \text{ limit~PA} 
	\\ 
	& \subsetneq \text{ Büchi~PA } \subsetneq  
	\text{ weak reset~PA } = \text{ strong reset~PA.}
\end{align*}

The automata models introduced by Guha et al.~\cite{infiniteZimmermann} do not have $\epsilon$-transitions, while blind counter automata as introduced by Fernau and Stiebe~\cite{blindcounter} have such transitions. 
Towards answering the question of Guha et al.\ we study the effect of $\epsilon$-transitions in all Parikh automata models. 
We show that all models except safety and co-Büchi Parikh automata admit $\epsilon$-elimination. This in particular answers the question of Guha et al.~\cite{infiniteZimmermann} whether blind counter automata and Büchi Parikh automata have the same expressiveness on infinite words affirmative. 
We show that safety and co-Büchi~PA with $\epsilon$-transitions are strictly more powerful than their variants without $\epsilon$-transitions. In particular, co-Büchi~PA with $\varepsilon$-transitions generalize Büchi~PA, while safety~PA with $\varepsilon$-transitions are even more powerful than reset~PA.

Finally, we study the classical decision problems for the newly introduced models, namely emptiness, membership and universality. We show that the results for (non-deterministic)~PA on finite words translate to the infinite word setting, that is, emptiness is $\coNP$-complete, membership is $\NP$-complete and universality is undecidable.
Additionally, we study
the intersection-emptiness and inclusion problems of the newly introduced models, as these are important problems in order to study model checking problems, the core problems in the field of formal verification. 
Formally, we are given a system~$\Kmc$ as a Kripke structure (a safety automaton) or a~PA and a specification $\Amc$ as a~PA. 
The question whether at least one computation of a Kripke structure satisfies the specification (which we call existential safety model checking and boils down to solving intersection-emptiness) is motivated by the testing the absence of incorrect computations, as described above.
%We show that this problem remains $\coNP$-complete for our models.
Similarly, the question whether all computations of a~PA satisfy the specification (which we call universal~PA model checking and boils down to solving inclusion) has been studied in~\cite{infiniteZimmermann} for reachability~PA, Büchi~PA, safety~PA and co-Büchi~PA.
Guha et~al.~\cite{infiniteZimmermann} show that this problem is undecidable for the non-deterministic variants of these models.
We study this problem as well as the universal safety model checking problem and the existential~PA model checking problem for the remaining deterministic models. 
%Again, deterministic limit~PA shine as all these problems are decidable for them. For the other models, some problems are decidable and some are not.
We refer to \Cref{tab:closure}, \Cref{tab:decision}, and \Cref{tab:mc} for an overview of the results.

As indicated above, there are many scenarios where non-determinism adds expressiveness or succinctness to their deterministic counterparts. However, this often comes at the price of important decision problems becoming hard to solve, or even undecidable, see also~\cite{ClarkeHandbook,thomas2002automata}.
This is also the case for Parikh automata, as witnessed, \eg, by co-Büchi~PA and safety~PA.
As mentioned above, Guha et al.~\cite{infiniteZimmermann} have shown that universality is undecidable for the non-deterministic variants of these models, yet being more powerful than their deterministic counterparts. However, the deterministic variants enjoy a $\coNP$-complete (and hence decidable) universality problem~\cite{infiniteZimmermann}.

This motivates the study of the deterministic variants of the newly introduced models, namely deterministic limit~PA, deterministic reachability-regular~PA, deterministic strong reset~PA, and deterministic weak reset~PA.
We investigate their expressiveness, closure properties, and common decision problems.
First we show that the above mentioned hierarchy results for the non-deterministic variants do not translate to the deterministic setting. While 

\vspace{-3mm}
\[\text{deterministic strong reset~PA $\subsetneq$ deterministic weak reset~PA}\] 
and 
\begin{align*}
	\text{
		deterministic} & \text{ reachability~PA} \\ &\subsetneq \text{ deterministic reachability-regular~PA} \\ & \subsetneq \text{deterministic weak reset~PA} 
\end{align*}
still holds, all other models become pairwise incomparable.
Furthermore, we show that among all studied deterministic models only deter\-ministic limit~PA generalize Büchi automata in the sense that they recognize all $\omega$-regular languages.

Deterministic limit~PA also shine in the light of closure properties: as we show, among all studied~PA operating on infinite words (the deterministic variants as well as the non-deterministic ones) they are the only ones closed under all Boolean operations.

\begin{table}
	\begin{center}
		\begin{tabular}{l|ccc}
			& $\cup$ & $\cap$ & $\overline{\phantom{a}}$ \\[0.5mm]
			\hline\\[-2.5mm]
			limit~PA & \YES & \YES & \NO \\[1mm]
			reachability-regular~PA & \YES & \YES & \NO \\[1mm]
			weak reset~PA & \YES & \NO & \NO \\[1mm]
			strong reset~PA & \YES& \NO & \NO \\[2mm]\hline\\[-2.5mm]
			
			deterministic limit~PA& \YES & \YES & \YES \\[1mm]
			deterministic reachability-regular~PA & \NO & \NO & \NO \\[1mm]
			deterministic weak reset~PA& \NO & \NO & \NO \\[1mm]
			deterministic strong reset~PA & \NO & \NO & \NO \\[1mm]
		\end{tabular}
		
		\smallskip
		\caption{\centering Closure properties. The bar in the right column denotes the complement.}
		\label{tab:closure}
	\end{center}
\end{table}

This benefit also yields decidable decision problems. In contrast to the all other models that were mentioned before, emptiness and universality are both decidable for deterministic limit~PA. We show that also strong reset~PA benefit from determinism: although having bad closure properties, their universality problems becomes decidable. However, as we show, for deterministic reachability-regular~PA and deterministic weak reset~PA the universality problem remains undecidable.

\begin{table}
	\begin{center}
		\begin{tabular}{l|ccc}
			& $L = \varnothing?$ & $uv^\omega \in L?$ & $L = \Sigma^\omega?$ \\[0.5mm] \hline\\[-2.5mm]
			
			limit~PA & $\coNP$-complete & $\NP$-complete & undecidable\\[1mm]
			reachability-regular~PA &$\coNP$-complete &\NP-complete & undecidable \\[1mm]
			weak reset~PA &$\coNP$-complete  &$\NP$-complete & undecidable \\[1mm]
			strong reset~PA &$\coNP$-complete & $\NP$-complete & undecidable \\[2mm]\hline\\[-2.5mm]
			
			det.\ limit~PA & \coNP-complete & \NP-complete & decidable, $\Pi_2^\P$-hard  \\[1mm]
			det.\ reachability-regular~PA & \coNP-complete & \NP-complete & undecidable  \\[1mm]
			det.\ weak reset~PA & \coNP-complete & \NP-complete & undecidable \\[1mm]
			det.\ strong reset~PA & \coNP-complete & \NP-complete & $\Pi_2^\P$-complete \\[1mm]\end{tabular}
		
		\smallskip
		\caption{\centering Decision problems.}
		\label{tab:decision}
	\end{center}
\end{table}

While universality and hence the universal model checking problems are undecidable for all non-deterministic variants of the aforementioned models, the situation changes for deterministic limit~PA and deterministic strong reset~PA. Again, we refer to \Cref{tab:closure}, \Cref{tab:decision}, and \Cref{tab:mc} for an overview of the results. Recall that $\Pi_2^\P$ denotes the second universal level of the polynomial hierarchy.

\begin{table}
	\begin{center}
		\begin{tabular}{l|llll}
			& \multicolumn{4}{c}{Model Checking}  \\ 
			& \multicolumn{2}{c|}{Kripke} & \multicolumn{2}{c}{PA} \\ 
			& \multicolumn{1}{c|}{$\exists$} & \multicolumn{1}{c|}{$\forall$} & \multicolumn{1}{c|}{$\exists$} & \multicolumn{1}{c}{$\forall$}  \\ \hline \\[-2.5mm]
			
			limit~PA & $\coNP$-c.\  &undec. &$\coNP$-c.&undec.\\[1mm]
			reachability-regular~PA &$\coNP$-c.\  &undec.&$\coNP$-c.&undec. \\[1mm]
			weak reset~PA &$\coNP$-c.\  &undec.&undec.&undec. \\[1mm]
			strong reset~PA &$\coNP$-c.\ & undec. & undec. & undec.\\[1mm]			
			reachability~PA &$\coNP$-c. $(\square)$\  & undec.\ $(\square)$& $\coNP$-c.\ $(\square)$ &undec.\ $(\square)$\\[1mm]
			Büchi~PA &$\coNP$-c.\ & undec.\ $(\square)$& $\coNP$-c.\ &undec.\ $(\square)$ \\[1mm]
			safety~PA &undec.\ $(\square)$& undec.\ $(\square)$& undec.\ $(\square)$ & undec.\ $(\square)$\\[1mm]
			co-Büchi~PA &undec.\ $(\square)$& undec.\ $(\square)$& undec.\ $(\square)$ &undec. $(\square)$ \\[2mm]\hline\\[-2.5mm]
			
			det.\  limit~PA & $\coNP$-c.\  &dec., $\Pi_2^\P$-hard &$\coNP$-c.&dec., $\Pi_2^\P$-hard\\[1mm]
			det.\  reachability-regular~PA &$\coNP$-c.\  &undec.&$\coNP$-c.&undec. \\[1mm]
			det.\  weak reset~PA &$\coNP$-c.\  &undec.&undec.&undec. \\[1mm]
			det.\  strong reset~PA &$\coNP$-c.\ & $\Pi_2^\P$-c.& undec. &$\Pi_2^\P$-c.\\[1mm]			
			det.\  reachability~PA &$\coNP$-c.\ $(\square)$ & undec.\ $(\square)$& $\coNP$-c.\ $(\square)$ &undec.\ $(\square)$\\[1mm]
			det.\  Büchi~PA &$\coNP$-c.\ & undec.\ $(\square)$& $\coNP$-c.\ &undec.\ $(\square)$ \\[1mm]
			det.\  safety~PA &undec.\ $(\square)$& $\coNP$-c.\ $(\square)$& undec.\ $(\square)$ &\coNP-c.\ $(\square)$\\[1mm]
			det.\  co-Büchi~PA &undec.\ $(\square)$& $\coNP$-c.\ $(\square)$& undec.\ $(\square)$ &\coNP-c.\ $(\square)$
		\end{tabular}
		
		\smallskip
		\caption{(Un)decidability results of model checking problems. The entries marked with a~$(\square)$ were shown in (or follow immediately from)~\cite{infiniteZimmermann}}
		\label{tab:mc}
	\end{center}
\end{table}

%\vspace{-6mm}

We remark that the $\Pi_2^\P$-hardness results we obtain for deterministic limit~PA are not tight, that is, it remains open whether these problems can be solved in $\Pi_2^\P$. However, we are able to show that if we have the guarantee that the semi-linear sets of the~PA can be complemented in polynomial time, then the $\Pi_2^\P$-hard problems for deterministic limit~PA and deterministic strong reset~PA become $\coNP$-complete. 
We also remark that the \coNP-completeness result for deterministic Büchi~PA does not follow from~\cite{infiniteZimmermann}.

%% file: 2_auto_prelims.tex
\section{Preliminaries}
\label{secauto:prelims}

In this section we recall the definitions and relevant notations.
We write $\Zbb$ for the set of all integers and $\Nbb$ for the set of non-negative integers including~$0$. 
For $m, n \in \Nbb$, we denote by $[m,n]$ the set $\{m, m+1, \dots, n\}$ with the convention that $[m, n] = \varnothing$ if $m > n$. We abbreviate the set $[1,n]$ by $[n]$.
Furthermore, let $\Nbbinfty = \Nbb \cup \{\infty\}$. Throughout this paper we mainly use the variables $c,d,i,j,k,\ell,m,n,z$ to denote positive integers and we will tacitly assume this if not explicitly stated otherwise.

\subsection{Words and Languages}
Let~$\Sigma$ be an alphabet, \ie, a finite non-empty set and denote by~$\Sigma^*$ the set of all finite words over~$\Sigma$. A language is a subset $L \subseteq \Sigma^*$.
For a word $w \in \Sigma^*$, we denote by $|w|$ the length of~$w$, and by $|w|_a$ the number of occurrences of the symbol $a \in \Sigma$ in~$w$. 
We write $\varepsilon$ for the empty word of length~$0$ and denote by $\Sigma^+ = \Sigma^* \setminus \{\varepsilon\}$ the set of all non-empty words over $\Sigma$.
%We say a word $v \in \Sigma^*$ is a \emph{subword} of a word $w \in \Sigma^*$, denoted by $v \subleq w$, if $v$ can be obtained from~$w$ by removing symbols at arbitrary positions of~$w$. More formally, $v = v_1 \dots v_m$ is a subword of $w = w_1 \dots w_n$ if there is a function $f : [m] \rightarrow [n]$ such that $f(i) < f(j)$ if $i < j$ and $v_i = w_{f(i)}$ for all $i \leq m$. 
We call $v$ \emph{prefix} of $w = w_1 \dots w_n$ if $v = w_1 \dots w_i$ for some $i \leq n$; \emph{suffix} of $w$ if $v = w_i \dots w_n$ for some $i \leq n+1$; and \emph{infix} of $w$ if $v = w_i \dots w_j$ for $i,j \leq n+1$. Note that $\varepsilon$ is a prefix, suffix, and infix of every word $w \in \Sigma^*$.

An \emph{infinite word} over an alphabet $\Sigma$ is a function $\alpha : \Nbb \setminus \{0\} \rightarrow \Sigma$. We often write~$\alpha_i$ instead of~$\alpha(i)$. 
Thus, we can understand an infinite word as an infinite sequence of symbols $\alpha = \alpha_1\alpha_2\alpha_3\ldots$. Hence, the notions of prefix, infix and suffix translate to infinite words the obvious way; however, we only consider finite prefixes and infixes of infinite words, while every suffix is always an infinite word itself. For $m \leq n$, we abbreviate the infix $\alpha_m \ldots \alpha_n$ of~$\alpha$ by $\alpha[m:n]$. 
We denote by $\Sigma^\omega$ the set of all infinite words over $\Sigma$. 
We call a subset $L \subseteq \Sigma^\omega$ an \emph{$\omega$-language}. 
Moreover, for $L \subseteq \Sigma^*$, we define $L^\omega = \{w_1w_2\dots \mid w_i \in L \setminus \{\varepsilon\}\} \subseteq \Sigma^\omega$. 
We call an infinite word $\alpha \subseteq \Sigma^\omega$ \emph{ultimately periodic} if it can be written as $\alpha = uv^\omega$ for finite words $u,v \in \Sigma^*$. Likewise, we call a non-empty $\omega$-language \emph{ultimately periodic} if it contains an ultimately periodic infinite word.
%We call a non-empty $\omega$-language $L \subseteq \Sigma^\omega$ \emph{ultimately periodic} if it contains an infinite word of the form $uv^\omega$ for finite words $u,v \in \Sigma^*$. We refer to such infinite words also as \emph{ultimately periodic}.

In addition to the convention of variable use mentioned above, we mainly use the variables~$a$ and~$b$ to denote symbols,~$w$ to denote finite words and~$\alpha, \beta$ to denote infinite words and we will tacitly assume this if not explicitly stated otherwise.
%\nopagebreak
\subsection{\boldmath Regular and \texorpdfstring{$\omega$}{omega}-regular Languages}

A \emph{non-deterministic finite automaton} (NFA) is a tuple $\Amc = (Q, \Sigma, q_0, \Delta, F)$, where~$Q$ is a finite set of states, $\Sigma$ is the input alphabet, $q_0 \in Q$ is the initial state, $\Delta \subseteq Q \times \Sigma \times Q$ is the set of transitions and $F \subseteq Q$ is the set of accepting states. 
We call~$\Amc$ \emph{deterministic} if for every pair $(p, a) \in Q \times \Sigma$ there is exactly one transition of the form $(p, a, q) \in \Delta$ for some $q \in Q$.
A \emph{run} of $\Amc$ on a word $w = w_1 \ldots w_n\in \Sigma^*$ is a (possibly empty) sequence of transitions $r = r_1 \ldots r_n$ with $r_i = (p_{i-1}, w_i, p_i)\in \Delta$ such that~$p_0=q_0$. 
We say $r$ is \emph{accepting} if $p_n \in F$. The empty run on~$\epsilon$ is accepting if $q_0 \in F$. We define the \emph{language recognized by \Amc} as ~
\[L(\Amc) = \{w \in \Sigma^* \mid \text{there is an accepting run of $\Amc$ on $w$}\}.\] If a language~$L$ is recognized by some NFA $\Amc$, we call~$L$ \emph{regular}.

A \emph{Büchi automaton} is an NFA $\Amc = (Q, \Sigma, q_0, \Delta, F)$ that takes \mbox{infinite} words as input. 
A \emph{run} of $\Amc$ on an infinite word $\alpha_1\alpha_2\alpha_3\dots$ is an infinite sequence of transitions $r = r_1 r_2 r_3 \dots$ with $r_i = (p_{i-1}, \alpha_i, p_i) \in \Delta$ such that $p_0=q_0$. 
We say $r$ is \emph{accepting} if there are infinitely many~$i$ with $p_i \in F$. 
We define the \emph{$\omega$-language recognized by~$\Amc$} as $L_\omega(\Amc) = \{\alpha \in \Sigma^\omega \mid \text{there is an accepting run of $\Amc$ on $\alpha$}\}$. 
If an $\omega$-language $L$ is recognized by some Büchi automaton $\Amc$, we call $L$ \emph{$\omega$-regular}.

Büchi's theorem establishes an important connection between regular and $\omega$-regular languages:
\begin{thm}[Büchi \cite{buechi}]
	\label{thm:buechi}
	A language $L \subseteq \Sigma^\omega$ is $\omega$-regular if and only if there are regular languages $U_1, V_1, \dots, U_n, V_n \subseteq \Sigma^*$ for some $n \geq 1$ such that $L = \bigcup_{i \leq n} U_i V_i^\omega$.
\end{thm}

We will often denote the class of $\omega$-regular languages by $\LRegReg$.

If every state of a Büchi automaton~$\Amc$ is accepting, we call $\Amc$ a \emph{safety automaton}.
Similarly, a \emph{Muller automaton} is a tuple $\Amc = (Q, \Sigma, q_0, \Delta, \Fmc)$, where $Q, \Sigma, q_0$, and $\Delta$ are defined as for Büchi automata, and $\Fmc \subseteq 2^Q$ is a collection of sets of accepting states. Runs are defined as for Büchi automata, and a run~$r$ is accepting if the sets of states that appear infinitely often in~$r$ is contained in~$\Fmc$. 
Deterministic Muller automata have the same expressiveness as non-deterministic Büchi automata~\cite{mcnaughton1966testing}. However, deterministic Büchi automata are less expressive than their non-deterministic counterpart~\cite{landweber}. 
If an $\omega$-language $L$ is recognized by some deterministic Büchi automaton, we call $L$ \emph{deterministic $\omega$-regular}.
For a (finite word) language $W \subseteq \Sigma^*$, we define \mbox{$\overline{W} = \{\alpha \in \Sigma^\omega \mid \alpha[1:i] \in W \text{ for infinitely many } i\}$}. An $\omega$-language~$L$ is deterministic $\omega$-regular if and only if $L = \overline{W}$ for a regular language~$W$~\cite{landweber}.

\subsection{Semi-linear Sets and Presburger Arithmetic}
%For some $d \geq 1$, let $\bbf \in \Nbb^d$ and $P = \{\pbf_1, \dots, \pbf_\ell\} \subseteq \Nbb^d$ be a finite set for some $\ell \geq 0$. A \emph{linear set} of dimension $d$ is a set of the form $C(\bbf, P) = \{\bbf + \pbf_1z_1 + \dots + \pbf_\ell z_\ell \mid z_1, \dots, z_\ell \in \Nbb\} \subseteq \Nbb^d$; we call $\bbf$ the \emph{base vector} and $P$ the set of \emph{period vectors}.
A \emph{linear set} (of dimension $d\geq 1$) is a set of the form 
\[C(\bbf, P) = \{\bbf + \pbf_1z_1 + \dots + \pbf_\ell z_\ell \mid z_1, \dots, z_\ell \in \Nbb\} \subseteq \Nbb^d,\] 
where $\bbf \in \Nbb^d$ and $P = \{\pbf_1, \dots, \pbf_\ell\} \subseteq \Nbb^d$ is a finite set of vectors for some $\ell \geq 0$. We call~$\bbf$ the \emph{base vector} and $P$ the set of \emph{period vectors}. We call linear sets of the form $C(\0, P)$ \emph{homogeneous}.
A \emph{semi-linear set} is a finite union of linear sets. 
We also consider semi-linear sets over~$\Nbbinfty^d$, that is, semi-linear sets with an additional symbol~$\infty$ for infinity. As usual, addition of vectors and multiplication of a vector with a number is defined component-wise, where \mbox{$z + \infty = \infty + z = \infty + \infty =\infty$} for all $z \in \Nbb$, $z \cdot \infty = \infty \cdot z = \infty$ for all $z \geq 1$, and \mbox{$0 \cdot \infty = \infty \cdot 0 = 0$}.
For vectors $\ubf = (u_1, \dots, u_c)\in \Nbbinfty^c$ and $\vbf = (v_1, \dots, v_d) \in \Nbbinfty^d$, we denote by $\ubf \cdot \vbf = (u_1, \dots, u_c, v_1, \dots, v_d) \in \Nbbinfty^{c+d}$ the \emph{concatenation of~$\ubf$ and $\vbf$}. 
We extend this definition to sets of vectors. Let $C \subseteq \Nbbinfty^c$ and $D \subseteq \Nbbinfty^d$. Then $C \cdot D = \{\ubf \cdot \vbf \mid \ubf \in C, \vbf \in D\}$ $\subseteq \Nbbinfty^{c+d}$.
We denote by~$\0^d$ the $d$-dimensional all-zero vector, by $\1^d$ the $d$-dimensional all-one-vector, by~$\ebf^d_i$ the $d$-dimensional vector where the $i$th entry is~$1$ and all other entries are~0, and by~$\ibf^d_i$ the $d$-dimensional vector where the $i$th entry is~$\infty$ and all other entries are~0. We often drop the superscript $d$ if the dimension is clear. For all complexity theoretical results, we assume a binary and explicit encoding of semi-linear sets unless stated otherwise.

\begin{figure}
\centering
\begin{tikzpicture}
\draw[step=1,gray,very thin] (0,0) grid (5.5,5.5);
\draw[very thick, ->] (0,0) -- (5.5,0) node[below right]{$v_1$};
\draw[very thick, ->] (0,0) -- (0,5.5) node[above left] {$v_2$};

\node at(-.3, -.3) {$0$};

\foreach \x in {1,...,5}
	\node at(\x, -.3) {$\x$};
	
\foreach \y in {1,...,5}
	\node at(-.3, \y) {$\y$};
	
\foreach \x in {0,...,4} {
	%\draw[orange] (\x - 0.2, \x - 0.2) -- (\x + 0.2, \x + 0.2);
	%\draw[orange] (\x - 0.2, \x + 0.2) -- (\x + 0.2, \x - 0.2); 
    \draw[orange,-{Latex}] (\x, \x) -- (\x + 1, \x + 1);   
}  

%\draw[orange] (4.8, 4.8) -- (5.2, 5.2);
%\draw[orange] (5.2, 4.8) -- (4.8, 5.2);

\draw[blue,-{Latex}] (0,0) -- (1,2);
\draw[blue,-{Latex}] (1,2) -- (3,3);
\draw[blue,-{Latex}] (3,3) -- (5,4);
\draw[blue,-{Latex}] (1,3) -- (3,4);
\draw[blue,-{Latex}] (3,4) -- (5,5);
\draw[blue,-{Latex}] (1,4) -- (3,5);

\draw[blue,-{Latex}] (1,2) -- (1,3);
\draw[blue,-{Latex}] (1,3) -- (1,4);
\draw[blue,-{Latex}] (1,4) -- (1,5);
\draw[blue,-{Latex}] (3,3) -- (3,4);
\draw[blue,-{Latex}] (3,4) -- (3,5);
\draw[blue,-{Latex}] (5,4) -- (5,5);

\draw[orange, dashed] (5,5) -- (5.5, 5.5);
\draw[blue, dashed] (1,5) -- (1, 5.5);
\draw[blue, dashed] (1,5) -- (1.5, 5.25);
\draw[blue, dashed] (3,5) -- (3, 5.5);
\draw[blue, dashed] (5,5) -- (5, 5.5);
\draw[blue, dashed] (3,5) -- (3.5, 5.25);
\draw[blue, dashed] (5,4) -- (5.5, 4.25);
\draw[blue, dashed] (5,5) -- (5.5, 5.25);
    
\end{tikzpicture}
\caption{The semi-linear set $C = C((0,0), \{(1,1)\}) \cup C((1,2), \{(0,1), (2,1)\})$.}
\label{fig:semi}
\end{figure}

Semi-linear sets coincide with sets definable in Presburger arithmetic, that is, the first-order theory of (non-negative) integers and order, in the following sense:
A Presburger formula $\varphi(\vbf)$ with $d$ free variables defines the set $\{\vbf \in \Nbb^d \mid (\Nbb, 0, 1, +, <) \models \varphi(\vbf)\}$. By a groundbreaking result of Ginsburg and Spanier~\cite{semilin}, a set $C \subseteq \Nbb^d$ is semi-linear if and only if it is definable in Presburger arithmetic. Observe that this result implies that semi-linear sets are closed under the Boolean operations. We refrain from formally introducing first-order logic and refer to the textbooks~\cite{chang2013model,libkinElements} for an in-depth introduction.
Instead, we outline the connection by the following example. Consider the linear sets $C_1 = C((0,0), \{(1,1)\})$ and $C_2 = C((1,2), \{(0,1), (2,1)\})$, and let $C = C_1 \cup C_2$, see \Cref{fig:semi} for an illustration.
We can define $C_1$ and $C_2$ by the Presburger formulas
\[\varphi_1(x_1, x_2) = \exists z. (x_1 = z \land x_2 = z)\]
and
\[\varphi_2(x_1, x_2) = \exists z_1. \exists z_2. (x_1 = 1 + 2z_2 \land x_2 = 2 + z_1 + z_2),\] 
respectively (where $2$ is an abbreviation for $1+1$ and $2z_2$ is an abbreviation for $z_2 + z_2$). 
Hence, for each period vector we guess the number of iterations using existential quantifiers. Consequently, we can define $C$ by 
\[\varphi(x_1, x_2) = \varphi_1(x_1, x_2) \lor \varphi_2(x_1, x_2).\]
The translation of Presburger formulas into semi-linear sets involves way more sophisticated methods, including a procedure to eliminate quantifiers; we refer the interested reader to~\cite{haaseTaming,haase}.

\subsection{Parikh Recognizable Languages}
A \emph{Parikh automaton} (PA) is a tuple $\Amc = (Q, \Sigma, q_0, \Delta, F, C)$ where $Q$, $\Sigma$, $q_0$, and~$F$ are defined as for NFA, $\Delta \subseteq Q \times \Sigma \times \Nbb^d \times Q$, for some \mbox{$d\geq 1$}, is a finite set of \emph{labeled transitions}, and $C \subseteq \Nbb^d$ is a semi-linear set. 
We call $d$ the \emph{dimension} of $\Amc$ and interpret~$d$ as a number of \emph{counters}.
%refer to the entries of a vector $\vbf$ in a transition $(p, a, \vbf, q) \in \Delta$ as \emph{counters}.
Analogously to NFA, we call~$\Amc$ \emph{deterministic} if for every pair $(p,a) \in Q \times \Sigma$ there is exactly one labeled transition of the form \mbox{$(p,a,\vbf,q) \in \Delta$} for some $\vbf \in \Nbb^d$ and $q \in Q$.
A \emph{run} of $\Amc$ on a word $w = w_1 \dots w_n$ is a (possibly empty) sequence of labeled transitions $r = r_1 \dots r_n$ with $r_i = (p_{i-1}, w_i, \vbf_i, p_i) \in \Delta$ such that $p_0 = q_0$. We define the \emph{extended Parikh image} of a run $r$ as $\rho(r) = \sum_{i \leq n} \vbf_i$ (with the convention that the empty sum equals~$\0$).
We say $r$ is accepting if $p_n \in F$ and \mbox{$\rho(r) \in C$},
referring to the latter condition as the \emph{Parikh condition}. 
The \emph{language recognized by~\Amc} is 
\[L(\Amc) = \{w \in \Sigma^*\mid \text{there is an accepting run of $\Amc$ on $w$}\}.\]
If a language $L\subseteq \Sigma^*$ is recognized by some PA, then we call $L$ \emph{Parikh recognizable}.

%\subsection{Groups}
%A \emph{group} is a tuple $\Gmc = (G, \cdot, e)$ where $G$ is a non-empty set and $\cdot : G \times G \rightarrow G$ is a binary operation with neutral element $e$ satisfying the following properties.
%\begin{description}
%\item[Associativity.] For all $x,y,z\in G$ we have $(x \cdot y) \cdot z = x \cdot (y \cdot z)$.
%\item[Neutral element.] For all $x \in G$ we have $x \cdot e = e \cdot x = x$.
%\item[Inverse element.] For all $x \in G$ there is a unique element $x^{-1} \in G$ with $x \cdot x^{-1} = e$.
%\end{description}
%We call a group $\Gmc$ \emph{Abelian}, if it further satisfies the following property.
%\begin{description}
%\item[Commutativity.] For all $x,y \in G$ we have $x \cdot y = y \cdot x$.
%\end{description}

%Throughout the part we mainly deal with the Abelian group $(\Zbb^d, +, \0)$ for some $d \geq 1$.

\subsection{Directed Graphs}
A \emph{(directed) graph} $G$ consists of its vertex set $V(G)$ and edge set \mbox{$E(G) \subseteq V(G) \times V(G)$}. In particular, a graph $G$ may have loops, that is, edges of the form $(u, u)$. A \emph{(simple) path} from a vertex $u$ to a vertex $v$ in $G$ is a sequence of pairwise distinct vertices $v_1 \dots v_k$ such that $v_1 = u$, $v_k = v$, and $(v_i, v_{i+1}) \in E(G)$ for all $1 \leq i < k$. Similarly, a \emph{(simple) cycle} in $G$ is a sequence of pairwise distinct vertices $v_1 \dots v_k$ such that $(v_i, v_{i+1}) \in E(G)$ for all $1 \leq i < k$, and $(v_k, v_1) \in E(G)$. If $G$ has no cycles, we call $G$ a directed acyclic graph (DAG).
For a subset $U \subseteq V(G)$, we denote by $G[U]$ the graph $G$ \emph{induced by} $U$, \ie, the graph with vertex set~$U$ and edge set $\{(u,v) \in E(G) \mid u, v \in U\}$.
A \emph{strongly connected component} (SCC) in $G$ is a maximal subset $U \subseteq V(G)$ such that for all $u, v \in U$ there is a path from~$u$ to $v$, \ie, all vertices in $U$ are reachable from each other. 
We write $SCC(G)$ for the set of all strongly connected components of $G$ (observe that $SCC(G)$ partitions~$V(G)$). 
The \emph{condensation} of~$G$, written $C(G)$, is the DAG obtained from $G$ by contracting each SCC of $G$ into a single vertex, that is $V(C(G)) = SCC(G)$ and $(U, V) \in E(C(G))$ if and only if there is $u \in U$ and $v \in V$ with $(u, v) \in E(G)$. 
We call the SCCs with no outgoing edges in $C(G)$ leaves.
Note that an automaton can be seen as a labeled graph. Hence, all definitions translate to automata by considering the underlying graph (to be precise, an automaton can be seen as a labeled multigraph; however, we simply drop parallel edges).

%% file: 3_auto_finite.tex
\section{Decision Problems for Parikh Automata on Finite Words}
In this section, we consider the following common decision problems for (deterministic)~Parikh automata.

\begin{itemize}
	\item Emptiness: given a~PA $\Amc$, is $L(\Amc) = \varnothing$?
	\item Membership: given a~PA $\Amc$ and a finite words $w$, is $w \in L(\Amc)$?
	\item Universality: given a~PA $\Amc$, is $L(\Amc) = \Sigma^*$?
\end{itemize}

As shown by Klaedtke and Ruess~\cite{klaedtkeruess}, emptiness and membership are decidable for deterministic and non-deterministic~PA. Furthermore, they have shown that universality is decidable for deterministic~PA but undecidable for non-deterministic~PA. 
Refining these results, Figueira and Libkin~\cite{emptynp} showed that non-emptiness and membership are $\NP$-complete for non-deterministic~PA, and these results translate quickly to the deterministic setting as well.

As a preparation for the results in \Cref{sec:det} we show that universality for deterministic~PA is $\Pi_2^\P$-complete in order to establish complexity bounds of several problems related to~PA operating on infinite words.

To achieve that, we first introduce an auxiliary problem for~PA and show that it is already $\Pi_2^\P$-hard even restricted to deterministic acyclic~PA.
We define the \emph{irrelevance problem} for~PA as follows: given a~PA $\Amc = (Q, \Sigma, q_0, \Delta, F, C)$ of dimension $d$, is $\Amc$ equivalent to $\Amc^* = (Q, \Sigma, q_0, \Delta, F, \Nbb^d)$? In other words: is $C$ irrelevant for $\Amc$ in the sense that every run of $\Amc$ reaching an accepting state satisfies the Parikh condition?
In the following, we denote by $\rho(\Amc) = \{\rho(r) \mid r \text{ is an accepting run of } \Amc^*)$ and note that this set is always semi-linear as a consequence of Parikh's theorem~\cite{parikh1966context} and~\cite[Lemma 5]{klaedtkeruess}.

\begin{lemma}
	\label{lem:irrelevance}
	The irrelevance problem for deterministic acyclic~PA is $\Pi^\P_2$-hard.
\end{lemma}
\begin{proof}
	We reduce from the inclusion problem for integer expressions, which is known to be $\Pi^\P_2$-complete (assuming that all numbers are encoded in binary)~\cite{intexpr,polyhierarchy}. The set of integer expressions is defined as follows. Every $n \in \Nbb$ is an integer expression with $L(n) = \{n\}$. If $e_1$ and $e_2$ are integer expressions, then so are $e_1 + e_2$ and $e_1 \cup e_2$ where $L(e_1 + e_2) = \{u+v \mid u \in L(e_1), v \in L(e_2)\}$ and $L(e_1 \cup e_2) = L(e_1) \cup L(e_2)$. The inclusion problem is defined as follows: given integer expressions $e_1, e_2$, is $L(e_1) \subseteq L(e_2)$?
	
	We proceed as follows: first, we construct a linear set $C(\0, P_2)$ of dimension $d+1$ from~$e_2$ with the property that $n \in L(e_2)$ if and only if $n \cdot \1^d \in C(\0, P_2)$, where $d$ and $|P_2|$ depend linearly on the number of operators in $e_2$. Second, we construct a deterministic acyclic~PA~$\Amc_1$ from~$e_1$ with the property $L(e_1) = \rho(\Amc_1)$. Finally, we construct a deterministic acyclic~PA~$\Amc$ from~$\Amc_1$ such that $C(\0, P_2)$ is irrelevant for~$\Amc$ if and only if $L(e_1) \subseteq L(e_2)$.

	\begin{claim}
		\label{claim:intexToSet} 
		Given an integer expression $e$, we can compute in polynomial time a linear set~$C(\0, P)$ such that $n \in L(e)$ if and only if $n \cdot \1 \in C(\0, P)$.  
	\end{claim}
	\begin{claimproof}
		We construct $C(\0, P)$ inductively from $e$ and maintain the following invariant in each step: $n \in L(e)$ if and only if $n \cdot \1 \in C(\0, P)$ and for all $\pbf = (m, p_1 \dots, p_d) \in P$ we have $p_i = 1$ for at least one $1 \leq i \leq d$. We refer to \Cref{fig:intexToLinear} for an illustration.
		\begin{figure}
			\centering
			\begin{subfigure}[b]{0.1\textwidth}
				\quad
				\begin{tabular}{|c|}
					\hline
					$n$ \\\hline $1$ \\
					\hline
				\end{tabular}
				\caption{$e = n$}
			\end{subfigure}
			\hfill
			\begin{subfigure}[b]{0.3\textwidth}
				\begin{tabular}{|c|c|}
					\hline
					$n_1 \dots n_k$ & \ $m_1 \dots m_\ell$\ {} \\\hline & \\[5pt]        
					\LARGE $P_1$ & \LARGE $\0$ \\[5pt] & \\\hline & \\[0pt]   
					\LARGE \0 & \LARGE $P_2$ \\[0pt] & \\\hline
				\end{tabular}
				\caption{$e = e_1 + e_2$}
			\end{subfigure}
			\hfill
			\begin{subfigure}[b]{0.4\textwidth}
				\begin{tabular}{|c|c|c|c|}
					\hline
					$n_1 \dots n_k$ & \ $m_1 \dots m_\ell$\ {} & 0 & 0 \\\hline & & & \\[5pt]      
					\LARGE $P_1$ & \LARGE $\0$ & \1 & \0 \\[5pt] & & &  \\\hline & & & \\[0pt]   
					\LARGE \0 & \LARGE $P_2$  & \0 & \1 \\[0pt] & & & \\\hline
					$0 \dots 0$ & $0 \dots 0$ & 1 & 1 \\\hline
				\end{tabular}
				\caption{$e = e_1 \cup e_2$}
			\end{subfigure}
			
			\caption{Illustration of the construction of a linear set from an integer expression.}
			\label{fig:intexToLinear}
		\end{figure}
		
		\textbf{Base Case.} If $e = n$ for $n \in \Nbb$, we choose $C(\0, P) \in \Nbb^2$ with $P = \{(n, 1)\}$. \par
		\textbf{Step.} We need to consider two cases.
		
		If $e = e_1 + e_2$, there are linear sets $C(\0, P_1) \subseteq \Nbb^{1+d_1}$ and $C(\0, P_2) + \Nbb^{1+d_2}$ for $e_1$ and~$e_2$ satisfying the invariant by assumption. We construct a linear set $C(\0, P)$ of dimension $1+d_1 + d_2$
		as follows. We pad the vectors in $P_1$ and $P_2$ with zeros to align the dimensions. Then~$P$ is the union of these vectors, \ie,
        \[P = \{n \cdot \pbf_1 \cdot \0^{d_2} \mid n \cdot \pbf_1 \in P_1\} \cup \{n \cdot \0^{d_1} \cdot \pbf_2 \mid n \cdot \pbf_2 \in P_2\}.\]
		The invariant is maintained using this construction: for every integer $m + n \in L(e)$ we have $m \cdot \1^{d_1} \cdot \0^{d_2} \in C(\0, \{m \cdot \pbf_1 \cdot \0^{d_2} \mid m \cdot \pbf_1 \in P_1\})$ and $n \cdot \0^{d_1} \cdot \1^{d_2} \in C(\0, \{n \cdot \0^{d_2} \cdot \pbf_2 \mid n \cdot \pbf_2 \in P_2\})$ by assumption and construction. Hence, $(n+m) \cdot \1^{d_1 + d_2} \in C(\0, P)$. 
		Likewise, if $n \cdot \1^{d_1 + d_2} \in C(\0, P)$, we can partition the sets of used period vectors in the set of period vectors originating from $P_1$, and the set originating from $P_2$. As only the period vectors originating from $P_1$ can modify the first $d_1$ counters (ignoring the first), they yield a number $n_1 \in L(e_1)$. Similarly, the period vectors from $p_2$ yield a number $n_2 \in L(e_2)$, and hence $n = n_1 + n_2 \in L(e)$.
		
		If $e = e_1 \cup e_2$, there are linear sets $C(\0, P_1) \subseteq \Nbb^{1+d_1}$ and $C(\0, P_2) \subseteq \Nbb^{1+d_2}$ for $e_1$ and~$e_2$ satisfying the invariant by assumption. We construct a linear set $C(\0, P)$ of dimension $1+d_1 + d_2+1$ as
		as follows. Again, we pad the vectors in $P_1$ and $P_2$ with zeros to align the dimensions, and add an additional 0-counter. Furthermore, we consider the two vectors $\vbf_1 = 0 \cdot \1^{d_1} \cdot \0^{d_2} \cdot 1$ and $\vbf_2 = 0 \cdot \0^{d_1} \cdot \1^{d_2} \cdot 1$. 
		Then~$P$ is the union of these vectors, \ie,%
		\[P = \{n \cdot \pbf_1 \cdot \0^{d_2} \cdot 0 \mid n \cdot \pbf_1 \in P_1\} \cup \{n \cdot \0^{d_1} \cdot \pbf_2 \cdot 0 \mid n \cdot \pbf_2 \in P_2\} \cup \{\vbf_1, \vbf_2\}.\] 
		The invariant is maintained using this construction: if $n \in L(e_1)$, then $n \cdot \1^{d_1 + d_2 + 1} \in C(\0, P)$ as witnessed by the following choice of period vectors. First, we can use~$\vbf_2$ to ensure that the last $d_2 + 1$ entries (including the new counter) are indeed set to one. Furthermore, the first $d_1 + 1$ entries of~$\vbf_1$ (including the first counter) are 0; hence, they do not modify the relevant counters for $n$. Hence, the containment follows from the invariant and construction. We argue analogously if $n \in L(e_2)$ using~$\vbf_1$.
		Now, if $n \cdot \1^{d_1 + d_2 + 1} \in C(\0, P)$ we observe that exactly one of the vectors~$\vbf_1$ and $\vbf_2$ must have been used, as these are the only ones that modify the last counter. If~$\vbf_1$ has been used, then no period vector originating in $P_1$ may have been used, as they all contain a one-entry by the invariant, and are hence blocked by the 1-entries in $\vbf_1$. As $\vbf_1$ does not modify the first counter, and all used period vectors originate from $P_2$, we conclude $n \in L(e_2)$. Analogously, if $\vbf_2$ has been used, we conclude $n \in L(e_1)$.
		
		Observe that the dimension $d$ and the size of $P$ both depend linearly on the size of $e$, hence $|C(\0, P)| \in \Omc(|e|^2)$.
	\end{claimproof}
	
	\begin{claim}
		\label{claim:intexToPA}
		Given an integer expression $e$, we can compute in polynomial time a deterministic acyclic~PA with a single accepting state $\Amc$ such that $L(e) = \rho(\Amc)$.
	\end{claim}
	\begin{claimproof}
		\begin{figure}
			\centering
			\begin{subfigure}[b]{0.32\textwidth}
				\begin{tikzpicture}[->,>=stealth',shorten >=1pt,auto,node distance=2.5cm, semithick]
					\tikzstyle{every state}=[minimum size=5mm]
					
					\node[state, initial above, initial text={}] (q0) {$q_0$};	
					\node[state, accepting] (q1) [right of=q0] {$q_1$};
					
					\path
					(q0) edge  node[align=left] {$a, n$ \\ $b, n$} (q1)
					;
				\end{tikzpicture}
				\caption{$e = n$}
			\end{subfigure}
			\hfill
			\begin{subfigure}[b]{0.32\textwidth}
				\begin{tikzpicture}[->,>=stealth',shorten >=1pt,auto,node distance=1cm, semithick]
					\tikzstyle{every state}=[minimum size=7mm]
					
					\node[state, initial, initial text={}] (p0) {$p_0$};	
					\node[state, accepting, dashed] (p1) [right = 6mm of p0] {$p_m$};
					\node[state, below = 18mm of p1] (q0) {$q_0$};	
					\node[state, accepting] (q1) [left = 6mm of q0] {$q_n$};
					
					\node[right = -0.3mm of p0] {$\dots$};
					\node[left = -0.6mm of q0] {$\dots$};
					
					\node[fit = (p0) (p1), draw]{};
					\node[fit = (q0) (q1), draw]{};
					
					\path
					(p1) edge node[align=left] {$a, 0$ \\ $b, 0$} (q0)
					;
				\end{tikzpicture}
				\caption{$e = e_1 + e_2$}
			\end{subfigure}
			\hfill
			\begin{subfigure}[b]{0.32\textwidth}
				\begin{tikzpicture}[->,>=stealth',shorten >=1pt,auto,node distance=1cm, semithick]
					\tikzstyle{every state}=[minimum size=7mm]
					
					\node[state, initial left, initial text={}] (q) {$q$};	
					\node[state, above right = 8mm and 3mm of q] (p0) {$p_0$};	
					\node[state, accepting, dashed] (p1) [right = 6mm of p0] {$p_m$};
					\node[state, below right = 8mm and 3mm of q] (q0) {$q_0$};	
					\node[state, accepting, dashed] (q1) [right = 6mm of q0] {$q_n$};
					\node[state, accepting] (qf) [below right = 8mm and 3mm of p1] {$q_f$};
					
					\node[right = -0.3mm of p0] {$\dots$};
					\node[right = -0.3mm of q0] {$\dots$};
					
					\node[fit = (p0) (p1), draw]{};
					\node[fit = (q0) (q1), draw]{};
					
					\path
					(q)  edge node[align=left] {$a, 0$}                        (p0)
					(q)  edge node[align=left, below left] {$b, 0$}            (q0)
					(p1) edge node[align=left] {$a, 0$ \\ $b, 0$}              (qf)
					(q1) edge node[align=left, below right] {$a, 0$ \\ $b, 0$} (qf)
					;
				\end{tikzpicture}
				\caption{$e = e_1 \cup e_2$}
			\end{subfigure}
			
			\caption{Illustration of the construction of a deterministic acyclic~PA from an integer expression.}
			\label{fig:intexToPA}
		\end{figure}
		We construct $\Amc$ over the alphabet $\{a,b\}$ of dimension $1$ with linear set $\Nbb$ inductively from $e$ and maintain the invariant in the claim in every step. We refer to \Cref{fig:intexToPA} for an illustration.
		
		\textbf{Base Case.} If $e = n$ for $n \in \Nbb$, the~PA $\Amc$ consists of an initial state and an accepting state that are connected by an $a$-transition and a $b$-transition, both labeled with $n$.
		
		\textbf{Step.} We need to consider two cases.\par
		If $e = e_1 + e_2$, there are~PA $\Amc_1$ and $\Amc_2$ for $e_1$ and $e_2$ satisfying the invariant by assumption. We construct a~PA $\Amc$ as follows: we take the disjoint union of $\Amc_1$ and $\Amc_2$ and connect the accepting state of $\Amc_1$ with the initial state of $\Amc_2$ via an $a$-transition and a $b$-transitions, both labeled with $0$. Finally, the accepting states of $\Amc_1$ is not accepting anymore in $\Amc$.
		
		If $e = e_1 \cup e_2$, there are~PA $\Amc_1$ and $\Amc_2$ for $e_1$ and $e_2$ satisfying the invariant by assumption. We construct a~PA $\Amc$ as follows: we take the disjoint union of $\Amc_1$ and~$\Amc_2$ and add a fresh initial state, say $q$, and a fresh accepting state, say $q_f$. Then, we connect $q$ with the initial state of $\Amc_1$ with an $a$-transition labeled with $\0$, and we connect $q$ with the initial state of~$\Amc_2$ with a $b$-transition labeled with $0$. Similarly, we connect the accepting state of $\Amc_1$ as well as the accepting state of $\Amc_2$ with $q_f$ via $a$-transitions and $b$-transitions, all labeled with ~$0$. Finally, the accepting states of $\Amc_1$ and $\Amc_2$ are not accepting anymore in~$\Amc$.
	\end{claimproof}
	
	We are now ready to prove \Cref{lem:irrelevance}. Let $\Amc_1$ be the~PA for $e_1$ as constructed in the proof of \Cref{claim:intexToPA}. Similarly, let $C(\0, P_2)$ be the linear set of dimension $1+d$ for~$e_2$ as constructed in the proof of \Cref{claim:intexToSet}.
	Now we construct the deterministic acyclic~PA~$\Amc$ as follows. We start with $\Amc_1$ and pad all vectors with zeros, that is, we replace every (one-dimensional) vector $n$ in $\Amc_1$ by $n \cdot \0^d$.
	Then, we add a fresh accepting state and connect it from the accepting state of $\Amc_1$ with an $a$-transition and $b$-transition, both labeled with $0 \cdot \1^d$.
	Finally, the accepting state of $\Amc_1$ is not accepting anymore in $\Amc$, and we choose $C(\0, P_2)$ as the linear set of $\Amc$.
	Observe that the properties of $\Amc_1$ and $C(\0, P_2)$ ensure that $C(\0, P_2)$ is irrelevant for $\Amc$ if and only if $L(e_1) \subseteq L(e_2)$. This concludes the proof. 
\end{proof}

\begin{remark*}
	Recall that we assume a binary and explicit encoding of semi-linear sets. In terms of expressiveness we can equally assume that they are given as Presburger formulas. However, this drastically changes the complexity: if the semi-linear sets are given as Presburger formulas, the problem becomes $\coNEXP$-complete, as we can interreduce the problem with the $\forall^*\exists^*$-fragment of Presburger arithmetic, which is $\coNEXP$-complete~\cite{haaseNexp}.
\end{remark*}

Observe that we can easily reduce the previous problem to the universality problem for deterministic~PA.
\begin{corollary}
	\label{cor:uniHardnessFinite}
	Universality for deterministic~PA is $\Pi_2^\P$-hard.
\end{corollary}
\begin{proof}
	Let $\Amc$ be the deterministic acyclic~PA with linear set $C(\0, P)$ as constructed in the previous proof. 
	We construct a deterministic $\Amc'$ such that $C(\0, P)$ is irrelevant for~$\Amc$ if and only if $\Amc'$ is universal.
	To achieve that, we start with $\Amc$ and make every state accepting. 
	Furthermore, we add an $a$-loop and a $b$-loop to the accepting state of $\Amc$, both labeled with~$\0$. Finally, the semi-linear set of $\Amc'$ is $C(\0, P) \cup (\Nbb \cdot \{\0\})$.
\end{proof}

Now we show that universality for deterministic~PA is in $\Pi_2^\P$, yielding completeness for universality and irrelevance by the reduction presented in the proof of the previous corollary.

\begin{lemma}
	\label{lem:uniFinite}
	Universality for deterministic~PA is $\Pi_2^\P$-complete.
\end{lemma}
\begin{proof}
	We show that the problem is contained in $\Pi_2^\P$ by showing that its complement is contained in $\Sigma_2^\P$. We make use of the following observation. As we can assume that every state is accepting (by encoding accepting states into the semi-linear set), the question whether a deterministic~PA $\Amc$ with semi-linear set $C$ is not universal boils down to the question $\rho(\Amc) \not\subseteq C$?
	If both sets are given explicitly (again assuming binary encoding), Huynh showed that this question can be decided in $\Sigma_2^\P$~\cite{huynhUpperSimplified,huynhUpper}.
	However, we cannot explicitly compute the set $\rho(\Amc)$ as its size can be exponentially large in the size of~$\Amc$~\cite{parikhComplexity}. To circumvent this problem, we guess a small linear subset of $\rho(\Amc)$ that witnesses non-inclusion.
	
	In the first step, we compute an existential Presburger formula, say $\varphi(\vbf) = \exists x_1 \dots x_m \psi$ defining the set $\rho(\Amc)$. This can be done in linear time using the results in~\cite{schwentikHorn}; we also refer to~\cite[Proposition III.2]{emptynp} for a short discussion on the construction.
	
	In the second step, we use the result in~\cite[Corollary II.2]{zvassnz} essentially stating that every semi-linear set can be written as a finite union of linear sets of small bit size. To be precise, the result states that every semi-linear set $C \subseteq \Nbb^d$ can be written as $\bigcup_i C(\bbf_i, P_i)$ with $|P_i| \leq 2d \log(4d \|C\|)$, where $\|C\|$ denotes the largest absolute value that appears in a base or period vector in any linear set in the union of $C$. Hence, we guess a linear set $C(\bbf_i, P_i) \subseteq \rho(\Amc)$ of small bit size. Note however, that we do not verify at this point whether $C(\bbf_i, P_i)$ is indeed a subset of $\rho(\Amc)$.
	
	In the third step, we use Huynh's results~\cite{huynhUpperSimplified,huynhUpper} to solve $C(\bbf_i, P_i) \not\subseteq C$ in $\Sigma_2^\P$. As one main ingredient to obtain containment $\Sigma_2^\P$, Huynh showed that every non-empty (set-theoretic) difference of two semi-linear sets contains a vector of polynomial bit size. Hence let $\vbf \in C(\bbf_i, P_i) \setminus C$ be such a vector whose bit size is bounded polynomially in $C(\bbf_i, P_i)$ and~$C$, and hence in the input size.
	
	Finally, we need to verify that $\vbf$ is indeed a member of $\rho(\Amc)$ (recall that we did not verify that $C(\bbf_i, P_i)$ is indeed a subset of $C_\Amc$). In order to do so, we guess a valuation of the quantified variables $x_1, \dots, x_n$ (again of polynomial bit size) of $\varphi$ and verify that~$\varphi(\vbf)$ is indeed satisfied under this valuation.
	
	Overall, we conclude that non-universality for deterministic~PA is in $\Sigma_2^\P$, yielding the desired result.
\end{proof}

%% file: 4_auto_infinite.tex
\section{Parikh Automata on Infinite Words}

In this section, we recall the acceptance conditions of Parikh automata operating on infinite words that were studied before in the literature and introduce our new models. We make some easy observations and compare the existing with the new automata models. First we focus on the non-deterministic variants of these models before studying their deterministic variants.
Whenever we do not explicitly specify that a model is deterministic, we mean the non-deterministic variant.

\smallskip
Let $\Amc = (Q, \Sigma, q_0, \Delta, F, C)$ be a~PA. A run of $\Amc$ on an infinite word $\alpha = \alpha_1 \alpha_2 \alpha_3 \dots$ is an infinite sequence of labeled transitions $r = r_1 r_2 r_3 \dots$ with $r_i = (p_{i-1}, \alpha_i, \vbf_i, p_i)\in \Delta$ such that $p_0 = q_0$. The automata defined below differ only in their acceptance conditions; hence, the notion of determinism translates directly. In the following, whenever we say that an automaton $\Amc$ accepts an infinite word $\alpha$, we mean that there is an accepting run of $\Amc$ on $\alpha$.

\begin{enumerate}
	\item The run $r$ satisfies the \emph{safety condition} if for every $i \geq 0$ we have $p_i \in F$ and $\rho(r_1 \dots r_i) \in~C$. We call a~PA accepting with the safety condition a \emph{safety~PA}~\cite{infiniteZimmermann}. 
	We define the $\omega$-language recognized by a safety~PA $\Amc$ as 
	\[S_\omega(\Amc) = \{\alpha\in \Sigma^\omega\mid \Amc \text{ accepts } \alpha\}.\] 
	
	\item The run $r$ satisfies the \emph{reachability condition} if there is an $i \geq 1$ such that $p_i \in F$ and $\rho(r_1 \dots r_i) \in~C$. We say there is an \emph{accepting hit} in $r_i$. We call a~PA accepting with the reachability condition a \emph{reachability~PA}~\cite{infiniteZimmermann}. We define the $\omega$-language recognized by a reachability~PA $\Amc$ as
	\[R_\omega(\Amc) = \{\alpha\in \Sigma^\omega\mid \Amc \text{ accepts } \alpha\}.\] 
	
	\item The run $r$ satisfies the \emph{Büchi condition} if there are infinitely many $i \geq 1$ such that $p_i \in F$ and $\rho(r_1 \dots r_i) \in C$. We call a~PA accepting with the Büchi condition a \emph{Büchi~PA}~\cite{infiniteZimmermann}. 
	We define the $\omega$-language recognized by a Büchi~PA $\Amc$ as 
	\[B_\omega(\Amc) = \{\alpha\in \Sigma^\omega\mid \Amc \text{ accepts } \alpha\}.\] 
	
	Hence, a Büchi~PA can be seen as a stronger variant of a reachability~PA where we require infinitely many accepting hits instead of a single one.
	
	\item The run $r$ satisfies the \emph{co-Büchi condition} if there is $i_0$ such that for every $i \geq i_0$ we have $p_i \in F$ and $\rho(r_1 \dots r_i) \in C$. We call a~PA accepting with the co-Büchi condition a \emph{co-Büchi~PA}~\cite{infiniteZimmermann}. 
	We define the $\omega$-language recognized by a co-Büchi~PA $\Amc$ as 
	\[CB_\omega(\Amc) = \{\alpha\in \Sigma^\omega\mid \Amc \text{ accepts } \alpha\}.\] 
	
	Hence, a co-Büchi~PA can be seen as a weaker variant of safety~PA where the safety condition needs not necessarily be fulfilled from the beginning, but from some point onwards.
\end{enumerate}

We now define the models newly introduced in this work. As already observed in \cite{infiniteZimmermann} among the above considered models only Büchi~PA can recognize all $\omega$-regular languages. For example, $\{\alpha\in \{a,b\}^\omega\mid |\alpha|_a=\infty\}$ cannot be recognized by safety~PA, reachability~PA or co-Büchi~PA. 

We first extend reachability~PA with the classical Büchi condition to obtain \emph{reachability-regular~PA}. In \Cref{thm:LimitEqualsReach} we show that these automata characterize the class
\[\LPAReg = \left\{\bigcup_{i \leq n} U_i V_i \mid n \geq 1, U_i \text{ is Parikh-recognizable}, V_i \text{ is regular}\right\}\]
hence, providing a robust and natural model. 

\begin{enumerate}
	\setcounter{enumi}{4}
	\item The run satisfies the \emph{reachability and regularity condition} if there is an $i \geq 1$ such that $p_i \in F$ and $\rho(r_1 \dots r_i) \in C$, and there are infinitely many $j \geq 1$ such that $p_j \in F$. We call a~PA accepting with the reachability and regularity condition a \emph{reachability-regular}~PA. We define the $\omega$-language recognized by a reachability-regular~PA $\Amc$~as 
	\[RR_\omega(\Amc) = \{\alpha\in \Sigma^\omega\mid \Amc \text{ accepts } \alpha\}.\] %and call it reachability-regular.
\end{enumerate}

Observe that every $\omega$-regular language is reachability-regular~PA recognizable, as we can turn an arbitrary Büchi automaton into an equivalent reachability-regular~PA by labeling every transition with $0$ and \mbox{setting~$C = \{0\}$}.

Next we introduce \emph{limit~PA}, which were proposed in the concluding remarks of~\cite{klaedtkeruess}. As we will prove in \Cref{thm:LimitEqualsReach}, this seemingly quite different model is equivalent to reachability-regular~PA. 

\begin{enumerate}
	\setcounter{enumi}{5}
	\item The run satisfies the \emph{limit condition} if there are infinitely many~$i \geq 1$ such that $p_i \in F$, and if additionally $\rho(r) \in C$, where the $j$th component of $\rho(r)$ is computed as follows. 
	If there are infinitely many~$i \geq 1$ such that the $j$th component of $\vbf_i$ has a non-zero value, then the $j$th component of $\rho(r)$ is~$\infty$. In other words, if the sum of values in a component diverges, then its value is set to $\infty$. 
	Otherwise, the infinite sum yields a positive integer. We call a~PA accepting with the limit condition a \emph{limit~PA}.
	We define the $\omega$-language recognized by a limit~PA $\Amc$ as 
	\[L_\omega(\Amc) = \{\alpha\in \Sigma^\omega\mid \Amc \text{ accepts } \alpha\}.\] 
	
\end{enumerate}

Still, none of the yet introduced models have $\omega$-closure. This shortcoming is addressed with the following two models, which will turn out to be equivalent and form the basis of the automata characterization of the classes 
\[\LPAPA = \left\{\bigcup_{i \leq n} U_i V_i \mid n \geq 1, U_i \text{ and } V_i \text{ are Parikh-recognizable}\right\}\]
and 
\[\LRegPA = \left\{\bigcup_{i \leq n} U_i V_i \mid n \geq 1, U_i \text{ is regular}, V_i \text{ is Parikh-recognizable}\right\}.\] 

\begin{enumerate}
	\setcounter{enumi}{6}
	\item The run satisfies the \emph{strong reset condition} if the following holds. Let $k_0 = 0$ and denote by $k_1 < k_2 < \dots$ the positions of all accepting states in~$r$. Then $r$ is accepting if $k_1, k_2, \dots$ is an infinite sequence and $\rho(r_{k_{i-1}+1} \dots r_{k_i}) \in C$ for all $i \geq 1$.
	We call a~PA accepting with the strong reset condition a \emph{strong reset~PA}. We define the $\omega$-language recognized by a strong reset~PA $\Amc$ as
	\[SR_\omega(\Amc) = \{\alpha\in \Sigma^\omega\mid \Amc \text{ accepts } \alpha\}.\] 
	
	\item The run satisfies the \emph{weak reset condition} if there are infinitely many reset positions $0 = k_0 < k_1 < k_2, \dots$ such that $p_{k_i} \in F$ and $\rho(r_{k_{i-1}+1} \dots r_{k_i}) \in C$ for all $i \geq 1$.
	We call a~PA accepting with the weak reset condition a \emph{weak reset~PA}.
	We define the $\omega$-language recognized by a weak reset~PA $\Amc$ as 
	\[W\!R_\omega(\Amc) = \{\alpha\in \Sigma^\omega\mid \Amc \text{ accepts } \alpha\}.\] 
	
\end{enumerate}

\smallskip
Intuitively worded, whenever a strong reset~PA enters an accepting state, the Parikh condition \emph{must} be satisfied. Then the counters are reset.
Similarly, a weak reset~PA may reset the counters whenever there is an accepting hit, and they must reset infinitely often, too.
In the following we will often just speak of reset~PA without explicitly stating whether they are weak or strong. In this case, we mean the strong variant. We will show the equivalence of the nondeterministic variants of two models in \Cref{lem:SPBAtoWPBA} and \Cref{lem:WPBAtoSPBA}, while in the deterministic setting weak reset~PA are more expressive than strong reset~PA, see~\Cref{lem:detStrongResetVsdetWeakReset}.

Guha et al.~\cite{infiniteZimmermann} assume that reachability~PA are complete, i.e., for every $(p,a)\in Q\times \Sigma$ there are $\vbf\in \Nbb^d$ and $q\in Q$ such that $(p,a,\vbf,q)\in \Delta$, as incompleteness allows to express additional safety conditions. 
We also make this assumption in order to avoid inconsistencies. 
In fact, we can assume that all models are complete, as the other models can be completed by adding a non-accepting sink. Observe that their deterministic variants are always complete by definition. 
%We remark that Guha et al.\ also considered \emph{asynchronous} reachability and Büchi~PA, where the Parikh condition does not necessarily need to be satisfied in accepting states. However, for non-deterministic automata this does not change the expressiveness of the considered models~\cite{infiniteZimmermann}.

\begin{figure}[h]
	\centering
	\begin{tikzpicture}[->,>=stealth',shorten >=1pt,auto,node distance=3cm, semithick]
		\tikzstyle{every state}=[minimum size=1cm]
		
		\node[state, accepting, initial, initial text={}] (q0) {$q_0$};	
		\node[state] (q1) [right of=q0] {$q_1$};

		\path
		(q0) edge [loop above] node {$b, \begin{pmatrix}0\\1\end{pmatrix}$} (q0)
		(q0) edge              node {$a, \begin{pmatrix}1\\0\end{pmatrix}$} (q1)
		(q1) edge [loop above] node {$a, \begin{pmatrix}1\\0\end{pmatrix}$} (q1)
		(q1) edge [bend left]  node {$b, \begin{pmatrix}0\\1\end{pmatrix}$} (q0)	
		;
	\end{tikzpicture}
	\caption{The automaton $\Amc$ with $C=\{(z,z'), (z, \infty) \mid z' \geq z\}$ from \Cref{ex}.}
	\label{fig:ex}
\end{figure}

\begin{example}
	\label{ex}
	Let $\Amc$ be the automaton in \Cref{fig:ex} with $C = \{(z,z'), (z, \infty) \mid z' \geq z\}$. 
	\begin{itemize}
		\item If we interpret $\Amc$ as a~PA (over finite words), then we have $L(\Amc) = \{w \in \{a,b\}^* \cdot \{b\} \mid |w|_a \leq |w|_b\} \cup \{\varepsilon\}$. The automaton is in the accepting state at the very beginning and every time after reading a $b$. The first counter counts the occurrences of $a$, the second one counts occurrences of $b$. By definition of~$C$ the automaton only accepts when the second counter value is greater or equal to the first counter value (note that vectors containing an $\infty$-entry have no additional effect). 
		
		\item If we interpret $\Amc$ as a safety~PA, then we have $S_\omega(\Amc) = \{b\}^\omega$. As $q_1$ is not accepting, only the $b$-loop on $q_0$ may be used.
		
		\item If we interpret $\Amc$ as a reachability~PA, then we have $R_\omega(\Amc) = \{\alpha \in \{a,b\}^\omega \mid \alpha$ has a prefix in $L(\Amc)\}$. The automaton has satisfied the reachability condition after reading a prefix in $L(\Amc)$ and accepts any continuation after that. 
		
		\item If we interpret $\Amc$ as a Büchi~PA, then we have $B_\omega(\Amc) = L(\Amc)^\omega$. The automaton accepts an infinite word if infinitely often the Parikh condition is satisfied in the accepting state. Observe that $C$ is a homogeneous linear set and the initial state as well as the accepting state have the same outgoing transitions.
		
		\item If we interpret $\Amc$ as a co-Büchi~PA, then we have $CB_\omega(\Amc) = L(\Amc) \cdot \{b\}^\omega$. This is similar to the safety~PA, but the accepted words may have a finite ``non-safe'' prefix from $L(\Amc)$.
		
		\item If we interpret $\Amc$ as a reachability-regular~PA, then we have $RR_\omega(\Amc) = \{\alpha \in \{a,b\}^\omega \mid \alpha \text{ has a prefix}$ in $L(\Amc) \text{ and } |\alpha|_b = \infty\}$. After having met the reachability condition the automaton still needs visit an accepting state infinitely often.
        % to satisfy the Büchi condition, which enforces infinitely many visits of the accepting state. 
		\item If we interpret $\Amc$ as a limit~PA, then we have $L_\omega(\Amc) = \{\alpha \in \{a,b\}^\omega \mid |\alpha|_a < \infty\}$. The automaton must visit the accepting state infinitely often. At the same time the extended Parikh image must belong to $C$, which implies that the infinite word contains only some finite number $z$ of symbol $a$ (note that only the vectors of the form $(z, \infty)$ have an effect here, as at least one symbol must be seen infinitely often by the infinite pigeonhole principle).
		
		\item If we interpret $\Amc$ as a weak reset~PA, then we have $W\!R_\omega(\Amc) = L(\Amc)^\omega$. As a weak reset~PA may (but is not forced to) reset the counters upon visiting the accepting state, the automaton may reset every time a (finite) infix in $L(\Amc)$ has been read.

		\item If we interpret $\Amc$ as a strong reset~PA, then we have $SR_\omega(\Amc) = \{b^*a\}^\omega \cup \{b^* a\}^* \cdot \{b\}^\omega$. Whenever the automaton reaches an accepting state also the Parikh condition must be satisfied. This implies that the $a$-loop on $q_1$ may never be used, as this would increase the first counter value to at least 2, while the second counter value is 1 upon reaching the accepting state $q_0$ (which resets the counters).
	\end{itemize}
\end{example}

%\begin{remark*}
%	The automaton $\Amc$ in the example is deterministic. We note that $L_\omega(\Amc)$ is not deterministic $\omega$-regular but deterministic limit~PA recognizable.
%\end{remark*}

%% file: 4.1_auto_nondet.tex
\section{Nondeterministic Parikh Automata on Infinite Words}
\label{sec:nondet}
In this section, we study the nondeterministic variants of Parikh automata on infinite words.
We first show that the newly introduced variants are closed under union and left-concatenation with Parikh recognizable languages, yielding the foundation in order to establish the claimed characterizations. These characterizations in turn help us to establish the remaining closure properties. Motivated by further examining the expressiveness of all these models, we study the effect of $\varepsilon$-transitions and show that almost all considered models admit $\varepsilon$-elimination, the exception being safety and co-Büchi~PA. Based on this result, we show that Büchi~PA and blind counter automata operating on infinite words as introduced by Fernau and Stiebe~\cite{blindcounter} are equivalent.
Finally, we study the classical decision problems with application to model checking.

\subsection{Preparation}
It was observed in~\cite{infiniteZimmermann} that Büchi~PA recognize a strict subset of $\LPAPA$. 
In this section we first show that the class of reset~PA recognizable $\omega$-languages is a strict superset of~$\LPAPA$. Then we provide an automata-based characterization of $\LPAReg, \LPAPA$, and~$\LRegPA$.
Towards this goal we first establish some closure properties.

Guha et al.~\cite{infiniteZimmermann} have shown that safety, reachability, Büchi, and co-Büchi~PA are closed under union using a modification of the standard construction for~PA, \ie, taking the disjoint union of the automata (introducing a fresh initial state), and the disjoint union of the semi-linear sets, where disjointness is achieved by ``marking'' every vector in the first set by an additional $1$ (increasing the dimension by 1), and all vectors in the second set by an additional $2$. 
Then all transitions from the new initial state leading to the copy of the first automaton are also marked with a $1$ and those leading to the copy of the second automaton are marked with a 2.
We observe that the same construction also works for reachability-regular and limit~PA, and a small modification is sufficient to make the construction also work for reset~PA. 
To be precise, we need to refresh the mark every time we leave an accepting state, as the reset ``forgets'' this information.
We leave the details to the reader. 

\begin{lem}
\label{cor:union}
    The classes of reachability-regular~PA recognizable, limit~PA recognizable, and reset~PA recognizable $\omega$-languages are closed under union.
\end{lem}

Furthermore, we show that these classes, as well as the class of Büchi~PA recognizable $\omega$-languages, are closed under left-concatenation with~PA recognizable languages. 
We provide some details in the next lemma, as we will need to modify the standard construction in such a way that we do not need to keep accepting states of the~PA on finite words. This will help to characterize $\LPAPA$ via (restricted) reset~PA. 

\begin{lemma}
\label{lem:concatenation}
     The classes of reachability-regular~PA recognizable, limit~PA recognizable, reset~PA recognizable, Büchi~PA recognizable, co-Büchi~PA recognizable, and safety~PA recognizable $\omega$-languages are closed under left-concatenation with~PA recognizable languages.
\end{lemma}
\begin{proof}
    We begin with reset~PA. Let $\Amc_1 = (Q_1, \Sigma, q_1, \Delta_1, F_1, C_1)$ be a~PA of dimension $d_1$ and let $\Amc_2 = (Q_2, \Sigma, q_2, \Delta_2, F_2, C_2)$ be a reset~PA of dimension $d_2$. We sketch the construction of a reset~PA $\Amc$ of dimension $d_1 + d_2$ that recognizes $L(\Amc_1) \cdot SR_\omega(\Amc)$. 
    We assume without loss of generality that $q_2$ is accepting (this can be achieved by introducing a fresh initial state). 
    Furthermore, for now we assume that $\varepsilon \notin L(\Amc_1)$, that is, every accepting run of~$\Amc_1$ is not empty.
    Again, $\Amc$ consists of disjoint copies of~$\Amc_1$ and $\Amc_2$ but only the accepting states of $\Amc_2$ remain accepting, and the initial state of $\Amc$ is $q_1$. All transitions of the copy of~$\Amc_1$ use the first $d_1$ counters (that is, the remaining $d_2$ counters are always 0), and, likewise, the transitions of $\Amc_2$ only use the last $d_2$ counters (that is, the first $d_1$ counters are always 0). 
    Finally, we copy every transition of $\Amc_1$ that leads to an accepting state of $\Amc_1$ such that it also leads to $q_2$, that is, we add the transitions $\{(p, a, \vbf \cdot 0^{d_2}, q_2) \mid (p,a,\vbf, q) \in \Delta_1, q \in F_1\}$. The semi-linear set $C$ of $\Amc$ is $C_1 \cdot \{0^{d_2}\} \cup \{0^{d_1}\} \cdot C_2$.
    As every accepting run of $\Amc_1$ is non-empty by assumption, $\Amc$ may guess the last transition of every accepting run of $\Amc_1$ and replace it with one of the new transitions that leads to $\Amc_2$ instead. As $q_2$ is accepting, the counters are reset, which justifies the choice of $C$.
    Now, if $\varepsilon \in L(\Amc_1)$, observe that $L(\Amc_1) \cdot SR_\omega(\Amc_2) = (L(\Amc_1) \setminus \{\varepsilon\} \cdot SR_\omega(\Amc_2)) \cup SR_\omega(\Amc_2)$. Hence, we may remove $\varepsilon$ from $L(\Amc_1)$ by replacing $q_1$ by a fresh non-accepting copy and use the closure under union.
    Hence, in any case only the copies of accepting states of $\Amc_2$ remain accepting; in particular no state of $\Amc_1$ is accepting in the corresponding copy of $\Amc$.

    The construction for reachability-regular~PA, limit~PA and Büchi~PA is very similar. The only difference is that we choose $C = C_1 \cdot C_2$ for the semi-linear set of $\Amc$, as counters are never reset here.

	Finally, for co-Büchi~PA and safety~PA we need the following modifications. First, we make every state of $\Amc$ accepting. Furthermore, we add one additional counter that is set to 1 on every transition from an accepting state of $\Amc_1$ to the initial state of $\Amc_2$ (note that every run of $\Amc$ can only take such a transition once). This counter is not modified on every other transition. The new semi-linear set is $C = (\Nbb^{d_1 + d_2} \cdot \{0\}) \cup (C_1 \cdot C_2 \cdot \{1\})$, that is, as long as we have not reached the state set of $\Amc_2$, we do not check the counter values yet.
\end{proof}

Before we continue, we show that we can normalize~PA (on finite words) such that the initial state is the only accepting state. This observation simplifies several proofs in this section.

\begin{lemma}\label{lem:normalized}
Let $\Amc = (Q, \Sigma, q_0, \Delta, F, C)$ be a~PA of dimension $d$. Then there exists an equivalent~PA~$\Amc'$ of dimension $d + 1$ with the following properties.
\begin{itemize}
    \item The initial state of $\Amc'$ is the only accepting state.
    \item $SCC(\Amc') = \{Q\}$.
\end{itemize}
We say that $\Amc'$ is \emph{normalized}.
\end{lemma}
\begin{proof}
    The normalized~PA $\Amc'$ is obtained from $\Amc$ by adding a fresh state $q_0'$, which is the initial state and only accepting state, and which inherits all outgoing transitions from $q_0$ and all incoming transitions from the accepting states. Furthermore, all transitions get a new counter, which is set to 0 except for the new incoming transitions of $q_0'$ where the counter is set to $1$, and all vectors in~$C$ are concatenated with $1$ (and we add the all zero-vector if we want to accept $\varepsilon$). 
    Finally, we remove all states that cannot reach $q'_0$ (such states can appear when shortcutting the incoming transitions of $F$, and are useless in the sense that their removal does not change the accepted language; however, this removal is necessary for the second property).

    Formally, we define $\Amc' = \{Q \cup \{q_0'\}, \Sigma, q_0', \Delta', \{q_0'\}, C')$, where 
     \begin{align*}
     \Delta' =&\ \{(p, a, \vbf \cdot 0, q) \mid (p, a, \vbf, q) \in \Delta\} \\
     \cup&\ \{(q_0', a, \vbf \cdot 0, q) \mid (q_0, a, \vbf, q) \in \Delta\} \\ \cup&\ \{(p, a, \vbf \cdot 1, q_0') \mid (p, a, \vbf, f) \in \Delta, f \in F\} \\
     \cup&\ \{(q_0', a, \vbf \cdot 1, q_0') \mid (q_0, a, \vbf, f) \in \Delta, f \in F\}.
     \end{align*}
     and 
     \[C' = \begin{cases}
         C \cdot \{1\} & \text{if } \varepsilon \not\in L(\Amc) \\
         C \cdot \{1\} \cup \{\0^{d+1}\} & \text{otherwise}.
     \end{cases}\]

	 It is now easily verified that $L(\Amc) = L(\Amc')$. 
\end{proof}

Observe that we have $SR_\omega(\Amc') = L(\Amc)^\omega$, that is, every normalized~PA interpreted as a reset~PA recognizes the $\omega$-closure of the language recognized by the~PA. 
As an immediate consequence we obtain the following corollary.
\begin{cor}
\label{cor:omegaclosure}
	For every Parikh recognizable language $L$ we have that $L^\omega$ is reset~PA recognizable.
\end{cor}

Combining these results we obtain that every $\omega$-language in $\LPAPA$, \ie, every $\omega$-language of the form $\bigcup_i U_i V_i^\omega$ is reset~PA recognizable. We show that the other direction does not hold, \ie, the inclusion is strict.

\begin{lemma}
\label{lem:papaSubsetReset}
The class $\LPAPA$ is a strict subclass of the class of reset~PA recognizable $\omega$-languages.
\end{lemma}
\begin{proof}
    The inclusion is a direct consequence of \Cref{cor:union}, \Cref{lem:concatenation}, and \Cref{cor:omegaclosure}. Hence it remains to show that the inclusion is strict.

    Consider the $\omega$-language $L =  \{a^n b^n \mid n \geq 1\}^\omega \cup \{a^n b^n \mid n \geq 1\}^* \cdot \{a\}^\omega$. This $\omega$-language is reset~PA recognizable, as witnessed by the strong reset~PA in \Cref{fig:resetpbacounterexample} with semi-linear set $C = \{(z,z) \mid z \in \Nbb\}$. 
    
\begin{figure}[h]
	\centering
	\begin{tikzpicture}[->,>=stealth',shorten >=1pt,auto,node distance=3cm, semithick]
	\tikzstyle{every state}=[minimum size=1.0cm]
	
	\node[initial, initial text = {}, state] (q0) {$q_0$};
	\node[state] (q1) [right of=q0] {$q_1$};
	\node[state, accepting] (q2) [right of=q1] {$q_2$};	
	\node[state, accepting] (q3) [right of=q2] {$q_3$};	

	\path
	(q0) edge [loop above] node {$a, \begin{pmatrix}1\\0\end{pmatrix}$} (q0)
	(q0) edge              node {$b, \begin{pmatrix}0\\1\end{pmatrix}$} (q1)
	(q0) edge [bend right = 60] node {$b, \begin{pmatrix}0\\1\end{pmatrix}$} (q2)
	(q1) edge [loop above] node {$b, \begin{pmatrix}0\\1\end{pmatrix}$} (q1)
	(q1) edge              node {$b, \begin{pmatrix}0\\1\end{pmatrix}$} (q2)
	(q2) edge [bend left = 70]  node {$a, \begin{pmatrix}1\\0\end{pmatrix}$} (q0)
	(q2) edge              node {$a, \begin{pmatrix}0\\0\end{pmatrix}$} (q3)	
	(q3) edge[loop above]  node {$a, \begin{pmatrix}0\\0\end{pmatrix}$} (q3)		
	;
	\end{tikzpicture}
	\caption{The strong reset~PA for $\{a^n b^n \mid n \geq 1\}^\omega \cup \{a^n b^n \mid n \geq 1\}^* \cdot $ $\{a\}^\omega$ with semi-linear set $C = \{(z,z) \mid z \in \Nbb\}$.}
	\label{fig:resetpbacounterexample}
\end{figure}

We claim that $L \notin \LPAPA$. 
Assume towards a contraction that $L \in \LPAPA$, \ie, there are Parikh recognizable languages $U_1, V_1, \dots, U_n, V_n$ such that $L = U_1 V_1^\omega \cup \dots \cup U_n V_n^\omega$. 
Then there is some $i \leq n$ such that for infinitely many $j \geq 1$ we have $\alpha_j = aba^2b^2 \dots a^jb^j \cdot a^\omega \in U_iV_i^\omega$. 
Then $V_i$ must contain a word of the form $v = a^k$, $k > 0$. 
Additionally, there cannot be a word in $V_i$ with infix~$b$. 
To see this assume for sake of contradiction that there is a word $w \in V_i$ with $\ell = |w|_b > 0$. Let  $\beta = (v^{\ell+1} w)^\omega$. 
Observe that~$\beta$ has an infix that consists of at least $\ell+1$ many $a$, followed by at most $\ell$, but at least one $b$, hence, no word of the form $u\beta$ with $u \in U_i$ is in $L$. This is a contradiction, thus $V_i \subseteq \{a\}^+$.

Since $U_i$ is Parikh recognizable, there is a~PA $\Amc_i$ with $L(\Amc_i) = U_i$. 
Let $m$ be the number of states in~$\Amc_i$ and $w' = aba^2b^2 \dots a^{m^4+1} b^{m^4+1}$. 
Then $w'$ is a prefix of a word accepted by~$\Amc_i$. Now consider the infixes $a^\ell b^\ell$ and the pairs of states $q_1,q_2$, where we start reading $a^\ell$ and end reading $a^\ell$, and $q_3,q_4$ where we start to read $b^\ell$ and end to read $b^\ell$, respectively. 
There are $m^2$ choices for the first pair and $m^2$ choices for the second pair, hence $m^4$ possibilities in total. 
Hence, as we have more than $m^4$ such infixes, there must be two with the same associated states $q_1,q_2,q_3,q_4$. 
Then we can swap these two infixes and get a word of the form $ab \dots a^rb^s \dots a^s b^r \dots a^{m^4+1} b^{m^4+1}$ that is a prefix of some word in $L(\Amc_i) = U_i$. 
But no word in $L$ has such a prefix, a contradiction. Thus, $U_1 V_1^\omega \cup \dots \cup U_nV_n^\omega \neq L$.
\end{proof}

\subsection{Characterization of Büchi Parikh Automata}
As mentioned in the last section, the class of $\omega$-languages recognized by Büchi~PA is a strict subset of $\LPAPA$, \ie, languages of the form $\bigcup_i U_i V_i^\omega$ for Parikh recognizable $U_i$ and~$V_i$. In this section we show that a restriction of the~PA recognizing the~$V_i$ is sufficient to exactly capture the expressiveness of Büchi~PA. To be precise, we show the following.

\begin{lem}
\label{lem:charBuchiPA}
    The following are equivalent for all $\omega$-languages $L \subseteq \Sigma^\omega$.
    \begin{enumerate}[(1)]
        \item $L$ is Büchi~PA recognizable.
        \item $L$ is of the form $\bigcup_i U_i V_i^\omega$, where $U_i \in \Sigma^*$ is Parikh recognizable and $V_i \in \Sigma^*$ is recognized by a normalized~PA where $C$ is a homogeneous linear set.
    \end{enumerate}
\end{lem}

We note that we can translate every~PA (with a linear set $C$) into an equivalent normalized~PA by \Cref{lem:normalized}. However, this construction adds a base vector, as we concatenate~$\{1\}$ to $C$. In fact, this can generally not be avoided without losing expressiveness. 
It turns out that this loss of expressiveness is exactly what we need to characterize the class of $\omega$-languages recognized by Büchi~PA as stated in the previous lemma. The main reason for this is pointed out in the following lemma.

\begin{lemma}
\label{lem:PArestrictionEqualsBuchi}
    Let $L$ be a language recognized by a (normalized)~PA $\Amc = (Q, \Sigma, q_0, \Delta, \{q_0\}, C)$ where $C$ is a homogeneous linear set. Then we have $B_\omega(\Amc) = L(\Amc)^\omega$.
\end{lemma}
\begin{proof}
 
 In this proof we assume that $C = C(\0, P)$ with $P = \{\pbf_1, \dots, \pbf_\ell\}$ for some $\ell \geq 0$.
 
 %$\Rightarrow$ To show $B_\omega(\Amc) \subseteq L(\Amc)^\omega$, 
 Let $\alpha \in B_\omega(\Amc)$ with accepting run $r = r_1r_2r_3 \dots$ where $r_i = (p_{i-1}, \alpha_i, \vbf_i, p_i)$. 
 As $r$ satisfies the Büchi condition and $\Amc$ is normalized there are infinitely many accepting hits, that is, infinitely many~$i$ such that $p_i = q_0$ and $\rho(r_1 \dots r_i) \in C$. 
 By Dickson's Lemma \cite{dickson}, there is an infinite monotone (sub)sequence of accepting hits $s_1 < s_2 < \dots$, \ie, for all $j > i$ we have $\rho(r_1 \dots r_{s_i}) = \pbf_1 z_1 + \dots + \pbf_\ell z_\ell$ for some $z_i \in \Nbb$ and $\rho(r_1 \dots r_{s_j}) = \pbf_1 z'_1 + \dots + \pbf_\ell z'_\ell$ for some $z'_i \in \Nbb$, and $z'_k \geq z_k$ for all $k \leq \ell$. 
 Hence, every infix $\alpha[s_i + 1, s_{i+1}]$ for $i \geq 0$ (assuming $s_0 = 0$) is accepted by $\Amc$.

% $\Leftarrow$ To show $L(\Amc)^\omega \subseteq B_\omega(\Amc)$, 
 Now let $w_1w_2 \dots \in L(\Amc)^\omega$ such that $w_i \in L(\Amc)$ for all $i \geq 1$. 
 Let $r^{(i)}$ be an accepting run of $\Amc$ on $w_i$. 
 Observe that for every $i \geq 1$ we have that $r^{(1)} \dots r^{(i)}$ is an accepting run of $\Amc$ on $w_1 \dots w_i$, as $C$ is a homogeneous linear set, and hence we have $\rho(r^{(1)} \dots r^{(i)}) = \rho(r^{(1)}) + \dots + \rho(r^{(i)}) \in C$. 
 Hence, the infinite sequence $r^{(1)} r^{(2)} \dots$ is a run of $\Amc$ on $w_1 w_2 \dots$ with infinitely many accepting hits. Hence $w_1 w_2 \dots \in B_\omega(\Amc)$.
\end{proof}

This is the main ingredient to prove \Cref{lem:charBuchiPA}. 
%Before we proceed to the proof, we need the following observation.
% \begin{observation}
% \label{obs:split}
%  Let $\Amc = (Q, \Sigma, q_0, \Delta, F, C)$ be a~PA, and write $C = \bigcup_i C_i$, where every~$C_i$ is linear. Then $L(\Amc) = \bigcup_i L(Q, \Sigma, q_0, \Delta, F, C_i)$.
% \end{observation}

% We are now ready to prove the main result of this subsection.

\begin{proof}[Proof of \Cref{lem:charBuchiPA}]

 We note that the proof in \cite{infiniteZimmermann} showing that every $\omega$-language~$L$ recognized by a Büchi-PA is of the form $\bigcup_i U_i V_i$ for~PA recognizable $U_i$ and $V_i$ already constructs~PA for the $V_i$ of the desired form. This shows the implication $(1) \Rightarrow (2)$.

 To show the implication $(2) \Rightarrow (1)$, we use that the $\omega$-closure of languages recognized by~PA of the stated form is Büchi~PA recognizable by \Cref{lem:PArestrictionEqualsBuchi}. As Büchi~PA are closed under left-concatenation with~PA recognizable languages (\Cref{lem:concatenation}) and union \cite{infiniteZimmermann}, the claim follows.
\end{proof}

\subsection{Characterization of $\LPAReg$}
In this section we characterize $\LPAReg$ by showing the following equivalences.
\begin{thm}
\label{thm:LimitEqualsReach}
 The following are equivalent for all $\omega$-languages $L \subseteq \Sigma^\omega$.
 \begin{enumerate}[(1)]
    \item $L$ is of the form $\bigcup_i U_i V_i^\omega$, where $U_i \in \Sigma^*$ is Parikh recognizable, and $V_i \subseteq \Sigma^*$ is regular.
    \item $L$ is limit~PA recognizable.
    \item $L$ is reachability-regular~PA recognizable.
\end{enumerate}
\end{thm}

Observe that in the first item we may assume that $L$ is of the form $\bigcup_i U_i V_i$, where \mbox{$U_i \in \Sigma^*$} is Parikh recognizable, and $V_i \subseteq \Sigma^\omega$ is $\omega$-regular. Then, by simple combinatorics and Büchi's theorem we have $\bigcup_i U_i V_i = \bigcup_i U_i (\bigcup_{j_i} X_{j_i} Y_{j_i}^\omega) = \bigcup_{i, j_i} U_i (X_{j_i} Y_{j_i}^\omega) = \bigcup_{i, j_i} (U_i X_{j_i}) Y_{j_i}^\omega$, for regular languages $X_{j_i}, Y_{j_i}$, where $U_i X_{j_i}$ is Parikh recognizable, as Parikh recognizable languages are closed under concatenation~\cite[Proposition~3]{cadilhac2013automates}. %\footnote{To the best of our knowledge there is no explicit construction for concatenation in the literature for~PA on finite words, however, a standard construction very similar to the one of \Cref{lem:concatenation} works. }

To simplify the proof, it is convenient to consider the following generalizations of Büchi automata. A \emph{transition-based generalized Büchi automaton} (TGBA) is a tuple $\Amc = (Q, \Sigma, q_0, \Delta, \Tmc)$ where $\Tmc \subseteq 2^\Delta$ is a collection of sets of transitions. Then a run $r_1 r_2 r_3 \dots$ of $\Amc$ is accepting if for all $T \in \Tmc$ there are infinitely many $i$ such that $r_i \in T$. It is well-known that TGBA have the same expressiveness as Büchi automata \cite{tgba}.

\Cref{thm:LimitEqualsReach} will be a direct consequence from the following lemmas. The first lemma shows the implication $(1) \Rightarrow (2)$.

\begin{lemma}
\label{lem:LPARegtoLimit}
 If $L \in \LPAReg$, then $L$ is limit~PA recognizable.   
\end{lemma}
\begin{proof}
    As the class of limit Parikh recognizable $\omega$-languages is closed under union by \Cref{cor:union}, it is sufficient to show how to construct a limit~PA for an $\omega$-language of the form $L = UV^\omega$, where $U$ is Parikh recognizable and $V$ is regular.

    Let $\Amc_1 = (Q_1, \Sigma, q_1, \Delta_1, F_1, C)$ be a~PA with $L(\Amc_1) = U$ and $\Amc_2 = (Q_2, \Sigma, q_2, \Delta_2, F_2)$ be a Büchi automaton with $L_\omega(\Amc_2) = V^\omega$. We use the following standard construction for concatenation. Let $\Amc = (Q_1 \cup Q_2, \Sigma, q_1, \Delta, F_2, C)$ be a limit~PA where 
    \[\Delta = \Delta_1 \cup \{(p, a, \0, q) \mid (p, a, q) \in \Delta_2\} \cup \{(f, a, \0, q) \mid (q_2, a, q) \in \Delta_2, f \in F_1\}.\]
    We claim that $L_\omega(\Amc) = L$.

	To show $L_\omega(\Amc) \subseteq L$, let $\alpha \in L_\omega(\Amc)$ with accepting run $r_1 r_2 r_3 \dots$ where $r_i = (p_{i-1}, \alpha_i, \vbf_i, p_i)$. As only the states in $F_2$ are accepting, there is a position $j$ such that $p_{j-1} \in F_1$ and $p_j \in Q_2$. 
    In particular, all transitions of the copy of $\Amc_2$ are labeled with~$\0$, \ie, $\vbf_i = \0$ for all $i \geq j$. Hence $\rho(r) = \rho(r_1 \dots r_{j-1}) \in C$ (in particular, there is no $\infty$ value in $\rho(r)$).
    We observe that $r_1 \dots r_{j-1}$ is an accepting run of $\Amc_1$ on $\alpha[1,j-1]$, as $p_{j-1} \in F_1$ and $\rho(r_1 \dots r_{j-1}) \in C$.
    For all $i \geq j$ let $r'_i = (p_{i-1}, \alpha_i, p_i)$. Observe that $(q_2, \alpha_j, p_j)r'_{j+1} r'_{j+2} \dots$ is an accepting run of $\Amc_2$ on $\alpha_j \alpha_{j+1} \alpha_{j+2} \dots$, hence $\alpha \in L(\Amc_1) \cdot L_\omega(\Amc_2) = L$.

	To show $L = UV^\omega \subseteq L_\omega(\Amc)$, let $w \in L(\Amc_1)=U$ with accepting run $s$, and \mbox{$\alpha \in L_\omega(\Amc_2)=V^\omega$} with accepting run $r = r_1 r_2 r_3 \dots$, where $r_i = (p_{i-1}, \alpha_1, p_i)$.
    Observe that~$s$ is also a partial run of $\Amc$ on $w$, ending in an accepting state $f$. By definition of~$\Delta$, we can continue the run $s$ in $\Amc$ basically as in $r$. To be precise, let $r'_1 = (f, \alpha_1, \0, p_1)$, and, for all $i > 1$ let $r'_i = (p_{i-1}, \alpha_i, \0, p_i)$. Then $s r'_1 r'_2 r'_3 \dots$ is an accepting run of $\Amc$ on $w \alpha$, hence $w \alpha \in L_\omega(\Amc)$.
\end{proof}

Observe that the construction in the proof of the lemma works the same way when we interpret $\Amc$ as a reachability-regular~PA (every visit of an accepting state has the same good counter value; this argument is even true if we interpret $\Amc$ as a Büchi~PA), showing the implication~$(1) \Rightarrow (3)$.
\begin{cor}
\label{cor:LPARegToReachReg}
If $L \in \LPAReg$, then $L$ is reachability-regular.
\end{cor}

For the backwards direction we need an auxiliary lemma, essentially stating that semi-linear sets over $C \subseteq (\Nbb \cup \{\infty\})^d$ can be modified such that $\infty$-entries in vectors in~$C$ are replaced by arbitrary integers, and remain semi-linear.

\begin{lemma}
    \label{lem:semi-linear-inf}
    Let $C \subseteq (\Nbb \cup \{\infty\})^d$ be semi-linear and $D \subseteq \{1, \dots, d\}$. Let $C_D \subseteq \Nbb^d$ be the set obtained from $C$ as follows.
    \begin{enumerate}
        \item Remove every vector $\vbf = (v_1, \dots, v_d)$ where $v_i = \infty$ for an $i \notin D$.
        \item As long as $C_D$ contains a vector $\vbf = (v_1, \dots, v_d)$ with $v_i = \infty$ for an $i \leq d$: replace $\vbf$ by all vectors of the form $(v_1, \dots v_{i-1}, z, v_{i+1}, \dots, v_d)$ for $z \in \Nbb$.
    \end{enumerate}
    Then $C_D$ is semi-linear. Furthermore, $C_D$ can be computed in polynomial time.
\end{lemma}
\begin{proof}
For a vector $\vbf = (v_1, \dots, v_d) \in (\Nbb \cup \{\infty\})^d$, let $\Inf(\vbf) = \{i \mid v_i = \infty\}$ denote the positions of $\infty$-entries in $\vbf$. 
Furthermore, let $\bar\vbf = (\bar{v}_1, \dots, \bar{v}_d)$ denote the vector obtained from $\vbf$ by replacing every $\infty$-entry by 0, \ie, $\bar{v}_i = 0$ if $v_i = \infty$, and $\bar{v}_i = v_i$ otherwise.

We carry out the following procedure for every linear set of the semi-linear set independently, hence we assume that $C = (\bbf, P)$ with is linear. 
We also assume that there is no $\pbf \in P$ with $\Inf(\pbf) \not\subseteq D$, otherwise, we simply remove it. 

Now, if $\Inf(\bbf) \not\subseteq D$, then $C_D = \varnothing$, as this implies that every vector in $C$ has an $\infty$-entry at an unwanted position (the first item of the lemma).

Otherwise, $C_D = C(\bbf, \bigcup_{\pbf \in P} \left(\{\bar{\pbf}\} \cup \{\ebf_i \mid i \in \Inf(\pbf)\})\right)$, which is linear by definition and computable in polynomial time.
\end{proof}

We are now ready to prove the following lemma, showing the implication $(2) \Rightarrow (1)$.
\begin{lemma}
\label{lem:limitToLPAReg}
If $L$ is limit~PA recognizable, then $L \in \LPAReg$. 
\end{lemma}
\begin{proof}
    Let $\Amc = (Q, \Sigma, q_0, \Delta, F, C)$ be an limit~PA of dimension $d$. The idea is as follows. We guess a subset $D \subseteq \{1, \dots, d\}$ of counters whose values we expect to be $\infty$.
    Observe that every counter not in $D$ has a finite value, hence for every such counter there is a point where all transitions do not increment the counter further.
    For every subset $D \subseteq \{1, \dots, d\}$ we decompose $\Amc$ into a~PA and a TGBA. In the first step we construct a~PA where every counter not in $D$ reaches its final value and is verified. In the second step we construct a TGBA ensuring that for every counter in $D$ at least one transition adding a non-zero value to that counter is used infinitely often. This can be encoded directly into the TGBA. Furthermore we delete all transitions that modify counters not in $D$. 

    Fix $D \subseteq \{1, \dots, d\}$ and $f \in F$, and define the~PA $\Amc^D_f = (Q, \Sigma, q_0, \Delta, \{f\}, C_D)$ where~$C_D$ is defined as in \Cref{lem:semi-linear-inf}.
    Furthermore, we define the TGBA $\Bmc^D_f = (Q, \Sigma, f, \Delta^D, \Tmc^D)$ where~$\Delta^D$ contains the subset of transitions of $\Delta$ where the counters not in $D$ have zero-values (just the transitions without vectors for the counters, as we construct a TGBA). On the other hand, for every counter $i$ in $D$ there is one acceptance component in~$\Tmc^D$ that contains exactly those transitions (again without vectors) where the $i$th counter has a non-zero value. Finally, we encode the condition that at least one accepting state in $F$ needs to by seen infinitely often in~$\Tmc^D$ by further adding the component $\{(p, a, q) \in \Delta \mid q \in F\}$ (\ie, now we need to see an incoming transition of a state in $F$ infinitely often).

    We claim that $L_\omega(\Amc) = \bigcup_{D \subseteq \{1, \dots, d\}, f \in F} L(\Amc^D_f)\cdot L_\omega(\Bmc^D_f)$, which by the comment below \Cref{thm:LimitEqualsReach} and the equivalence of TGBA and Büchi automata implies the statement of the lemma. 

    To show $L_\omega(\Amc) \subseteq \bigcup_{D \subseteq \{1, \dots, d\}, f \in F} L(\Amc^D_f) \cdot L_\omega(\Bmc^D_f)$, let $\alpha \in L_\omega(\Amc)$ with accepting run $r_1 r_2 r_3 \dots$ where $r_i = (p_{i-1}, \alpha_i, \vbf_i, p_i)$. Let $D$ be the positions of $\infty$-entries in $\rho(r) = (v_1, \dots, v_d)$. As the $v_i$ with $i \notin D$ have integer values, there is a position $j$ such that in all $\vbf_k$ for $k \geq j$ the $i$-th entry of $\vbf_k$ is 0. Let $\ell \geq j$ be minimal such that $p_\ell$ in $F$. We split $\alpha = w \beta$, where $w = \alpha[1,\ell]$, and $\beta = \alpha_{\ell + 1} \alpha_{\ell +2} \dots$.
    
    First we argue that $w \in L_\omega(\Amc^D_{p_{\ell}})$. Observe that $\Amc^D_{p_{\ell}}$ inherits all transitions from $\Amc$, hence $r_1 \dots r_{\ell}$ is a run of $\Amc^D_{p_{\ell}}$ on $w$. As $p_\ell$ is accepting by definition, it remains to show that $\rho(r_1 \dots r_\ell) \in C_D$.
    By the choice of $\ell$, all counters not in $D$ have reached their final values. As $C_D$ contains all vectors of $C$ where all $\infty$-entries are replaced by arbitrary values, the claim follows, hence $w \in L(\Amc^D_{p_{\ell}})$.

    Now we argue that $\beta \in L_\omega(\Bmc^D_{p_\ell})$. For every $k > \ell$ define $r'_k = (p_{k-1}, \alpha_k, p_k)$. Observe that $r' = r'_{k+1} r'_{k+2} \dots$ is a run of $\Bmc^D_{p_\ell}$ on $\beta$ (all $r'_{k+1}$ exist in $\Bmc^D_{p_\ell}$, as the counters not in~$D$ of all transitions $r_k$ have zero-values by the definition of $\ell$).
    It remains to show that~$r'$ is accepting, \ie, that for every counter in $D$ at least one transition with a non-zero value is used infinitely often, and an accepting state is visited infinitely often. This is the case, as these counter values are $\infty$ in $\rho(r)$ and by the acceptance condition of limit~PA, hence $\beta \in L_\omega(\Bmc^D_{p_\ell})$.
    We conclude $\alpha \in \bigcup_{D \subseteq \{1, \dots, d\}, f \in F} L(\Amc^D_f) \cdot L_\omega(\Bmc^D_f)$.

	To show $\bigcup_{D \subseteq \{1, \dots, d\}, f \in F} L(\Amc^D_f) \cdot L_\omega(\Bmc^D_f) \subseteq L_\omega(\Amc)$, let $w \in L(\Amc^D_f)$ and $\beta \in L_\omega(\Bmc^D_f)$ for some $D \subseteq \{1, \dots, d\}$ and $f \in F$. We show that $w\beta \in L_\omega(\Amc)$.

    Let $s$ be an accepting run of $\Amc^D_f$ on $w$, which ends in the accepting state $f$ by definition. Let $\rho(s) = (v_1, \dots, v_d)$. By definition of $C_D$, there is a vector $\ubf = (u_1, \dots, u_d)$ in $C$ where $u_i = \infty$ if $i \in D$, and $u_i = v_i$ if $i \notin D$.
    Furthermore, let $r = r_1r_2r_3\dots$, where $r_i = (p_{i-1}, \alpha_i, p_i)$, be an accepting run of $\Bmc^D_f$ on $\beta$, which starts in the accepting state~$f$ by definition. 
    By definition of $\Tmc^d$, for every counter $i \in D$ at least one transition where the $i$-th counter of the corresponding transition in $\Delta$ is non-zero is used infinitely often. 
    Hence, let $r' = r'_1 r'_2 r'_3 \dots$ where $r'_i = (p_{i-1}, \alpha_i, \vbf_i, p_i)$ for a suitable vector $\vbf_i$. 
    Furthermore, the labels of transitions of counters not in $D$ have a value of zero, hence $\rho(r') = (x_1, \dots, x_d)$, where $x_i = \infty$ if $i \in D$, and $x_i = 0$ if $i \notin D$. 
    A technical remark: it might be the case that there are more than one transitions in $\Delta$ that collapse to the same transition in $\Delta^D$, say $\delta_1 = (p, a, \ubf, q)$ and $\delta_2 = (p, a, \vbf, q)$ appear in $\Delta$ and collapse to $(p, a, q)$ in $\Delta^D$. If both transitions, $\delta_1$ and~$\delta_2$, are seen infinitely often, we need to take care that we also see both infinitely often when translating the run $r$ back. This is possible using a round-robin procedure.

    Now observe that $sr'$ is a run of $\Amc$ on $w\beta$ (recall that $s$ ends in $f$, and $r'$ starts in $f$). Furthermore, we have $\rho(sr') = \rho(s) + \rho(r') = (v_1 + x_1, \dots, v_d + x_d)$, where $v_i + x_i = \infty$ if $i \in D$, and $v_i + x_i = v_i$ if $i \notin D$ by the observations above. Hence $\rho(sr') \in C$. Finally, $\Tmc^D$ enforces that at least one accepting state in $\Bmc^D_f$ is seen infinitely often, hence $w\beta \in L_\omega(\Amc)$.
\end{proof}

Observe that the construction in \Cref{lem:LPARegtoLimit} yields a limit~PA whose semi-linear set $C$ contains no vector with an $\infty$-entry. Hence, by this observation and the construction in the previous lemma we obtain the following corollary.
\begin{cor}
    For every limit~PA there is an equivalent limit~PA whose semi-linear set does not contain any $\infty$-entries.
\end{cor}

Finally we show the implication $(3) \Rightarrow (1)$.
\begin{lemma}
    If $L$ is reachability-regular, then $L \in \LPAReg$.
\end{lemma}
\begin{proof}
    Let $\Amc = (Q, \Sigma, q_0, \Delta, F, C)$ be a reachability-regular~PA. The intuition is as follows. a reachability-regular~PA just needs to verify the counters a single time. Hence, we can recognize the prefixes of infinite words $\alpha \in B_\omega(\Amc)$ that generate the accepting hit with a~PA. Further checking that an accepting state is seen infinitely often can be done with a Büchi automaton.

    Fix $f \in F$ and let $\Amc_f = (Q, \Sigma, q_0, \Delta, \{f\}, C)$ be the~PA that is, syntactically equal to $\Amc$ with the only difference that $f$ is the only accepting state. Similarly, let $\Bmc_f = (Q, \Sigma, f, \{(p,a,q) \mid (p,a,\vbf, q) \in \Delta\}, F)$ be the Büchi automaton obtained from $\Amc$ by setting~$f$ as the initial state and the forgetting the vector labels. 
    
    We claim that $RR_\omega(\Amc) = \bigcup_{f \in F} L(\Amc_f) \cdot L_\omega(\Bmc_f)$.

	To show $RR_\omega(\Amc) \subseteq \bigcup_{f \in F} L(\Amc_f) \cdot L_\omega(\Bmc_f)$, let $\alpha \in B_\omega(\Amc)$ with accepting run $r = r_1 r_2 r_3 \dots$ where $r_i = (p_{i-1}, \alpha_i, \vbf_i, p_i)$. Let $k$ be arbitrary such that there is an accepting hit in $r_k$ (such a~$k$ exists by definition) and consider the prefix $\alpha[1,k]$. Obviously $r_1 \dots r_k$ is an accepting run of $\Amc_{p_k}$ on $\alpha[1,k]$.
    Furthermore, there are infinitely many $j$ such that $p_j \in F$ by definition. In particular, there are also infinitely many $j \geq k$ with this property. Let $r'_i = (p_{i-1}, \alpha_i, p_i)$ for all $i > k$.  
    Then $r'_{k+1} r'_{k+2} \dots$ is an accepting run of $\Bmc_{p_k}$ on $\alpha_{k+1} \alpha_{k+2}\dots$ (recall that $p_k$ is the initial state of $\Bmc_{p_k}$). Hence we have $\alpha[1,k] \in L(\Amc_{p_k})$ and $\alpha_{k+1} \alpha_{k+2} \dots \in L_\omega(\Bmc_{p_k})$.

	To show $\bigcup_{f \in F} L(\Amc_f) \cdot L_\omega(\Bmc_f) \subseteq RR_\omega(\Amc)$, let $w \in L(\Amc_f)$ and $\beta \in L_\omega(\Bmc_f)$ for some $f \in F$. We show $w\beta \in B_\omega(\Amc)$.
    Let $s = s_1 \dots s_n$ be an accepting run of $\Amc_f$ on $w$, which ends in the accepting state $f$ with $\rho(s) \in C$ by definition.
    Furthermore, let $r = r_1 r_2 r_3 \dots$ be an accepting run of $\Bmc^D_f$ on $\beta$ which starts in the accepting state $f$ by definition. It is now easily verified that $sr'$ with $r' = r'_1r'_2r'_3\dots$ where $r'_i = (p_{i-1}, \alpha_i, \vbf_i, p_i)$ (for an arbitrary $\vbf_i$ such that $r'_i \in \Delta)$ is an accepting run of $\Amc$ on $w\beta$, as there is an accepting hit in $s_n$, and the (infinitely many) visits of an accepting state in $r$ translate one-to-one, hence $w\beta \in B_\omega(\Amc)$.
\end{proof}

As shown in \Cref{lem:charBuchiPA}, the class of Büchi~PA recognizable $\omega$-languages is equivalent to the class of $\omega$-languages of the form $\bigcup_i U_i V_i^\omega$ where $U_i$ and $V_i$ are Parikh recognizable, but the~PA for $V_i$ is restricted in such a way that the initial state is the only accepting state and the set is a homogeneous linear set.
Observe that for every regular language $L$ there is a Büchi automaton $\Amc$ where the initial state is the only accepting state with $L_\omega(\Amc) = L^\omega$ (see \eg, \cite[Lemma 1.2]{thomasinfinite}). Hence, $\LPAReg$ is a subset of the class of Büchi~PA recognizable $\omega$-languages. This inclusion is also strict, as witnessed by the Büchi~PA in \Cref{ex} which has the mentioned property.
%Observe that the~PA in \Cref{ex} has exactly this form and recognizes a non-regular language, which means that this restriction still captures a strict superset of the class of regular languages. Hence we obtain the following corollary as a consequence of \Cref{thm:LimitEqualsReach}.

%As a corollary we get the following strict inclusion. \textcolor{red}{todo. Argue why the corollary is true.}
\begin{cor}
The class $\LPAReg$ is a strict subclass of the class of Büchi~PA recognizable $\omega$-languages.
\end{cor}

We finish this section by observing that reachability~PA capture a subclass of $\LPAReg$ where, due to completeness, all $V_i = \Sigma$.
\begin{observation}
\label{obs:reachability}
    The following are equivalent for all $\omega$-languages $L \subseteq \Sigma^\omega$.
    \begin{enumerate}[(1)]
        \item $L$ is of the form $\bigcup_i U_i \Sigma^\omega$ where $U_i \subseteq \Sigma^*$ is Parikh recognizable.
        \item $L$ is reachability~PA recognizable.
    \end{enumerate}
\end{observation}

\subsection{Characterization of \LPAPA and \LRegPA}
In this section we give a characterization of $\LPAPA$ and a characterization of $\LRegPA$. 
As mentioned in the beginning of this section, %
%Grobler et al.~\cite{infiniteOurs} have shown that $\LPAPA \subsetneq \LSPBA$, \ie, 
reset~PA are too strong to capture this class. However, restrictions of strong reset~PA are good candidates to capture $\LPAPA$ as well as~$\LRegPA$. 
In fact we show that it is sufficient to restrict the appearances of accepting states to capture $\LPAPA$, as specified by the first theorem of this section. Further restricting the vectors yields a model capturing $\LRegPA$, as specified in the second theorem of this section. Recall that the condensation of $\Amc$ is the DAG of strong components of the underlying graph of~$\Amc$. 

\begin{theorem}
\label{thm:lpapa}
    The following are equivalent for all $\omega$-languages $L \subseteq \Sigma^\omega$.
    \begin{enumerate}[(1)]
        \item $L \in \LPAPA$.
        \item $L$ is recognized by a strong reset~PA $\Amc$ with the property that accepting states appear only in the leaves of the condensation of $\Amc$, and there is at most one accepting state per leaf.
    \end{enumerate}
\end{theorem}
\begin{proof}
To show the implication $(1) \Rightarrow (2)$, let $\Amc_i = (Q_i, \Sigma, q_i, \Delta_i, F_i)$ for $i \in \{1,2\}$ be~PA and let $L = L(\Amc_1) \cdot L(\Amc_2)^\omega$. By \Cref{lem:normalized} we may assume that $\Amc_2$ is normalized (recall that by \Cref{cor:omegaclosure} this implies $SR_\omega(\Amc_2) = L(\Amc_2)^\omega$) and hence write $L = L(\Amc_1) \cdot SR_\omega(\Amc_2)$. 
As pointed out in the proof of \Cref{lem:concatenation}, we can construct a reset~PA $\Amc$ that recognizes~$L$ such that only the accepting states of $\Amc_2$ remain accepting in $\Amc$. As $\Amc_2$ is normalized, this means that only $q_2$ is accepting in~$\Amc$. Hence~$\Amc$ satisfies the property of the theorem.
Finally observe that the construction in \Cref{cor:union} maintains this property, implying that the construction presented in \Cref{lem:papaSubsetReset} always yields a reset~PA of the desired form.

\medskip
To show the implication $(2) \Rightarrow (1)$, let $\Amc = (Q, \Sigma, q_0, \Delta, F, C)$ be a strong reset~PA of dimension $d$ with the mentioned property.
Let $f \in F$ and let $\Amc_f = (Q, \Sigma, q_0, \Delta_f, \{f\}, C \cdot \{1\})$ with $\Delta_{f} = \{p,a,\vbf \cdot 0,q) \mid (p,a,\vbf, q) \in \Delta, q \neq f\} \cup \{(p, a, \vbf \cdot 1, f) \mid (p, a, \vbf, f) \in \Delta\}$ be the~PA of dimension $d+1$ obtained from $\Amc$ by setting $f$ as the only accepting state with an additional counter that is 0 at every transition except the incoming transitions of $f$, where the counter is set to 1. 
Additionally all vectors in $C$ are concatenated with $1$. Similarly, let $\Amc_{f,f} = (Q, \Sigma, f, \Delta_{f}, \{f\}, C \cdot \{1\})$ be the~PA of dimension $d+1$ obtained from $\Amc_f$ by setting~$f$ as the initial state. 
% and only accepting state, where $\Delta_f$ is defined as for $\Amc_f$. 
% We claim $SR_\omega(\Amc) = \bigcup_{f \in F} L(\Amc_{f}) \cdot L(\Amc_{f,f})^\omega$.

To show $SR_\omega(\Amc) \subseteq \bigcup_{f \in F} L(\Amc_{f}) \cdot L(\Amc_{f,f})^\omega$, let $\alpha \in S_\omega(\Amc)$ with accepting run $r = r_1 r_2 r_3 \dots$ where $r_i = (p_{i-1}, \alpha_i, \vbf_i, p_i)$. Let $k_1 < k_2 < \dots$ be the positions of accepting states in $r$, \ie, $p_{k_i} \in F$ for all $i \geq 1$.
First observe that the property in the theorem implies $p_{k_i} = p_{k_j}$ for all $i, j \geq 1$, \ie, no two distinct accepting states appear in~$r$, since accepting states appear only in different leaves of the condensation of $\Amc$. 
%However, it is not possible to reach $p_{k_i}$ from $p_{k_j}$ (or vice versa) by the definition of SCC and leaf. This is a contradiction, hence $p_{k_i} = p_{k_j}$ for all $i,j \geq 1$.

For $j \geq 1$ define $r'_j = (p_{j-1}, \alpha_j, \vbf_j \cdot 0, p_j)$ if $j \neq k_i$ for all $i \geq 1$, and $r'_j = (p_{j-1}, \alpha_j, \vbf_j \cdot 1, p_j)$ if $j = k_i$ for some $i \geq 1$, \ie, we replace every transition $r_j$ by the matching transition in~$\Delta_f$.

Now consider the partial run $r_1 \dots r_{k_1}$ and observe that $p_i \neq p_{k_1}$ for all $i < k_1$, and $\rho(r_1 \dots r_{k_1}) \in C$ by the definition of strong reset~PA. Hence $r' = r'_1 \dots r'_{k_1}$ is an accepting run of~$\Amc_{p_{k_1}}$ on $\alpha[1, k_1]$, as only a single accepting state appears in $r'$, the newly introduced counter has a value of $1$ when entering $p_{k_1}$, \ie, $\rho(r') \in C \cdot \{1\}$, hence $\alpha[1, k_1] \in L(\Amc_{p_{k_1}})$.

Finally, we show that $\alpha[k_i + 1, k_{i+1}] \in L(\Amc_{p_{k_1},p_{k_1}})$.
Observe that $r'_{k_i + 1} \dots r'_{k_{i+1}}$ is an accepting run of $\Amc_{p_{k_1},p_{k_1}}$ on $\alpha[k_i + 1, k_{i+1}]$: we have $\rho(r_{k_i + 1} \dots r_{k_{i+1}}) = \vbf \in C$ by definition. Again, as only a single accepting state appears in $r'_{k_i + 1} \dots r'_{k_{i+1}}$, we have $\rho(r'_{k_i + 1} \dots r'_{k_{i+1}}) = \vbf \cdot 1 \in C \cdot \{1\}$, and hence $\alpha[k_i + 1, k_{i+1}] \in L(\Amc_{p_{k_1},p_{k_1}})$. We conclude $\alpha \in L(\Amc_{p_{k_1}}) \cdot L(\Amc_{p_{k_1}, p_{k_1}})^\omega$.

To show $\bigcup_{f \in F} L(\Amc_{f}) \cdot L(\Amc_{f,f})^\omega \subseteq SR_\omega(\Amc)$, let $u \in L(\Amc_{f})$, and $v_1, v_2, \dots \in L(\Amc_{f,f})$ for some $f \in F$. We show that $uv_1v_2 \dots \in SR_\omega(\Amc)$.

First let $u = u_1 \dots u_n$ and $r' = r'_1 \dots r'_n$ with $r'_i = (p_{i-1}, u_i, \vbf_i \cdot c_i, p_i)$, where $c_i \in \{0,1\}$, be an accepting run of $\Amc_{f}$ on $u$. 
Observe that $\rho(r') \in C \cdot \{1\}$, hence $\sum_{i \leq n} c_i = 1$, \ie,~$p_n$ is the only occurrence of an accepting state in $r'$ (if there was another, say $p_j$, then $c_j = 1$ by the choice of $\Delta_f$, hence $\sum_{i \leq n} c_i > 1$, a contradiction).
For all  $1 \leq i \leq n$ let $r_i = (p_{i-1}, u_i, \vbf_i, p_i)$. Then $r_1 \dots r_n$ is a partial run of $\Amc$ on $w$ with $\rho(r_1 \dots r_n) \in C$ and $p_n = f$.

Similarly, no run of $\Amc_{f,f}$ on any $v_i$ visits an accepting state before reading the last symbol, hence we continue the run from $r_n$ on $v_1, v_2, \dots$ using the same argument. Hence $uv_1v_2 \dots \in SR_\omega(\Amc)$, concluding the proof.
\end{proof}

As a side product of the proof of \Cref{thm:lpapa} we get the following corollary, which is in general not true for arbitrary reset~PA.
\begin{cor}
Let $\Amc = (Q, \Sigma, q_0, \Delta, F, C)$ be a strong reset~PA with the property that accepting states appear only in the leaves of the condensation of $\Amc$, and there is at most one accepting state per leaf. 
Then we have $SR_\omega(\Amc) = \bigcup_{f \in F} S_\omega(Q, \Sigma, q_0, \Delta, \{f\}, C)$.
\end{cor}

By even further restricting the power of strong reset~PA, we get the following characterization of~$\LRegPA$.

\begin{theorem}
\label{thm:lregpa}
    The following are equivalent for all $\omega$-languages $L \subseteq \Sigma^\omega$.
    \begin{enumerate}[(1)]
        \item $L$ is of the form $\bigcup_i U_i V_i^\omega$, where $U_i \subseteq \Sigma^*$ is regular and $V_i \subseteq \Sigma^*$ is Parikh recognizable.
        \item $L$ is recognized by a strong reset~PA $\Amc$ with the following properties.
        \begin{enumerate}[(a)]
            \item At most one state $q$ per leaf of the condensation of $\Amc$ may have incoming transitions from outside the leaf, this state $q$ is the only accepting state in the leaf, and there are no accepting states in non-leaves.
            \item only transitions connecting states in a leaf may be labeled with a non-zero vector. 
        \end{enumerate}
    \end{enumerate}
\end{theorem}

Observe that property (a) is a stronger property than the one of \Cref{thm:lpapa}, hence, strong reset~PA with this restriction are at most as powerful as those that characterize~$\LPAPA$. However, as a side product of the proof we get that property (a) is equivalent to the property of \Cref{thm:lpapa}. Hence, property (b) is necessary to sufficiently weaken strong reset~PA such that they capture $\LRegPA$. 
In fact, using the notion of normalization, we can re-use most of the ideas in the proof of \Cref{thm:lpapa}.

\begin{proof}[Proof of \Cref{thm:lregpa}]

We can trivially convert an NFA into an equivalent~PA by labeling every transition with $0$ and choosing $C = \{0\}$, showing the implication $(1) \Rightarrow (2)$. 
Let $\Amc$ be an arbitrary~PA and assume that it is normalized; in particular implying that it is only a single SCC. 
Again, we have $L(\Amc)^\omega = S_\omega(\Amc)$ and the constructions for concatenation and union preserve the properties, hence, we obtain a strong reset~PA of the desired form.

\medskip
To show the implication $(2) \Rightarrow (1)$, let $\Amc = (Q, \Sigma, q_0, \Delta, F, C)$ be a strong reset~PA of dimension $d$ with properties~(a) and (b). Fix $f \in F$ and let 
\[\Bmc_f = (Q_f, \Sigma, q_0, \{(p, a, q) \mid (p, a, \vbf, q) \in \Delta, p,q \in Q_f\}, \{f\})\] with $Q_f = \{q \in Q \mid q \text{ appears in a non-leaf SCC of } C(\Amc)\} \cup \{f\}$
be the NFA obtained from $\Amc$ by removing all leaf states except $f$, and removing all labels from the transitions.
Recycling the automaton from \Cref{thm:lpapa}, let \mbox{$\Amc_{f,f} = (Q, \Sigma, f, \Delta_{f}, \{f\}, C \cdot \{1\})$} with $\Delta_{f} = \{(p,a,\vbf \cdot 0,q) \mid (p,a,\vbf, q) \in \Delta, q \neq f\} \cup \{(p, a, \vbf \cdot 1, f) \mid (p, a, \vbf, f) \in \Delta\}$.
We claim $SR_\omega(\Amc) = \bigcup_{f \in F} L(\Bmc_f) \cdot L(\Amc_{f,f})^\omega$.

To show $SR_\omega(\Amc) \subseteq \bigcup_{f \in F} L(\Bmc_f) \cdot L(\Amc_{f,f})^\omega$, let $\alpha \in SR_\omega(\Amc)$ with accepting run $r = r_1r_2r_3 \dots$ where $r_i = (p_{i-1}, \alpha_i, \vbf_i, p_i)$, and let $k_1< k_2< \dots$ be the positions of the accepting states in $r$, and consider the partial run $r_1 \dots r_{k_1}$ (if $k_1 = 0$, \ie, the initial state is already accepting, then $r_1 \dots r_{k_1}$ is empty). 

By property (a) we have that $p_{k_1}$ is the first state visited in $r$ that is, located in a leaf of~$C(\Amc)$. Hence $r'_1 \dots r'_{k_1}$, where $r'_i = (p_{i-1}, \alpha_i, p_i)$, is an accepting run of $\Bmc_{p_{k_1}}$ on $\alpha[1, k_1]$ (in the case $k_1 = 0$ we define $\alpha[1, k_1] = \varepsilon$).

By the same argument as in the proof of \Cref{thm:lpapa} we have $p_{k_i} = p_{k_j}$ for all $i,j \geq 1$, hence $\alpha[k_i + 1, k_{i+1}] \in L(\Amc_{p_{k_1}, p_{k_1}})$, and hence $\alpha \in L(\Bmc_{p_k}) \cdot L(\Amc_{p_{k_1}, p_{k_1}})^\omega$.

To show $\bigcup_{f \in F} L(\Amc_{f}) \cdot L(\Amc_{f,f})^\omega \subseteq SR_\omega(\Amc)$, let $u \in L(\Bmc_{f})$, and $v_1, v_2, \dots \in L(\Amc_{f,f})$ for some $f \in F$. We show that $uv_1v_2 \dots \in S_\omega(\Amc)$.

First observe that properties (a) and (b) enforce that $\0 \in C$, as the accepting state of a leaf of~$C(\Amc)$ is visited before a transition labeled with a non-zero can be used.
Let $u = u_1 \dots u_n$ and $s_1 \dots s_n$ with $s_i = (p_{i_1}, u_i, p_i)$ be an accepting run of $\Bmc_f$ on $u$. Define $s'_i = (p_{i_1}, u_i, \0, p_i)$ and observe that $s'_1 \dots s'_n$ is a partial run of $\Amc$ with $\rho(s'_1 \dots s'_n) \in C$ and $p_n = f$ by the observation above.
Again we can very similarly continue the run on $v_1, v_2, \dots$ using the same argument. Hence $uv_1v_2 \dots \in SR_\omega(\Amc)$, concluding the proof.
\end{proof}

\subsection{Blind Counter Automata}

Blind counter automata on infinite words were first studied by Fernau and Stiebe \cite{blindcounter}. 
In this section we first recall the definition of blind counter automata as introduced by Fernau and Stiebe \cite{blindcounter}. 
The definition of these automata admits $\varepsilon$-transitions. 
It is easily observed that Büchi~PA with $\varepsilon$-transitions are equivalent to blind counter automata.
Therefore, we extend all models studied in this paper with $\epsilon$-transitions and consider the natural question whether they admit $\varepsilon$-elimination.
We show that almost all models allow $\varepsilon$-elimination, the exception being safety and co-Büchi~PA. 
As it turns out, co-Büchi~PA with $\varepsilon$-transitions are powerful enough to simulate Büchi~PA, and safety~PA with $\varepsilon$-transitions are powerful enough to simulate reset~PA. This is a stark contrast to their variants without $\varepsilon$-transitions, as they do not even recognize all $\omega$-regular languages.
%For the latter two models we observe that $\varepsilon$-transitions allow to encode $\omega$-regular conditions, meaning that such transitions give the models enough power such that they can recognize all $\omega$-regular languages.
%\textcolor{red}{We furthermore extend the models to \emph{multi-PA}, which can have one semi-linear set for each individual state. 
%This is convenient but does not give additional expressive power in all models. TODO: Check if this is still needed.}

%In order to avoid confusion (and sticking to the notation in the literature), we use the terms blind multicounter machines (BMCM), meaning the model operating on finite words introduced in \Cref{chap:finite}, while blind counter automata (BCA) always denotes the model operating on infinite words which we introduce now.
A \emph{blind counter automaton} (BCA) is a tuple $\Mmc = (Q, \Sigma, q_0, \Delta, F)$ where $\Delta \subseteq Q \times (\Sigma \cup \{\varepsilon\}) \times \Zbb^d \times Q$ for some $d \geq 1$ is a finite set. 
We adapt the definition of configuration to infinite words as follows.
%A \emph{configuration} for an infinite word $\alpha = \alpha_1\alpha_2\alpha_3\dots$ of $\Mmc$ is a tuple of the form $c = (p, \alpha_1 \dots \alpha_i, \alpha_{i+1} \alpha_{i+2} \dots, \vbf) \in Q \times \Sigma^* \times \Sigma^\omega \times \Zbb^k$ for some $i \geq 0$. 
A \emph{configuration} of $\Mmc$ is a tuple $c = (q, \alpha, \vbf) \in Q \times \Sigma^\omega \times \Zbb^d$, where $q$ is the current state, $\alpha$ is the suffix of the input word that has not been read yet, and $\vbf$ represents the current counter value.
%Slightly abusing notation, we simply call $c$ configuration for the remainder of this section, as we do not consider BMCM on finite words here.
We say $c$ \emph{derives} into $c'$, written $c \vdash_\Mmc c'$, if $c = (q, a\beta, \vbf)$, $c' = (q', \beta, \vbf')$ and $(q, a, \ubf, q') \in \Delta$ with $\vbf' = \vbf + \ubf$; or if $c = (q, \beta, \vbf)$, $c' = (q', \beta, \vbf')$ and $(q, \varepsilon, \ubf, q') \in \Delta$ with $\vbf' = \vbf + \ubf$.
%
%The configuration~$c$ \emph{derives} into a configuration $c'$, written $c \vdash c'$, if either $c' = (q, a \beta \dots, \vbf + \ubf)$ and $(p, \alpha_{i+1}, \ubf, q) \in \Delta$, or $c' = (q, \alpha_1 \dots \alpha_i, \alpha_{i+1}\alpha_{i+2} \dots, \vbf + \ubf)$ and $(p,\varepsilon, \ubf,q) \in \Delta$. 
Wa say $\Mmc$~\emph{accepts} an infinite word $\alpha$ if there is an infinite sequence of configuration derivations $c_1 \vdash c_2 \vdash c_3 \vdash \dots$ with $c_1 = (q_0, \varepsilon, \alpha, \0)$ such that for infinitely many $i$ we have $c_i = (p_i, \alpha_{j+1} \alpha_{j+2} \dots, \0)$ with $p_i \in F$ and for all $j \geq 1$ there is a configuration of the form $(p, \alpha_{j+1} \alpha_{j+2} \dots, \vbf)$ for some $p \in Q$ and $\vbf \in \Zbb^k$ in the sequence.
That is, an infinite word is accepted if we infinitely often visit an accepting state when the counters are $\0$, and every symbol of $\alpha$ is read at some point.
We define the $\omega$-language recognized by $\Mmc$ as $L_\omega(\Mmc) = \{\alpha \in \Sigma^\omega \mid \Mmc \text{ accepts } \alpha\}$.

Parikh automata naturally generalize to Parikh automata with \mbox{$\epsilon$-transitions}.
An \mbox{$\epsilon$-PA} is a tuple $\Amc = (Q, \Sigma, q_0, \Delta, \Emc, F, C)$ where $\Emc \subseteq Q \times \{\varepsilon\} \times \Nbb^d \times Q$ is a finite set of \emph{labeled $\varepsilon$-transitions}, and all other entries are defined as for~PA. 
A run of $\Amc$ on an infinite word $\alpha_1\alpha_2\alpha_3 \dots$ is an infinite sequence of transitions $r \in (\Emc^* \Delta)^\omega$, say $r = r_1r_2r_3 \dots$ with $r_i = (p_{i-1}, \gamma_i, \vbf_i, p_i)$ such that $p_0 = q_0$, and $\gamma_i = \varepsilon$ if $r_i \in \Emc$, and $\gamma_i = \alpha_j$ if $r_i \in \Delta$ is the $j$-th occurrence of a (non-$\varepsilon$) transition in $r$. 
The acceptance conditions of the models translate to runs of $\varepsilon$-PA in the obvious way. We use terms like $\varepsilon$-safety~PA, $\varepsilon$-reachability~PA, etc, to denote an $\varepsilon$-PA with the respective acceptance condition.

Note that we can treat every~PA as an $\epsilon$-PA, that is, a~PA $\Amc = (Q,\Sigma, q_0, \Delta, F, C)$ is equivalent to the $\varepsilon$-PA $\Amc' = (Q, \Sigma, q_0, \Delta, \varnothing, F, C)$. 

We start with the following simple observation.

\begin{lemma}
\label{lem:kcountertoPPBA} 
BCA and $\varepsilon$-Büchi~PA are equivalent. 
\end{lemma}
\begin{proof}
We first show that for every BCA $\Mmc$ there is an equivalent $\varepsilon$-Büchi~PA $\Amc$.
Let $\Mmc = (Q, \Sigma, q_0, \Delta, F)$ be a BCA of dimension $d$. For a vector \mbox{$(x_1, \dots, x_d) \in \Zbb^d$} we define the vector $\vbf^\pm = (x_1^+, \dots x_d^+, x_1^-, \dots x_d^-) \in \Nbb^{2d}$ as follows: if $x_i$ is positive, then $x_i^+ = x_i$ and $x_i^- = 0$. Otherwise, $x_i^+ = 0$ and $x_i^- = |x_i|$. 
We construct an equivalent $\varepsilon$-Büchi~PA $\Amc = (Q, \Sigma, q_0, \Delta', \Emc', F, C)$ of dimension $2d$, where $\Delta' = \{(p, a, \vbf^\pm, q) \mid (p, a, \vbf, q) \in \Delta\}$ and $\Emc' = \{(p, \varepsilon, \vbf^\pm, q) \mid (p, \varepsilon, \vbf, q) \in \Delta\}$. 
Finally, let $C = \{(x_1, \dots, x_d, x_1, \dots, x_d) \mid x_i \in \Nbb\}$. It is now easily verified that $L_\omega(\Mmc) = P_\omega(\Amc)$.

\medskip
For the reverse direction we show that for every Büchi~PA $\Amc$ there is an equivalent BCA~$\Mmc$.
Let $\Amc = (Q, \Sigma, q_0, \Delta, F, C)$ be a Büchi~PA of dimension $d$ where $C = \bigcup_{i \leq \ell} C(\bbf_i, P_i)$. Note that we have $B_\omega(\Amc) = \bigcup_{i \leq \ell} B_\omega(Q, \Sigma, q_0,\Delta, F, C(\bbf_i, P_i))$ by the infinite pigeonhole principle. Hence, we can assume that~$C = C(\bbf, P)$ is linear as BCA are closed under union \cite{blindcounter}.
We construct a blind $d$-counter machine $\Mmc$ that simulates $\Amc$ as follows: $\Mmc$ consists of a copy of $\Amc$ where the accepting states have are equipped with additional $\varepsilon$-transitions subtracting the period vectors $\pbf \in P$.
We only need to consider the base vector $\bbf$ a single time, hence we introduce a fresh initial state $q_0'$ and a $\varepsilon$-transition from $q'_0$ to $q_0$ subtracting~$\bbf$.
%Observe that a vector $\vbf$ lies in $C = \{b_0 + b_1z_1 + \dots + b_\ell z_\ell \mid z_1, \dots, z_\ell\}$ if and only if $\vbf - b_1z_1 - \dots - b_\ell z_\ell - b_0 = \0$ for some $z_i$. 
%Intuitively, $\Mmc$ computes the vector $\vbf$ in the copies of $Q$ and guesses the $z_i$ in the accepting states.
Formally, we construct $\Mmc = (Q \cup \{q_0'\}, \Sigma, q_0', \Delta', F)$ where
\[\Delta' = \Delta \cup \{(q_0', \varepsilon, -\bbf, q_0\} \cup \{(q_f, \varepsilon, -\pbf, q_f) \mid q_f \in F, \pbf \in P\}.\]
It is now easily verified that $B_\omega(\Amc) = L_\omega(\Mmc)$.
\end{proof}

\subsection{$\epsilon$-elimination}

We now show that almost all~PA models admit $\epsilon$-elimination. 
We first consider Büchi~PA, where $\epsilon$-elimination implies the equivalence of blind counter automata and Büchi~PA by~\Cref{lem:kcountertoPPBA}. 
We provided an elementary combinatorial proof on the automaton level in the manuscript~\cite[Theorem 4.2]{infiniteOurs}. 
We thank Georg Zetzsche for outlining a much simpler proof (using a heavier toolbox) which we present here.

\begin{theorem}
    $\epsilon$-Büchi~PA admit $\epsilon$-elimination. 
\end{theorem}
\begin{proof}
    Observe that the construction in \Cref{lem:kcountertoPPBA} translates $\epsilon$-free BCA into $\epsilon$-free Büchi~PA. We can hence translate a given Büchi~PA into a BCA and eliminate $\epsilon$-transitions and then translate back into a Büchi~PA. Therefore, all we need to show is that BCA admit $\epsilon$-elimination. 
    To show that BCA admit $\epsilon$-elimination we observe that 
    \[\text{$L$ is recognized by a BCA} \quad \Leftrightarrow \quad L=\bigcup_i U_iV_i^\omega,\]
    where $U_i$ is a language of finite words that is recognized by a finite-word BCA (we do not give a formal definition here as it does not yield further insights) and~$V_i$ is a language that is recognized by a finite-word BCA where the initial state in the only accepting state. 
    The proof of this observation is very similar to the proof of \Cref{lem:charBuchiPA} and we leave the details to the reader. 

    As shown in \cite{greibach,latteux,georgBlindCounter}, finite-word BCA admit $\varepsilon$-elimination. Furthermore, from the proof technique established in \cite[Lemma 7.7]{georgBlindCounter} it is immediate that the condition $F = \{q_0\}$ can be preserved. 
    We obtain $\epsilon$-free finite-word BCA $\Amc_i'$ and $\Bmc_i'$ for the languages~$U_i$ and $V_i$. 
    Using the construction of~\cite[Theorem 32]{klaedtkeruess}, we can translate~$\Amc_i'$ and $\Bmc_i'$ into~PA~$\Amc_i$ and $\Bmc_i$, where the $\Bmc_i$ satisfy $F_i=\{q_0\}$ and the sets~$C_i$ are homogeneous linear sets. Now the statement follows by \Cref{lem:charBuchiPA}.
\end{proof}

We continue with $\varepsilon$-reachability, $\varepsilon$-reachability-regular and $\varepsilon$-limit~PA, as we show $\varepsilon$-elimination using the same technique for these models.
As shown in \Cref{obs:reachability} and \Cref{thm:LimitEqualsReach}, the class of $\omega$-languages recognized by reachability~PA coincides with the class of $\omega$-languages of the form $\bigcup_i U_i \Sigma^\omega$ for Parikh recognizable $U_i$, and the class of reachability-regular and limit~PA recognizable $\omega$-languages coincides with the class of $\omega$-languages of the form $\bigcup_i U_i V_i^\omega$ for Parikh recognizable $U_i$ and regular $V_i$, respectively. 
It is well-known that NFA and~PA on finite words are closed under homomorphisms and hence admit \mbox{$\varepsilon$-elimination}~\cite{klaedtkeruess} (as a consequence of \cite[Proposition II.11]{latteux}, $\varepsilon$-transitions can even be eliminated without changing the semi-linear set). 
The characterizations allow us to reduce $\varepsilon$-elimination of these infinite word~PA to the finite case.
\begin{lemma}
 $\varepsilon$-reachability, $\varepsilon$-reachability-regular, and $\varepsilon$-limit~PA admit $\varepsilon$-elimination.
\end{lemma}
\begin{proof}
We show the statement for $\varepsilon$-reachability~PA. The technique can very easily be translated to the other two models.
Let~$\Amc$ be an $\varepsilon$-reachability~PA with $R_\omega(\Amc) = L \subseteq~\Sigma^\omega$. 
Let~$\Amc_e$ be the reachability~PA obtained from $\Amc$ by replacing every $\varepsilon$-transition with an $e$-transition, where $e$ is a fresh symbol that does not appear in $\Sigma$. 
Let $h$ be the homomorphism that erases the symbol $e$, i.e., $h(e) = \varepsilon$.
Observe that $\Amc_e$ recognizes an $\omega$-language $L_e \subseteq (\Sigma \cup \{e\})^\omega$ with the property that $h(L_e) = L$ (note that by definition $\{\epsilon\}^\omega=\varnothing$).
Now, by \Cref{obs:reachability} we can write $L_e$ as $\bigcup_i U_i \cdot (\Sigma \cup \{e\})^\omega$ where $U_i \subseteq (\Sigma \cup \{e\})^*$ is Parikh recognizable.
As the class of Parikh recognizable languages is closed under homomorphisms~\cite{klaedtkeruess}, we have
\[L = h(L_e) = h\left(\bigcup_i U_i \cdot (\Sigma \cup \{e\})^\omega\right) = \bigcup_i h(U_i)\cdot \Sigma^\omega,\]
and can hence find a reachability~PA for $L$. 
The proof for reachability-regular and limit~PA works the same way, as the regular languages are also closed under homomorphisms.
\end{proof}

Now we show that strong $\varepsilon$-reset~PA and weak $\varepsilon$-reset~PA admit $\varepsilon$-elimination.
We show that these two models are equivalent. Hence to show this statement we only need to argue that strong $\varepsilon$-reset~PA admit $\varepsilon$-elimination.

\begin{lemma}
\label{lem:SPBAtoWPBA}
Every strong $\varepsilon$-reset~PA $\Amc$ is equivalent to a weak $\varepsilon$-reset~PA $\Amc'$ that has the same set of states and uses one additional counter. If $\Amc$ is a strong reset~PA, then $\Amc'$ is a weak reset~PA.
\end{lemma}
\begin{proof}
Let $\Amc = (Q, \Sigma, q_0, \Delta, \Emc, F, C)$ be a strong $\varepsilon$-reset~PA. We construct an equivalent weak $\varepsilon$-reset~PA $\Amc'$ that simulates $\Amc$, ensuring that no run visits an accepting state without resetting. 
To achieve that, we add an additional counter that tracks the number of visits of an accepting state (without resetting). 
Moreover, we define~$C'=C\cdot \{1\}$, such that this new counter must be set to 1 when visiting an accepting state, thus disallowing to pass such a state without resetting. 
Now it is clear that $\Amc'$ is a weak $\varepsilon$-reset~PA equivalent to $\Amc$. 
Observe that if $\Amc$ has no $\epsilon$-transitions, then $\Amc'$ has no $\epsilon$-transitions.
\end{proof}

\begin{lemma}
\label{lem:WPBAtoSPBA}
Every weak $\varepsilon$-reset~PA $\Amc$ is equivalent to a strong $\varepsilon$-reset~PA $\Amc'$ with at most twice the number of states and the same number of counters. If~$\Amc$ is a strong reset~PA, then~$\Amc'$ is a weak reset~PA.
\end{lemma}
\begin{proof}
Let $\Amc = (Q, \Sigma, q_0, \Delta, \Emc, F, C)$ be a weak $\varepsilon$-reset~PA. We construct an equivalent strong $\varepsilon$-reset~PA $\Amc'$ that simulates $\Amc$ by having the option to ``avoid'' accepting states arbitrarily long. For this purpose, we create a non-accepting copy of $F$. Consequently, $\Amc'$ can decide to continue or reset a partial run using non-determinism. Again, it is clear that~$\Amc'$ is equivalent to~$\Amc$. Observe that if $\Amc$ has no $\epsilon$-transitions, then $\Amc'$ has no $\epsilon$-transitions.

\end{proof}

\begin{lemma}
\label{lem:SPBAepselim}
Strong $\varepsilon$-reset~PA admit $\varepsilon$-elimination.
\end{lemma}

\begin{proof} 
    Let $\Amc = (Q, \Sigma, q_0, \Delta, \Emc, F, C)$ be a strong $\varepsilon$-reset~PA of dimension $d$. We assume without loss of generality that~$q_0$ has no incoming transitions (this can be achieved by introducing a fresh copy of~$q_0$). Furthermore, we assume that $F \neq \varnothing$ (otherwise \mbox{$SR_\omega(\Amc) = \varnothing$}). 
    Let the states of $Q$ be ordered arbitrarily, say $Q = \{q_0, \dots, q_{n-1}\}$.
    We construct an equivalent strong reset~PA $\Amc' = (Q', \Sigma, q_0, \Delta', F', C')$ of dimension $d + n$. In the beginning, $\Amc'$ is a copy of $\Amc$ (keeping the $\epsilon$-transitions for now), which is modified step-by-step.
    The purpose of the new counters is to keep track of the states that have been visited (since the last reset). Initially, we hence modify the transitions as follows:
    for every transition $(q_i, \gamma, \vbf, q_j) \in \Delta \cup \Emc$ we replace $\vbf$ by~$\vbf \cdot \ebf_j^n$.

    Let $p, q \in Q$. Assume there is a sequence of transitions $\tilde\lambda = r_1 \dots r_j \dots r_k \in \Emc^* \Delta \Emc^*$; $1\leq j \leq k \leq 2n+1$, where 
  
    \begin{itemize}
        \item $r_j = (p_{j-1}, a, \vbf_j, p_j) \in\Delta$, and 
        \item $r_i = (p_{i-1}, \varepsilon, \vbf_i, p_i) \in \Emc$ for all $i \neq j, i \leq k$,
        \item such that $p_0 = p, p_k = q$, and $p_i\neq p_\ell$ for $i,\ell\leq j$ and $p_i\neq p_\ell$ for $i,\ell\geq j$, and
        \item all internal states are non-accepting, \ie, $p_i \notin F$ for all $0 < i < k$.
    \end{itemize}  
    
    Then we introduce the \emph{shortcut} $(p, a, \rho(\tilde\lambda), q)$, where $\rho(\tilde\lambda)$ is computed already with respect to the new counters, tracking that the $p_i$ in $\tilde\lambda$ have been visited, \ie, the counters corresponding to the $p_i$ in this sequence have non-zero values. 
    %\item First, for every pair $(p,q)$ if there is an $\epsilon^{\leq |Q|}a\epsilon^{\leq |Q|}$-path between $p$ and $q$ without internal accepting states, then we introduce the $a$ transition from $p$ to $q$ with vector value as on the path. We now assume that these transitions are there (we introduced no new states yet).
    %\item We assume there is a accepting state, otherwise the language is empty. 
    
    Let $p,q\in Q$. We call a (possibly empty) sequence $\lambda = r_1 \dots r_k \in \Emc^*$ with $r_i = (p_{i-1}, \varepsilon, \vbf_i, p_i)$ and $p_0 = p, p_k = q$ a \emph{no-reset $\varepsilon$-sequence} from $p$ to~$q$ if all internal states are non-accepting, \ie, $p_i \notin F$ for all $0 < i < k$. A \emph{no-reset $\epsilon$-path} is a no-reset sequence such that $p_i\neq p_j$ for $i\neq j$. 
    Observe that the set of no-reset $\varepsilon$-paths from $p$ to $q$ is finite, as the length of each path is bounded by $n-1$.
    We call the pair $(p,q)$ a \emph{$C$-pair} if there is a no-reset $\varepsilon$-sequence $r$ from $p$ to $q$ with $\rho(r) \in C$, where $\rho(r)$ is computed in $\Amc$. %\textcolor{red}{Note that we can decide if a pair is good by constructing $\hat\Bmc$ and using Parikh's theorem.}
    %\item Let $p,q\in Q$. We call the pair $(p,q)$ \emph{$\epsilon$-good} if there is a $p$-$q$-walk of $\epsilon$-transitions such that all internal nodes are non-accepting and such that the vector of the walk belongs to $C$. 
    
    Let $S=(f_1,\ldots, f_\ell)$ be a non-empty sequence of pairwise distinct accepting states (note that this implies $\ell \leq n$). We call $S$ a \emph{$C$-sequence} if each $(f_i,f_{i+1})$ is a $C$-pair.
    
    For all $p,q\in F$ and $C$-sequences $S$ such that $p=f_1$ if $p\in F$ and $q=f_\ell$ if $q\in F$, we introduce a new state $(p,S,q)$.
    We add $(p,S,q)$ to $F'$, that is, we make the new states accepting. 
    State $(p,S,q)$ will represent a partial run of the automaton with only $\epsilon$-transitions starting in $p$, visiting the accepting states of $S$ in that order, and ending in $q$.

    Observe that in the following we introduce only finitely many transitions by the observations made above; we will not repeat this statement in each step.
    Let $p, q \in Q$ and $S = (f_1, \dots, f_\ell)$ be a $C$-sequence.
    For every transition of the form $(s, a, \vbf, p) \in \Delta$ we insert new transitions 
    \[\{(s, a, \vbf + \rho(\lambda), (p,S,q)) \mid \lambda \text{ is a no-reset $\varepsilon$-path from $p$ to $f_1$}\}\] 
	to~$\Delta'$. 
    Similarly, for every transition of the form $(q, a, \vbf, t) \in \Delta$ we insert new transitions 
    \[\{((p,S,q), a, \vbf + \rho(\lambda), t) \mid \lambda \text{ is a no-reset $\varepsilon$-path from $f_\ell$ to $q$}\}\]
    to $\Delta'$. 
    Again this set is finite.
    Additionally, let $p', q' \in Q$ and $S' = (f_1', \dots, f'_k)$ be a $C$-sequence. For every sequence $\tilde\lambda=\lambda \delta \lambda'$ where $\lambda$ is a no-reset $\varepsilon$-path from $f_\ell$ to $q$, $\delta = (q, a, \vbf, p')$, and~$\lambda'$ is a no-reset $\varepsilon$-path from~$p'$ to $f'_1$
    we add $\big((p,S,q), a, \rho(\tilde{\lambda}), (p',S',q')\big)$ to $\Delta'$. 
    
    Lastly, we connect the initial state $q_0$ in a similar way (recall that we assume that $q_0$ has no incoming transitions, and in particular no loops). 
    For every transition $(p, a, \vbf, q) \in \Delta$ and every $C$-sequence $S = (f_1, \dots, f_l)$ with the property that $(q_0, f_1)$ is a $C$-pair and there is a no-reset $\varepsilon$-path $\lambda$ from $f_\ell$ to $p$, we introduce the transition $(q_0, a, \rho(\lambda) + \vbf, q)$ for every such path $\lambda$. 
    Additionally, for every $C$-sequence $S' = (f_1', \dots f'_k)$ such that there is a no-reset $\varepsilon$-path~$\lambda'$ from $q$ to $f_1'$, we introduce the transition $\big(q_0, a, \rho(\lambda) + \vbf + \rho(\lambda'), (q,S',t)\big)$ for all such paths $\lambda, \lambda'$ and $t\in Q$. 
    Furthermore, for every no-reset $\varepsilon$-path $\hat\lambda$ from $q_0$ to~$p$, we introduce the transition $(q_0, a, \rho(\hat{\lambda}) + \vbf + \rho(\lambda'), (q,S',t))$ for all $t \in Q$. 
    A reader who is worried that we may introduce too many transitions at this point shall recall that $(q,S',t)$ has no outgoing transition if there does not exist a no-reset $\epsilon$-path from $f_k'$ to $t$.
    %\item For every state $s$ with a symbol-transition $(a,x)$ to $p$ and every state $t$ with a symbol-transition $(b,y)$ from $q$ to $t$ we insert the symbol-transition $(a,x+u)$ from $s$ to $(p,q,S)$ if the pair $(p,f_1)$ is $\epsilon$-good and there is an $\epsilon$-walk from $p$ to $f_1$ in $\Bmc_{p,f_1,S}^\epsilon$ with vector value $u$. Similarly, we add the symbol-transition $(y+w)$ from $(p,q,S)$ to $t$ if the pair $(f_\ell, q)$ is $\epsilon$-good and there is an $\epsilon$-walk from $f_\ell$ to $q$ in $\Bmc_{f_\ell,q,S}$ with vector value $w$. Here, $\Bmc_{p,q,S}$ is defined as $\Bmc_{p,q}$ where the nodes of $S$ have been deleted. 
    %\item We introduce new counters. For all new transitions outgoing of $p$, we count the nodes that we have seen on the $p$-$f_1$-path. For the new transitions outgoing of $(p,q,S)$ we count the nodes that we have seen on the $f_\ell$-$q$-path and verify with the semi-linear sets that we took valid paths.
    Finally, we delete all $\epsilon$-transitions.

    We define $C'$ similar to the construction by Klaedtke and Ruess \cite{klaedtkeruess} used to eliminate $\varepsilon$-transitions in the finite setting.
    For every $q \in Q \setminus F$ we define $C_q = \{\rho(r) \mid r \in \Emc^*$ is partial run of $\Amc$ starting and ending in~$q$ that does not visit any accepting state$\}$. As a consequence of Parikh's theorem \cite{parikh1966context} and \cite[Lemma 5]{klaedtkeruess}, the sets $C_q$ are semi-linear.
    %
    %For every $q \in Q$ we define $C_q = \sum \hat{\Bmc}_{q,q}$, where $\hat{\Bmc}_{q,q}$ is defined as $\Bmc_{q,q}$ but without any accepting states, that is, $\hat{\Bmc}_{q,q} = (Q \setminus F, \Gamma, q, \{(p, \vbf, p') \mid (p, \varepsilon, \vbf, p') \in \Emc, p, p' \notin F\},\{q\})$ for a suitable alphabet $\Gamma \subseteq \Nbb^d$. 
    Then $C' = \{\vbf \cdot (x_0, \dots, x_{n-1}) \mid \vbf + \ubf \in C, \ubf \in \sum_{x_i \geq 1} C_{q_i}\}$. 
    By this, we substract the $C_{q_i}$ if the counter for $q_i$ is greater or equal to one, that is, the state has been visited. 
    This finishes the construction. 
    
    We now prove that $\Amc'$ is equivalent to $\Amc$. In the one direction we compress the run by using the appropriate shortcuts, in the other direction we unravel it accordingly. 

    \medskip
    To show that $SR_\omega(\Amc)\subseteq SR_\omega(\Amc')$, let $\alpha \in SR_\omega(\Amc)$ with accepting run $r = r_1 r_2 r_3 \dots$. If there are no $\varepsilon$-transitions in $r$, we are done (as $r$ is also an accepting run of $\Amc'$ on $\alpha$).
    
    Otherwise, we construct an accepting run $r'$ of $\Amc'$ on $\alpha$ by replacing maximal \mbox{$\varepsilon$-sequences} in $r$ step-by-step. Let $i$ be minimal such that $r_i \dots r_j$ is a maximal $\varepsilon$-sequence. 
    Let $r_i = (p_{i-1}, \varepsilon, \vbf_i, p_i), r_j = (p_{j-1}, \varepsilon, \vbf_j, p_j)$, and $r_{j+1} = (p_j, \alpha_z, \vbf_{j+1}, p_{j+1})$. It might be the case that $i = 1$, \ie, the run $r$ starts with an $\varepsilon$-transition leaving $q_0$. Otherwise $i > 1$ and we can write $r_{i-1} = (p_{i-2}, \alpha_{z-1}, \vbf_{i-1}, p_{i-1})$.
    By allowing the empty sequence, we may assume that there is always a second (possibly empty) maximal $\varepsilon$-sequence $r_{j+2} \dots r_{k}$ starting directly after $r_{j+1}$. 
    %If $r_{i+2} \in \Emc$, this makes sense. Otherwise we assume that $r_{i+2} \dots r_k$ is empty by assuming that $k = i+1$. This looks weird, but significantly simplifies the following case distinction, 
    We distinguish (the combination of) the following cases.
    \begin{itemize}
        \item At least one state in $r_i \dots r_j$ is accepting, \ie, there is a position $i-1 \leq \ell \leq j$ such that $p_\ell \in F$ (F) or not (N).
        \item At least one state in $r_{j+2} \dots r_k$ is accepting, \ie, there is a position $j+1 \leq \ell' \leq k$ such that $p_{\ell'} \in F$ (F) or not (N). If $r_{j+2} \dots r_k$ is empty, we are in the case~(N).
    \end{itemize}

Hence, we consider four cases in total.
    
    \begin{itemize}
        \item Case (NN). That is, there is no accepting state in $r_i \dots r_k$. Note that the $\varepsilon$-sequence $r_i \dots r_j$ can be decomposed into an $\varepsilon$-path and $\varepsilon$-cycles as follows. 
    If we have $p_{i_1} \neq p_{j_1}$ for all $i \leq i_1 < j_1 \leq j$ we are done as $r_i \dots r_j$ is already an $\varepsilon$-path. Otherwise let $i_1 \geq i$ be minimal such that there is $j_1 > i_1$ with $p_{i_1} = p_{j_1}$, that is, $r_{i_1+1} \dots r_{j_1}$ is an $\varepsilon$-cycle. 
    If $r_i \dots r_{i_1} r_{j_1+1} \dots r_j$ is an $\varepsilon$-path, we are done. 
    Otherwise, let $i_2 > j_1$ be minimal such that there is $j_2 > i_2$ with $p_{i_2} = p_{j_2}$, that is, $r_{i_2+1} \dots r_{j_2}$ is an $\varepsilon$-cycle. 
    Then again, if $r_i \dots r_{i_1} r_{j_1+1} \dots r_{i_2} r_{j_2+1} \dots r_j$ is an $\varepsilon$-path, we are done. 
    Otherwise, we can iterate this argument and obtain a set of $\varepsilon$-cycles $r_{i_1+1} \dots r_{j_1}, \dots, r_{i_m+1} \dots r_{j_m}$ for some $m$, and an $\varepsilon$-path $\hat{r}_{i,j} = r_i \dots r_{i_1} r_{j_1+1} \dots r_{i_{m}} r_{j_m+1} \dots r_j$ which partition $r_i \dots r_j$. 
    Now observe that $\rho(r_{i_1+1} \dots r_{j_1}) + \dots + \rho(r_{i_m+1} \dots r_{j_m}) \in C_{p_{i_1}} + \dots + C_{p_{i_m}}$. 
    We can do the same decomposition for the $\varepsilon$-sequence $r_{j+2} \dots r_k$ into a set of $\varepsilon$-cycles and an $\varepsilon$-path $\hat{r}_{j+2,k}$. 
    By the construction of $\Delta'$, there is a shortcut 
    \[\delta = (p_{i-1}, \alpha_z, (\rho(\hat{r}_{i,j}) + \vbf_{j+1} + \rho(\hat{r}_{j+2,k})) \cdot \hat\ebf^n, p_k),\] 
	where $\hat\ebf^n$ is the $n$-dimensional vector counting the states appearing in $\hat r_{i,j}$ and $\hat r_{j+2,k}$ and the state $p_{j+1}$. 
    By the construction of $\Delta'$ and $C'$, we may subtract all $\varepsilon$-cycles that have been visited in $r_i \dots r_k$, hence, we may replace $r_i \dots r_k$ by $\delta$ to simulate exactly the behavior of~$\Amc$.

\item 
    Case (NF). That is, there is no accepting state in $r_i \dots r_j$ but at least one accepting state in $r_{i+2} \dots r_k$ (in particular, this sequence is not empty). 
    %First, we assume $i > 0$ (the current $\varepsilon$-sequence does not start at the beginning of $r$). 
    Let $\ell_1, \dots, \ell_m$ denote the positions of accepting states in $r_{i+2} \dots r_k$, and let $\ell_0 < \ell_1$ be maximal such that~$\ell_0$ is resetting (this is before $r_i$, and if such an $\ell_0$ does not exist, let $\ell_0 = 0$), \ie, $\ell_0$ is the position of the last reset before the reset at position $\ell_1$. 
    As $r$ is an accepting run, the sequence $S = (\ell_1, \dots, \ell_m)$ is a $C$-sequence (we may assume that all states in~$S$ are pairwise distinct, otherwise there is a reset-cycle, which can be ignored). 
    In the same way as in the previous case we can partition the $\varepsilon$-sequence $r_i \dots r_j$ into an $\varepsilon$-path $\hat{r}_{i,j}$ and a set of $\varepsilon$-cycles, which may be subtracted from~$C$. 
    Likewise, we can partition the sequence $r_{j+2} \dots r_{\ell_1}$ into an $\varepsilon$-path $\hat{r}_{j+2,\ell_1}$ and $\varepsilon$-cycles with the same property. 
    By the construction of $\Delta'$ there is a shortcut $(p_{i-1}, a, \rho(\hat{r}_{i,j}) + \vbf_{j+1}, p_{j+1})$ and hence a transition 
    \[\delta = (p_{i-1}, a, \rho(\hat{r}_{i,j}) + \vbf_{j+1} + \rho(\hat{r}_{j+2, \ell_1}), (p_{j+1}, S,p_k)).\] 
    Note that this is also the case if~$i = 1$. 
    Thus, we replace $r_i \dots r_k$ by~$\delta$. 
    In particular, $\rho(r_{\ell_0+1} \dots r_{i-1}\delta)$ can be obtained from $\rho(r_{\ell_0+1} \dots r_{\ell_1})$ by subtracting all $\varepsilon$-cycles that have been visited within this partial run. 
    Furthermore, observe that the containment of $\rho(r_{\ell_1+1} \dots r_{\ell_2}), \dots, \rho(r_{\ell_{m-1}+1} \dots r_{\ell_m})$ in $C$ depends only on the automaton, and not the input word. 
    As the counters are reset in $r_{\ell_m}$, we may continue the run from $\delta$ the same way as in $r_k$, using an appropriate transition from~$\Delta'$ that adds the vector $\rho(\hat{r}_{\ell_{m}+1, k})$, thus respecting the acceptance condition.
    %everything this is always the same dear readers if you came this far you probably do not need any further explanation right? Thus let us just abbreviate the rest.

    \item
    Case (FN). Similar to (NF), but this time we replace $r_{i-1}r_i \dots r_j$ by an appropriate transition into a state of the form $(p_{i-2}, \alpha_z, \vbf, (p_{i-1}, S, p_j))$ for a suitable $C$-sequence $S$ and vector $\vbf$, followed by a shortcut leading to $p_k$. If $i = 0$ (we enter a $C$-sequence before reading the first symbol), we make use of the transitions introduced especially for $q_0$.

    \item 
    Case (FF). Similar to (FN) and (NF), but we transition from a state of the form $(p_{i-1}, S, p_j)$ into a state of the form $(p_{j+1}, S', p_k)$ for suitable $C$-sequences $S, S'$, again respecting the case $i = 0$.
    \end{itemize}

	To show that $SR_\omega(\Amc')\subseteq SR_\omega(\Amc)$ we unravel the shortcuts and $(p, S,q)$-states introduced in the construction. Let $\alpha \in SR_\omega(\Amc')$ with accepting run $r' = r'_1 r'_2 r'_3 \dots$. We replace every transition $r'_i \in \Delta' \setminus \Delta$ (\ie, transitions that do not appear in $\Amc$) by an appropriate sequence of transitions in $\Amc$.
    Let $i \geq 1$ be minimal such that $r'_i$ is a transition in~$\Delta' \setminus \Delta$.

    %First we argue for the case $i = 1$, that is the first transition in $r'$ does not appear in $\Amc$. We distinguish three subcases.
    We distinguish the form of $r'_i$ and show that the possible forms correspond one-to-one to the cases in the forward direction.
    
    \begin{itemize}
        \item Case (NN). The case that $r'_i = (p, a, \rho(\tilde{\lambda}), q)$ is a shortcut, \ie, $\tilde{\lambda} \in \Emc^* \Delta \Emc^*$, corresponds to the case (NN). 
        In particular, there are no accepting states in $r$. Let $k < i$ be the position of the last reset before $r'_i$, and $k'$ the position of the first reset after $r'_i$, where $k' = i$ if $r'_i$ transitions into a accepting state. By the acceptance condition we have $\rho(r'_{k+1} \dots r'_{k'}) \in C - (\sum_{q \in Q'} C_q)$ for some set $Q' \subseteq Q$ based on the counter values. Hence, we can replace $r'_i$ by the partial run $\tilde{\lambda}$ filled with possible $\varepsilon$-cycles on some states in $Q'$.
        \item Case (NF). The case that $r'_i = (s, a, \vbf + \rho(\lambda), (p,S,q))$ such that $S = (f_1, \dots f_\ell)$ is a $C$-sequence, there is a transition $\delta = (s, a, \vbf, p) \in \Delta$ and $\lambda$ is a no-reset $\varepsilon$-path from~$p$ to~$f_1$, corresponds to the case (NF). 
        By the definition of $C$-sequence there is a sequence~$r_{f_1, f_\ell}$ of $\varepsilon$-transitions in $\Amc$ starting in $f_1$, ending in $f_\ell$, visiting the accepting states $f_1$ to~$f_\ell$ (in that order) such that the reset-acceptance condition is satisfied on every visit of one the accepting states. 
        Then we can replace $r'_i$ by $\delta \lambda r_{f_1, f_\ell}$, possibly again filled with some $\varepsilon$-cycles based on the state counters of~$\lambda$, similar to the previous case. Note that at this point we do not yet unravel the path from $f_\ell$ to $q$, as it depends on how the run $r'$ continues (as handled by the next two cases).
        \item Case (FN). The case that $r'_i = ((p,S,q), a, \vbf + \rho(\lambda), t)$ such that $S = (f_1, \dots f_\ell)$ is a $C$-sequence, there is a transition $\delta = (q, a, \vbf, t) \in \Delta$ and $\lambda$ is a no-reset $\varepsilon$-path from~$f_\ell$ to~$q$, corresponds to the case (FN). 
        Similar to the previous case, we can replace~$r'_i$ by~$\lambda\delta$, possibly again amended with some $\varepsilon$-cycles based on the state counters of~$\lambda$.
        If $i = 1$, the transition might also be of the form $r'_1 = (q_0, \alpha_1, \rho(\lambda) + \vbf, t)$ such that $S$ is a $C$-sequence with the property that $(q_0, f_1)$ is a $C$-pair. Then there is a sequence of $\varepsilon$-transitions $r_{q_0, f_\ell}$ in $\Amc$ as above. 
        Then we replace $r'_1$ by $r_{q_0, f_\ell} \lambda \delta$ (with possible $\varepsilon$-cycles) instead.
        \item Case (FF). The case that $r'_i = ((p,S,q), a, \rho(\tilde{\lambda}), (p', S', q')$ such that $S = (f_1, \dots f_\ell)$ and $S' = (f'_1, \dots, f'_k)$ are $C$-sequences, there is a transition $\delta = (q, a, \vbf, p') \in \Delta$ and $\tilde{\lambda} = \lambda \delta \lambda'$,  where $\lambda$ is a no-reset $\varepsilon$-path from $f_\ell$ to $q$ and $\lambda'$ is a no-reset $\varepsilon$-path from $p'$ to $f_1'$, corresponds to the case (FF).
        This case can be seen as the union of the previous cases. There is a sequence $r_{f'_1, f'_k}$ of $\varepsilon$-transitions in $\Amc$, as in the case (RF). Hence, we replace~$r'_i$ by $\tilde{\lambda} r_{f'_1, f'_k}$ (with possible $\varepsilon$-cycles).
        If $i = 1$, the transition might also be of the form $r'_1 = (q_0, \alpha_1, \rho(\lambda) + \vbf + \rho(\lambda'), (p', S', q'))$ such that $(q_0, f_1)$ is a $C$-pair. Then there is a sequence of $\varepsilon$-transitions $r_{q_0, f_\ell}$ in $\Amc$ as above, and we replace $r'_1$ by $r_{q_0, f_\ell} \tilde{\lambda} r_{f'_1, f'_k}$.
    \end{itemize}

Observe that the size of $\Amc'$ is in $\Omc(|\Amc|^2|\Amc|!)$. This finishes the proof of the lemma. 
\end{proof}

Finally we show that safety and co-Büchi~PA do not admit $\varepsilon$-elimination.
\begin{lemma}
\label{lem:safetyEps}
$\varepsilon$-safety~PA and $\varepsilon$-co-Büchi~PA do not admit $\varepsilon$-elimination.
\end{lemma}

\begin{proof}
Consider the automaton $\Amc$ in \Cref{fig:safetyEps} with $C=\{(z,z') \mid z' \geq z\}$.
\begin{figure}
    \centering
    \begin{tikzpicture}[->,>=stealth',shorten >=1pt,auto,node distance=3cm, semithick]
    \tikzstyle{every state}=[minimum size=1cm]

    \node[state, accepting, initial, initial text={}] (q0) {$q_0$};	
    \node[state, accepting] (q1) [right of=q0] {$q_1$};

    \path
    (q0) edge [loop above] node {$\varepsilon, \begin{pmatrix}0\\1\end{pmatrix}$} (q0)
    (q0) edge              node {$a, \begin{pmatrix}0\\0\end{pmatrix}$} (q1)
    (q1) edge [loop above] node {$b, \begin{pmatrix}1\\0\end{pmatrix}$} (q1)
    (q1) edge [bend left]  node {$b, \begin{pmatrix}0\\0\end{pmatrix}$} (q0)	
    ;
    \end{tikzpicture}
    \caption{The $\varepsilon$-PA with $C=\{(z,z') \mid z' \geq z\}$ for the proof of \Cref{lem:safetyEps}.}
    \label{fig:safetyEps}
\end{figure}

If we interpret $\Amc$ as an $\varepsilon$-safety or $\varepsilon$-co-Büchi~PA, we have we have $S_\omega(\Amc) = CB_\omega(\Amc) = \{ab^+\}^\omega$. This $\omega$-language is neither safety~PA nor co-Büchi~PA recognizable (one can easily adapt the proof in \cite[Theorem 3]{infiniteZimmermann} showing that $\{\alpha \in \{a,b\}^\omega \mid |\alpha|_a = \infty\}$ is neither safety~PA nor co-Büchi~PA recognizable).

Observe how $\Amc$ utilizes the $\varepsilon$-transition to enforce that $q_0$ is seen infinitely often: whenever the $b$-loop on $q_1$ is used, the first counter increments. The semi-linear set states that at no point the first counter value may be greater than the second counter value which can only be increased using the $\varepsilon$-loop on $q_0$. Hence, any infinite word accepted by $\Amc$ may contain arbitrary infixes of the form $b^n$ for $n < \infty$, as the automaton can use the $\varepsilon$-loop on~$q_0$ at least $n$ times before, but not $b^\omega$.
\end{proof}

We generalize the trick presented in the previous proof to show that $\varepsilon$-co-Büchi~PA recognize all Büchi~PA recognizable $\omega$-languages. Further adapting the trick we show that $\varepsilon$-safety~PA recognize all reset~PA recognizable $\omega$-languages.
\begin{lemma}
\label{lem:BukkuSubsetEpsCoBukki}
    The class of Büchi~PA recognizable $\omega$-languages is a strict subclass of the class of $\varepsilon$-co-Büchi~PA recognizable $\omega$-languages.
\end{lemma}
\begin{proof}
 Let $\Amc = (Q, \Sigma, q_0, \Delta, F, C)$ be a Büchi~PA of dimension $d$. Without loss of generality we assume $q_0 \notin F$ (this can be achieved by adding a fresh initial state). We construct an equivalent $\varepsilon$-co-Büchi~PA $\Amc'$ of dimension $d+3$ as follows.
 Let $Q' = Q \cup \{q_0'\} \cup \{q_f' \mid q_f \in F\}$ where $q_0'$ is a fresh state and the $q'_f$ are copies of the accepting states of $\Amc$ in which we do not expect good counter values. We define $\Amc' = (Q', \Sigma, q_0', \Delta', \Emc', Q', C')$. The idea is as follows: we use one of the three new counters to indicate that we would like to check the current counter values for membership in $C$. The other two additional counters are used to enforce that we see a state in $F$ infinitely often, that is that we indeed must check the current counter values for membership in $C$ infinitely often. To achieve that, we introduce $\varepsilon$-loops on every state that is accepting in $\Amc$ as well as on the new initial state. The automaton then guesses the number of symbols that are read before visiting the next accepting state and increases the second counter accordingly. The first counter is incremented every time a symbol of the infinite input word is read. At every state we expect the second counter value to be greater or equal the first counter value.

 Hence, we define
 \begin{align*}
     \Delta' =&\ \{(p, a, \vbf \cdot (1,0,0), q) \mid (p,a,\vbf,q) \in \Delta, p,q \notin F\} \\
     \cup&\ \{(p, a, \vbf \cdot (1,0,1), q_f), (p, a, \vbf \cdot (1,0,0), q_f') \mid (p,a,\vbf,q_f) \in \Delta, p \notin F, q_f \in F\} \\
     \cup&\ \{(p_f, a, \vbf \cdot (1,0,1), q), (p_f', a, \vbf \cdot (1,0,0), q) \mid (p_f,a,\vbf,q) \in \Delta, p_f \in F, q\notin F\} \\
     \cup&\ \{(p_f', a, \vbf \cdot (1,0,1), q_f), (p_f, a, \vbf \cdot (1,0,1), q_f') \mid (p_f,a,\vbf,q_f) \in \Delta, p_f, q_f \in F\} \\
     \cup&\ \{(p_f, a, \vbf \cdot (1,0,0), q_f), (p_f', a, \vbf \cdot (1,0,0), q_f') \mid (p_f,a,\vbf,q_f) \in \Delta, p_f, q_f \in F\},
 \end{align*}
 and
 \begin{align*}
     \Emc' =&\ \{(q'_0, \varepsilon, \0^{d} \cdot (0,0,1), q_0),(q'_0, \varepsilon, \0^d \cdot (0,1,0), q'_0)\} \cup \{(q_f, \varepsilon, \0^d \cdot (0,1,0), q_f) \mid q_f \in F\}.
 \end{align*}

 Finally, we define 
 \begin{align*}
 C' =&\ \{\0^{d} \cdot (0, y, 0) \mid y \in \Nbb\}\\
 \cup&\ \{(v_1, \dots, v_d, x, y, 2p+1) \mid p \in \Nbb, y \geq x, (v_1, \dots, v_d) \in \Nbb^{d}\} \\
  \cup&\ \{(v_1, \dots, v_d, x, y, 2p+2) \mid p \in \Nbb, y \geq x, (v_1, \dots, v_d) \in C\}.
 \end{align*}

 The interested reader can now (fairly) easy verify that $\Amc$ and $\Amc'$ are equivalent.

 The strictness of the inclusion follows immediately from~\cite{infiniteZimmermann}, as there are $\omega$-languages recognized by co-Büchi~PA but not by any Büchi~PA.
\end{proof}

Before we show that $\varepsilon$-safety~PA are powerful enough to simulate reset~PA,
we first introduce the following Myhill-Nerode-like notion for linear sets.
Let $C$ be a linear set of dimension $d$. We call $C$ \emph{congruent} if for every $\vbf_1, \vbf_2 \in C$ and every $\vbf \in \Nbb^d$ we have $\vbf_1 + \vbf \in C$ if and only of $\vbf_2 + \vbf \in C$. We call a congruent linear set $C$ \emph{resetting} if $\0 \in C$ (\ie, if $C$ is congruent and homogeneous). 

\begin{lemma}
    The class of reset~PA recognizable $\omega$-languages is a strict subclass of the class of $\varepsilon$-safety~PA recognizable $\omega$-languages.
\end{lemma}
\begin{proof}
Let $\Amc = (Q, \Sigma, q_0, \Delta, F, C)$ be a strong reset~PA of dimension $d$ with $C = \bigcup_{i \leq \ell} C(\bbf_i, P_i)$. We assume again without loss of generality that $q_0 \notin F$. We construct an equivalent $\varepsilon$-safety~PA $\Amc = (Q', \Sigma, q_0', \Delta', \Emc', Q', C')$ of dimension $2d + 3$. The idea is to combine the ideas from the previous proof while observing that we can simulate resets by testing membership in a resetting linear set. To be precise, the additional three counters serve the same purpose as in the last proof while the additional $d$ counters allow us to simulate resets.
Hence, we start by defining $C' = C_{\mathsf{check}} \cup C_{\mathsf{don't\ care}}$ where
\[C_{\mathsf{check}} = \{(v_1, \dots, v_d, v_1, \dots, v_d, z,z,2p) \mid p, z, v_1, \dots, v_d \in \Nbb\}\]
and
\[C_{\mathsf{don't\ care}} = \{(v_1, \dots, v_{2d} x,y,2p+1) \mid p \in \Nbb, y \geq x, (v_1, \dots, v_{2d}) \in \Nbb^{2d}\}.\]
The crucial observation is that $C_{\mathsf{check}}$ is a resetting linear set. In the next step we introduce a new initial state that is equipped with an $\varepsilon$-loop and a $\varepsilon$-transition to $q_0$ exactly as the in previous proof.
Then we replace every accepting state $f \in F$ by the gadget depicted in \Cref{fig:gadget}.
Here we use the following observation: whenever a vector $\vbf$ is contained in $C$, then it is in particular contained $C(\bbf_i, P_i)$ for a $i \leq \ell$ and can hence be written as $\bbf_i + \sum_{\pbf \in P_i} \pbf z_\pbf$ for some values $z_\pbf \in \Nbb$. While the first~$d$ counters of $\Amc'$ are copies of the~$d$ counters of $\Amc$, the second set of~$d$ counters are used to represent any vector that is contained in $C$, \ie, can be written as mentioned above.
The $\varepsilon$-loops in the state $f^{(i)}$ allow us to add any vector contained in $C(\bbf_i, P_i)$ to the counter values of the second set of $d$ counters. Hence, the vector induced by the first $d$ counters is contained in $C$ if and only if $f$ can be reached with a vector whose vector induced by the second set of~$d$ counters is equivalent to the vector induced by the first~$d$ counters. This is checked in the state $f$ where membership in $C_{\mathsf{check}}$ is tested (as the last counter has an even value when entering $f$).

We refrain from giving the details of the construction of $\Amc'$ as they do not yield further insights.
\begin{figure}
	\centering
	\begin{tikzpicture}[->,>=stealth',shorten >=1pt,auto,node distance=3cm, semithick]
	\tikzstyle{every state}=[minimum size=1cm]

 	\node[state, accepting] (fin) {$f_{in}$};	
  	\node[state, accepting] (f1) [above right = 5mm and 2cm of fin] {$f^{(1)}$};
    \node[state, accepting] (fl) [below right = 5mm and 2cm of fin] {$f^{(\ell)}$};
    \node[state, accepting] (f)  [above right = 5mm and 2cm of fl ] {$f$};
	\node[state, accepting] (fout) [right of=f] {$f_{out}$};
	
	\path
	(fin) edge node {$\varepsilon, \begin{pmatrix}\0^d\\\bbf_1 \\ 0 \\ 0 \\ 0\end{pmatrix}$} (f1)
 	(fin) edge node[below left] {$\varepsilon, \begin{pmatrix}\0^d\\\bbf_\ell \\ 0 \\ 0 \\ 0\end{pmatrix}$} (fl)
  	(f1) edge [loop above] node {$\varepsilon, \begin{pmatrix}\0^d\\\pbf_1 \\ 0 \\ 0 \\ 0\end{pmatrix}$} (f1)
    (f1) edge node {$\varepsilon, \begin{pmatrix}\0^d\\\0^d \\ 0 \\ 0 \\ 1\end{pmatrix}$} (f)
    (fl) edge [loop below] node {$\varepsilon, \begin{pmatrix}\0^d\\\pbf_\ell \\ 0 \\ 0 \\ 0\end{pmatrix}$} (fl)
    (fl) edge[below right] node {$\varepsilon, \begin{pmatrix}\0^d\\\0^d \\ 0 \\ 0 \\ 1\end{pmatrix}$} (f)
    (f) edge node {$\varepsilon, \begin{pmatrix}\0^d\\\0^d \\ 0 \\ 0 \\ 1\end{pmatrix}$} (fout)
    (fout) edge [loop above] node {$\varepsilon, \begin{pmatrix}\0^d\\\0^d \\ 0 \\ 1 \\ 0\end{pmatrix}$} (fout)
	;

    \node[right = 2cm of fin] {$\vdots$};
    \node (in1)  [above left = 4mm and 6mm of fin] {};
    \node (in2)  [below left = 4mm and 6mm of fin] {};
    \node (out1) [above right = 4mm and 6mm of fout] {};
    \node (out2) [below right = 4mm and 6mm of fout] {};

    \path[dashed]
    (in1) edge (fin)
    (in2) edge (fin)
    (fout) edge (out1)
    (fout) edge (out2)
    ;
        
	\end{tikzpicture}
	\caption{We replace every accepting state $f \in F$ by such a gadget. Every state $f^{(i)}$ is equipped with an $\varepsilon$-loop for every $\pbf_i \in P_i$. Note that the actual membership of the current counter values in~$C$ is (still) checked in $f$ as we require the current counter values to be in $C_{\mathsf{check}}$.}
	\label{fig:gadget}
\end{figure}
\end{proof}

Observe that $\varepsilon$-safety~PA generalize $\varepsilon$-co-Büchi~PA, as they can guess the number $n$ of transitions a $\varepsilon$-co-Büchi~PA uses before they constantly verify the Parikh-condition. Very similar to the previous proof, the $\varepsilon$-safety~PA uses an $\varepsilon$-transition on a fresh initial state $n$ times to postpone the verification of the Parikh-condition for $n$ transitions.

We conclude by arguing that $\varepsilon$-co-Büchi~PA do not generalize reset~PA. To achieve that,  we define the following property, which is heavily inspired by the pumping-style lemma for Parikh recognizable (finite word) languages as introduced in~\cite{cadilhac2011expressiveness}.

Let $\Lmc$ be a class of $\omega$-languages. We say that $\Lmc$ has the \emph{exchange property} if for every $L \in \Lmc$ there are $p, \ell \in \Nbb$ such that for every $\alpha \in L$ the following holds:
there is a partition $\alpha = uvxv\beta$ with $0 < |v| \leq p < |x|$ and $|uvxv| \leq \ell$ such that $uv^2x\beta \in L$ and $uxv^2\beta \in L$.

\begin{observation}
The classes of Büchi~PA and ($\varepsilon$-)co-Büchi~PA recognizable \mbox{$\omega$-languages} have the exchange property. The classes of reset~PA and ($\varepsilon$-)safety~PA recognizable $\omega$-languages do not have the exchange property.
\end{observation}

As an immediate consequence of this observation and the argument above we obtain the following corollaries.
\begin{corollary}
The class of $\varepsilon$-co-Büchi~PA recognizable $\omega$-languages is a strict subset of the class of $\varepsilon$-safety~PA recognizable $\omega$-languages.
\end{corollary}
\begin{corollary}
The classes of reset~PA recognizable $\omega$-languages and $\varepsilon$-co-Büchi~PA recognizable $\omega$-languages are incomparable.    
\end{corollary}

\subsection{Remaining Closure Properties}
As a preparation for the next section we establish the remaining closure properties. As limit~PA and reachability-regular~PA are equivalent by \Cref{thm:LimitEqualsReach}, we only argue for limit~PA. Similarly, as strong reset~PA and weak reset~PA are equivalent by \Cref{lem:SPBAtoWPBA} and \Cref{lem:WPBAtoSPBA}, we only argue for strong reset~PA.

First we observe that the $\omega$-languages recognized by reset~PA are ultimately periodic.
\begin{lemma}
\label{lem:resetUlti}
 Let $\Amc$ be a reset~PA. If $SR_\omega(\Amc) \neq \varnothing$, then $\Amc$ accepts an infinite word of the form $uv^\omega$.
\end{lemma}
\begin{proof}
Assume $SR_\omega(\Amc) \neq \varnothing$. Then there exists an infinite word $\alpha \in SR_\omega(\Amc)$ with accepting run $r = r_1 r_2 r_3 \dots$, where $r_i = (p_{i-1}, \alpha_i, \vbf_i, p_i)$. Let $k_1 < k_2 < \dots$ be the positions of all accepting states in $r$. Let $k_i < k_j$ be two such positions such that $p_{k_i} = p_{k_j}$.
Let $u = \alpha_1 \dots \alpha_{k_i}$ be the prefix of $\alpha$ read upon visiting $p_{k_i}$ and $v = \alpha_{k_i + 1} \dots \alpha_{k_j}$ the infix read between $p_{k_i}$ and~$p_{k_j}$. 
%Note that $v \neq \varepsilon$ by the choice of $k_j$.
Then $\Amc$ also accepts $uv^\omega$, as $r_1 \dots r_{k_i} (r_{k_i + 1} \dots r_{k_j})^\omega$ is an accepting run of $\Amc$ on $uv^\omega$ by definition.
\end{proof}

\begin{lemma}
\label{lem:intersectionNondet}
The classes of limit~PA recognizable and reachability-regular~PA recognizable $\omega$-languages are closed intersection. The class of reset~PA recognizable $\omega$-languages is not closed under intersection.
\end{lemma}
\begin{proof}
The closure of limit~PA (and hence reachability-regular~PA) under intersection follows from a simple product construction, exactly as for finite word~PA~\cite[Theorem 21]{klaedtkeruess}. 
The non-closure of reset~PA under intersection follows from an argument similar to the argument showing non-closure of blind counter machines~\cite{blindcounter} (on infinite words). Let $L_1 = \{a^n b^n \mid n > 0\}^\omega$ and $L_2 = \{a\}\{b^n a^{2n} \mid n > 0\}$. Then $L_1 \cap L_2$ contains only one infinite word, namely $ab a^2 b^2 a^4 b^4 \dots$ which is not ultimately periodic. Hence, $L_1 \cap L_2$ is not accepted by any strong reset~PA as a consequence of the previous lemma.
\end{proof}

Finally, we show that all our models are not closed under complement.
\begin{lemma}
\label{lem:complementNondet}	
The classes of limit~PA recognizable, reachability-regular~PA recognizable, and reset~PA recognizable $\omega$-languages are not closed complement.
\end{lemma}
\begin{proof}
Non-closure for reset~PA follows immediately from the closure under union and non-closure under intersection by De Morgan's law.

Hence, we only need to argue for limit~PA. Let $D = \{ww \mid w \in \{a,b\}^*\}$. As shown in~\cite[Lemma 26]{klaedtkeruess}, this language is not Parikh recognizable. However, its complement~$\overline{D}$ is Parikh recognizable.
Now let $L = \overline{D} \cdot \{c\}^\omega$. By \Cref{lem:concatenation}, this $\omega$-language is limit~PA recognizable.
Observe that $\overline{L} = \{a,b\}^\omega \cup \{a,b\}^* \{c\}^* \{a,b\} \{a,b,c\}^\omega \cup PAL \cdot \{c\}^\omega$.
The first two languages of the union are $\omega$-regular and hence limit~PA recognizable. Hence, it is sufficient to show that 
$D \cdot \{c\}^\omega$ is not limit~PA recognizable.
Observe that limit~PA have the exchange property. Cadilhac et al.~\cite{cadilhac2011expressiveness} have shown how to establish their pumping-syle lemma to show that the language $\{w\#w \mid w \in \{a,b\}^*\}$ is not Parikh recognizable.
The proof can easily adapted to show that $D$ is not Parikh recognizable (one only needs to remove the~\#), which again can easily adapted to show that $D \cdot \{c\}^\omega$ and hence $\overline{L}$ is not recognized by any limit~PA by exploiting the exchange property.
\end{proof}

At this point we are ready to present a complete picture of comparing all (non-deterministic) models, see~\Cref{fig:expr-nondet}.

\begin{figure}
    \begin{tikzpicture}[%
	node distance=27mm,>=Latex,
	initial text="", initial where=below left,
	every state/.style={rectangle,rounded corners,draw=black,thin,fill=black!5,inner sep=1mm,minimum size=6mm},
	every edge/.style={draw=black,thin}
	]
	
	\node[state] (reach) {reachability~PA};
	\node[state, below = 3.2cm of reach] (reg) {$\omega$-regular};
	\node[state, below right = 1cm and 0cm of reach, align=center] (limit) {reachability-regular~PA \\ = limit~PA = $\LPAReg$};
	\node[state, right = 18mm  of reach] (cobuchi) {co-Büchi~PA};
	\node[state, right = 18mm of cobuchi] (epscobuchi) {$\varepsilon$-co-Büchi~PA};
	\node[state, right = 2cm of reg] (regPA) {$\LRegPA$ = reset~PA $(**)$};
	\node[state, right = 2cm of regPA] (PAPA) {$\LPAPA$ = reset~PA $(*)$};
	\node[state, right = 1cm of limit] (buchi) {Büchi~PA = BCA};
	\node[state, right = 1.5cm of buchi] (reset) {reset~PA};
	\node[state, above right = 1cm and 0cm of epscobuchi] (safety) {safety~PA};
	\node[state, right = 18mm of epscobuchi] (epssafety) {$\varepsilon$-safety~PA};

	%\node at (2.8, -3.5) {$(*)$ suitable graph theoretical restrictions};
	
	%\node[red] at (2.8,0) {\Huge TODO};
	
	\path[-{Latex}]
	(safety) edge (epssafety)
	(cobuchi) edge (epscobuchi)
	(reach) edge (cobuchi)
	(reach) edge (limit)
	(limit) edge (buchi)
	(reg) edge (limit)
	(reg) edge (regPA)
	(regPA) edge (PAPA)
	(PAPA) edge (reset)
	(buchi) edge (PAPA)
	(buchi) edge (epscobuchi)
	(reset) edge (epssafety)
	(epscobuchi) edge (epssafety)
	;
	
    \node[align=left, anchor=west] at (-1.5, -5.5) {\small $(*)$ At most one state $q$ per leaf of $C(\Amc)$ may have incoming transitions from outside the  leaf, this \\\small\phantom{$(*)$} state $q$ is the only accepting state in the leaf, and there are no accepting states in non-leaves;};
	\node[align=left, anchor=west] at (-1.68, -6.25) {\small $(**)$ and only transitions connecting states in leaves may be labeled with non-zero vectors.};
\end{tikzpicture}   
\caption{Comparison of Parikh automata on infinite words. Arrows indicate strict inclusions while no (non-transitive) connections mean incomparability.}
\label{fig:expr-nondet}
\end{figure}

\subsection{Decision Problems}
\label{subsec:decision}
In this section, we study the following classical decision problems for~PA on infinite words. 
\begin{itemize}
    \item Emptiness: given a~PA $\Amc$, is the $\omega$-language of $\Amc$ empty?
    \item Membership: given a~PA $\Amc$ and finite words $u, v$, does $\Amc$ accept $uv^\omega$?
    \item Universality: given a~PA $\Amc$, does $\Amc$ accept every infinite word?
\end{itemize}

Furthermore, we study the classical model checking problem, where we are given a system~$\Kmc$ and a specification, \eg, represented as an automaton $\Amc$, and the question is whether every run of $\Kmc$ satisfies the specification, 
\ie, we ask $L(\Kmc) \subseteq L(\Amc)$, which is true if and only if $L(\Kmc) \cap \overline{L(\Amc)} = \varnothing$.
However, as complementing is often expensive or not even possible, another approach is to specify the set of all bad runs and ask whether no run of $\Kmc$ is bad, which boils down to the question is $L(\Kmc) \cap L(\Amc) = \varnothing$?
We call the first approach \emph{universal model checking} and the latter approach \emph{existential model checking}.
In our setting we assume the specification $\Amc$ to be a~PA operating on infinite words, while the system $\Kmc$ may be given as a Kripke-structure (which can be seen as a safety automaton~\cite{ClarkeHandbook}), in which case the goal is to solve \emph{safety model checking}, or also as a~PA operating on infinite words, in which case the goal is to solve \emph{PA model checking}.
Hence, we consider four problems in total, which boil down to the following decision problems.

\begin{itemize}
    \item Inclusion: given a safety automaton or a~PA $\Amc_1$, and a~PA $\Amc_2$, is the $\omega$-language of~$\Amc_1$ a subset of the $\omega$-language of $\Amc_2$?
    \item Intersection emptiness: given a safety automaton or a~PA $\Amc_1$, and a~PA $\Amc_2$, is the $\omega$-language of $\Amc_1$ disjoint from the $\omega$-language of $\Amc_2$?
\end{itemize}

Observe that our characterization results imply that we can efficiently translate reachability-regular~PA into Büchi~PA which in turn can efficiently be translated into strong reset~PA.
Furthermore strong reset~PA and weak reset~PA can efficiently be translated into each other by \Cref{lem:SPBAtoWPBA} and \Cref{lem:WPBAtoSPBA}.
Hence, for showing upper bounds for these models it is sufficient to argue for reset~PA.
However, the proof of \Cref{lem:limitToLPAReg} showing that every limit~PA recognizable $\omega$-language is a member of $\LPAReg$ constructs a number of finite word~PA that is exponentially in the dimension of the limit~PA we start with. Hence, we argue for limit~PA separately.
Contrary, we can translate reachability-regular~PA efficiently into limit~PA; hence, for showing lower bounds of all newly introduced models it is sufficient to argue for reachability-regular~PA.

We begin by showing that emptiness for reset~PA is $\coNP$-complete, where the upper bound even holds if we allow $\varepsilon$-transitions (which does not increase their expressiveness but our $\varepsilon$-elimination procedure constructs an equivalent reset~PA of super-polynomial size).
Hence, reset~PA (even with $\varepsilon$-transitions) are a powerful model that can still be used for algorithmic applications.

We present an algorithm solving non-emptiness for reset~PA in $\NP$ by exploiting that $\omega$-languages accepted by reset~PA are ultimately periodic.
As a consequence, we can reduce non-emptiness for $\varepsilon$-reset~PA to the finite word case, as clarified in the following lemma.
\begin{lemma}
\label{lem:nonemptiness}
    Emptiness for reachability-regular~PA, strong reset~PA and weak reset~PA is $\coNP$-complete.
\end{lemma}
\begin{proof}
    The $\coNP$-hardness for reachability-regular~PA (and hence all of these models) follows immediately from the $\coNP$-hardness of finite word~PA~\cite{emptynp}. 
    Hence, we focus on the membership in $\coNP$ by presenting an $\NP$-algorithm for non-emptiness that works even for strong reset~PA with $\varepsilon$-transitions. 
    Let $\Amc$ be a strong $\varepsilon$-reset~PA.
    By \Cref{lem:resetUlti} it suffices to check whether $\Amc$ accepts 
    an ultimately periodic infinite word $uv^\omega$. If such a word exists, we may assume that there is an accepting run $r_u r_v^\omega$ of $\Amc$ on $uv$ where neither~$r_u$ nor~$r_v$ visit the same accepting state twice (otherwise we simply remove such cycles in the run).
    For any $p,q \in Q$ we
    define $\Amc_{p \Rightarrow q} = (Q \cup \{q_0'\}, \Sigma, q_0', \Delta', \Emc', \{q\}, C)$,
    where $\Delta' = \{(q_1, a, \vbf, q_2) \mid (q_1, a, \vbf, q_2) \in \Delta, q_1 \notin F\}
    \cup \{(q_0', a, \vbf, q_2) \mid (p, a, \vbf, q_2) \in \Delta\}$ and, analogously,
    \mbox{$\Emc' = \{(q_1, \varepsilon, \vbf, q_2) \mid (q_1, \varepsilon, \vbf, q_2) \in \Emc, q_1 \notin F\}
    \cup \{(q_0', \varepsilon, \vbf, q_2) \mid (p, \varepsilon, \vbf, q_2) \in \Delta\}$}.

  	Hence, if we interpret $\Amc_{p\Rightarrow q}$ as a finite word~PA, then it accepts all words accepted by the finite word~PA $\Amc$ when starting in $p$, ending in $q$, and not visiting an accepting state in-between. In other words, if $p,q\in F$, then $\Amc_{p \Rightarrow q}$ accepts all finite infixes that the strong reset~PA $\Amc$ may read when the last reset was in $p$, and the next reset is in $q$
    
    Now, the following $\NP$ algorithm solves non-emptiness:
    \begin{enumerate}
    \item Guess a sequence $f_1, \ldots, f_k$ of accepting states with $k \leq 2|F|$ such that $f_i = f_k$ for some $i \leq k$.
    \item Verify that $L(\Amc_{q_0 \Rightarrow f_1}) \neq \varnothing$ and $L(\Amc_{f_j \Rightarrow f_{j+1}}) \neq \varnothing$ for all $1 \leq j < k$ (interpreted as~PA over finite words).
    \item Verify that $L(\Amc_{f_i \Rightarrow f_{i+1}}) \cdot \ldots \cdot L(\Amc_{f_{k-1} \Rightarrow f_k}) \not\subseteq \{\varepsilon\}$.
    \end{enumerate}
    
    The third step can be done by adding a fresh symbol (say $e$) to the automata and replacing every $\varepsilon$-transition with an $e$-transition (observe that this does construction does not change the emptiness behavior, and is, in contrast to the $\varepsilon$-elimination procedure in \cite{klaedtkeruess} computable in polynomial time).
    Afterwards we use the NP-algorithm for non-emptiness for~PA~\cite{emptynp}.
    
    The third step essentially states that not all $L(\Amc_{f_j \Rightarrow f_{j+1}})$ for $j \geq i$ may only accept the empty word, as we require $v \neq \varepsilon$. To check this property, we can construct a~PA\footnote{This is possible in polynomial time by a standard construction very similar to the one of \Cref{lem:concatenation}.} recognizing $L(\Amc_{f_i \Rightarrow f_{i+1}}) \cdot \ldots \cdot L(\Amc_{f_{k-1} \Rightarrow f_k})$, and again replace every $\varepsilon$-transition with an $e$-transition. Finally, we build the product automaton with the~PA (NFA) that recognizes the language $\{w \in (\Sigma \cup \{e\})^* \mid w \text{ contains at least 1 symbol from } \Sigma\}$, which is possible in polynomial time~\cite{klaedtkeruess} and test non-emptiness for the resulting~PA.
\end{proof}

Let us quickly explain how to modify the algorithm such that we obtain an $\NP$-algorithm for non-emptiness for limit~PA.
\begin{lemma}
\label{lem:nonemptiness-limit}
Emptiness for limit~PA is $\coNP$-complete.
\end{lemma}
\begin{proof}
Let $\Amc$ be a limit~PA of dimension $d$.
Recall the proof of \Cref{lem:limitToLPAReg} showing $L_\omega(\Amc)$ is a member of $\LPAReg$ by iterating over all subsets $D \subseteq [d]$.
Intuitively, such a subset $D$ indicates the counters which we expect to be $\infty$ when processing an infinite word $\alpha \in L_\omega(\Amc)$.
For every fixed $D$, we utilize \Cref{lem:semi-linear-inf} to compute a sequence of finite word~PA, say $\Amc_1, \Amc'_1, \dots \Amc_n, \Amc'_n$ for some $n \geq 1$ in polynomial time such that
$\alpha \in \bigcup_{i \leq n} L(\Amc_i) \cdot L(\Amc'_i)^\omega = L$. Given these automata, we can compute a reachability-regular~PA recognizing $L$ in polynomial time by \Cref{cor:LPARegToReachReg}.
In particular, for every fixed $D$ we can compute a reachability-regular~PA (and hence a reset~PA by the comment above) accepting all words that are accepted by $\Amc$ such that there is an accepting run $r$ of~$\Amc$ on $\alpha$ with $\mathsf{Inf}(\rho(r)) = D$ (recall that $\Inf(\vbf)$ for some $\vbf \in \Nbb^d$ denotes the positions of all $\infty$-entries in $\vbf$).
Hence, we can use the $\NP$-algorithm in the previous proof by first guessing a good set $D$ and turning only the matching ``relevant part'' of $\Amc$ into a reset~PA.
\end{proof}

We will now turn our attention to the membership problem.
Note that, given finite words $u$ and $v$, we can always construct a safety automaton that recognizes $uv^\omega$ and no other infinite word with $|uv|$ many states. Recall that every state of a safety automaton is accepting. We show that the intersection of a reset~PA recognizable $\omega$-language and a safety automaton-recognizable $\omega$-language remains reset~PA recognizable using a product construction which is computable in polynomial time. Hence, we can reduce the membership problem to the non-emptiness the standard way.

\begin{lemma}
\label{lem:resetsafety}
The class of reset~PA recognizable $\omega$-languages is closed under intersection with safety automata-recognizable $\omega$-languages.
\end{lemma}
\begin{proof}
We show a construction for strong $\varepsilon$-reset~PA that is computable in polynomial time.
Let $\Amc_1 = (Q_1, \Sigma, q_1, \Delta_1, \Emc_1, F_1, C_1)$ be a strong $\varepsilon$-reset~PA and $\Amc_2 = (Q_2, \Sigma, q_2, \Delta_2, Q_2)$ be a safety automaton.
Consider the product automaton
\[\Amc = (Q_1 \times Q_2, \Sigma, (q_1, q_2), \Delta, \Emc, F_1 \times Q_2, C_1)\] with 
\[\Delta = \{((p,q), a, \vbf, (p',q') \mid (p, a, \vbf, p') \in \Delta_1 \text{ and } (q, a, q') \in \Delta_2\}\] and 
\[\Emc = \{((p,q), \varepsilon, \vbf, (p',q)) \mid (p, \varepsilon, \vbf, p') \in \Emc_1 \text{ and } q \in Q_2\}.\]

As every state of $\Amc_2$ is accepting, we need to take care that $\Amc$ does not use a transition that is not enabled in $\Amc_2$ while mimicking the behavior of $\Amc_1$.
Hence, it is easily verified that $SR_\omega(\Amc) = SR_\omega(\Amc_1) \cap L_\omega(\Amc_2)$.
\end{proof}

Combining the previous two results, we obtain that membership for all our models is in~$\NP$. For the lower bound, we reduce from the membership problem for semi-linear sets, that is, given a semi-linear set $C \subseteq \Nbb^d$ and a vector $\vbf \in \Nbb^d$, deciding membership of~$\vbf$ in $C$. This problem is known to be $\NP$-complete, which follows immediately from the $\NP$-algorithms for integer programming~\cite{semilinUpper1,semilinUpper2} and the \NP-hardness of (a variant of) subset sum~\cite{intract,karp}; see also~\cite{haase} for a short discussion.
\begin{cor}
\label{cor:membership}	
    Membership for limit~PA, reachability-regular~PA, strong reset~PA and weak reset~PA is $\NP$-complete.
\end{cor}

We will now turn our attention to the universality and inclusions problems, the latter being the core of solving universal model checking.
Note that we can always reduce universality to inclusion, as an automaton $\Amc$ is universal if and only if $\Sigma^\omega$ is a subset of the $\omega$-language of $\Amc$.
Observe however that universality (and hence inclusion) remain undecidable for our models, as these problems are already undecidable for reachability~PA \cite{infiniteZimmermann} which can effectively be translated into reachability-regular~PA.
This implies that the universal model checking problems are undecidable for our models.
\begin{cor}
Universality and inclusion are undecidable for limit~PA, reachability-regular~PA, strong reset~PA and weak reset~PA.
\end{cor}

Contrary, the existential safety model checking for all our models, that is, intersection emptiness for~PA with safety automata is $\coNP$-complete.
Hardness follows immediately from the hardness of emptiness, while containment in $\coNP$ is as an immediate consequence of~\Cref{lem:resetsafety}.
\begin{cor}
Intersection-emptiness for limit~PA, reachability-regular~PA, strong reset~PA and weak reset~PA with safety automata is $\coNP$-complete.
\end{cor}

We continue with the existential~PA model checking problem. This problems is $\coNP$ complete for limit~PA and reachability-regular~PA.
Again, hardness follows from the hardness of emptiness, while containment follows from the closure under intersection, as witnessed by the product construction in \Cref{lem:intersectionNondet} which is computable in polynomial time
\begin{lem}
\label{lem:interLimitReachreg}	
Intersection-emptiness for limit~PA and reachability-regular~PA is $\coNP$-complete.
\end{lem}

We continue with Büchi~PA, as their intersection-emptiness problems has not been studied in the literature before. Their intersection-emptiness problem remains $\coNP$-complete; however, we need more sophisticated methods to proof this statement.
\begin{lemma}
\label{lem:interBuchi}
    Intersection-emptiness for Büchi~PA is $\coNP$-complete.
\end{lemma}
\begin{proof}
    The lower bound follows again from their $\coNP$-complete emptiness problem.
    
    We give a proof sketch showing that intersection non-emptiness for Büchi~PA is in $\NP$ by utilizing a recent result essentially stating that Ramsey-quantifiers in Presburger formulas can be eliminated in polynomial time~\cite{ramsey}. The authors show how to use the Ramsey-quantifier to check liveness properties for systems with counters. In particular, the existence of an accepting run of a Büchi~PA (answering the question whether the accepted $\omega$-language is non-empty) can be expressed with a Presburger formula with a Ramsey-quantifier.
    Hence, checking if the intersection of the two $\omega$-languages recognized by two Büchi~PA can be tested by intersecting two Presburger-formulas and moving the quantifiers to the front. We refer to Sections 4.1 and 8.2 in~\cite{ramsey} for more information.
\end{proof}

We conclude by showing that intersection emptiness is undecidable for strong reset~PA.
The result relies on the fact that the intersection of two such languages can encode non-terminating computations of two-counter machines.

A two-counter machine $\Mmc$ is a finite sequence of instructions
\[(1 : \Isf_1)(2 : \Isf_2) \dots (k-1 : \Isf_{k-1})(k : \mathsf{STOP})\]
where the first component of a pair $(\ell, \Isf_\ell)$ is the line number, and the second component is the instruction in line $\ell$.
An instruction is of one of the following forms:
\begin{itemize}
    \item $\mathsf{Inc}(Z_i)$, where $i = 0$ or $i = 1$.
    \item $\mathsf{Dec}(Z_i)$, where $i = 0$ or $i = 1$.
    \item $\mathsf{If}\ Z_i = 0\ \mathsf{goto}\ \ell'\ \mathsf{else}\ \ell''$, where $i = 0$ or $i = 1$, and $\ell', \ell'' \leq k$.
\end{itemize}

Instructions of the first or second form are called increments resp.\ decrements, while instructions of the latter form are called zero-tests.
A configuration of $\Mmc$ is a tuple $c = (\ell, z_0, z_1)$, where $\ell \leq k$ is the current line number, and $z_0, z_1 \in \Nbb$ are the current counter values of $Z_0$ and $Z_1$ respectively. We say $c$ \emph{derives} into its unique successor configuration $c'$, written $c \vdash c'$, as follows.

\begin{itemize}
    \item If $\Isf_\ell = \mathsf{Inc}(Z_0)$, then $c' = (\ell + 1, z_0 + 1, z_1)$.
    \item If $\Isf_\ell = \mathsf{Inc}(Z_1)$, then $c' = (\ell + 1, z_0, z_1 + 1)$.
    \item If $\Isf_\ell = \mathsf{Dec}(Z_0)$, then $c' = (\ell + 1, \max\{z_0 - 1, 0\}, z_1)$.
    \item If $\Isf_\ell = \mathsf{Dec}(Z_1)$, then $c' = (\ell + 1, z_0, \max\{z_1 - 1, 0\})$.
    \item If $\Isf_\ell = \mathsf{If}\ Z_0 = 0\ \mathsf{goto}\ \ell'\ \mathsf{else}\ \ell''$, then $c' = (\ell', z_0, z_1)$ if $z_0 = 0$, and $c' = (\ell'', z_0, z_1)$ if $z_0 > 0$.
    \item If $\Isf_\ell = \mathsf{If}\ Z_1 = 0\ \mathsf{goto}\ \ell'\ \mathsf{else}\ \ell''$, then $c' = (\ell', z_0, z_1)$ if $z_1 = 0$, and $c' = (\ell'', z_0, z_1)$ if $z_1 > 0$.
    \item If $\Isf_\ell = \mathsf{STOP}$, then $c$ has no successor configuration.
\end{itemize}

The unique computation of $\Mmc$ is a finite or infinite sequence of configurations $c_0 c_1 c_2 \dots$ such that $c_0 = (1, 0, 0)$ and $c_i \vdash c_{i+1}$ for all $i \geq 0$.
Observe that the computation is finite if and only if the instruction $(k : \mathsf{STOP})$ is reached. If this is the case, we say $\Mmc$ terminates.
Given a two-counter machine $\Mmc$, it is undecidable to decide whether $\Mmc$ terminates~\cite{minsky}.

In the following we assume without loss of generality that our two-counter machines satisfy the guarded-decrement property~\cite{infiniteZimmermann}, which guarantees that every decrement does indeed change a counter value: every decrement $(\ell : \mathsf{Dec}(Z_i))$ is preceded by a zero-test of the form $(\ell - 1, \mathsf{If}\ Z_i = 0\ \mathsf{goto}\ \ell+1\ \mathsf{else}\ \ell)$.
Note that this modification does not change the termination behavior of a two-counter machine, as decrementing a counter whose value is already zero does not have an effect.

\begin{lemma}
\label{lem:interReset}
    The intersection emptiness problem for reset~PA is undecidable.
\end{lemma}
\begin{proof}
We can encode infinite computations of two-counter machines as infinite words over $\Sigma = \{a,b, 1, 2, \dots, k\} \cup \Sigma_\Isf$, where $\Sigma_\Isf = \{I_a, I_b, D_a, D_b, Z_a, Z_b, \bar{Z}_a, \bar{Z}_b\}$. The idea is as follows. Let $c = (\ell, z_0, z_1)$ be a configuration of $\Mmc$. We encode $c$ as a finite word $w_c = \ell u x \in \Sigma^*$, where $\ell \in \{1,2, \dots, k\}$ encodes the current line number, $u \in \{a,b\}^*$ with $|u|_a = z_0$ and $|u|_b = z_1$ encodes the current counter values, and $x \in \Sigma_\Isf$ encodes the instruction $\Isf_\ell$ of line $\ell$ as follows: 
\begin{itemize}
    \item If $\Isf_\ell = \mathsf{Inc}(Z_0)$, then $x = I_a$, and if $\Isf_\ell = \mathsf{Inc}(Z_1)$, then $x = I_b$.
    \item If $\Isf_\ell = \mathsf{Dec}(Z_0)$, then $x = D_a$, and if $\Isf_\ell = \mathsf{Dec}(Z_1)$, then $x = D_b$.
    \item If $\Isf_\ell = \mathsf{If}\ Z_0 = 0\ \mathsf{goto}\ \ell'\ \mathsf{else}\ \ell''$, and the line number of the unique successor configuration of $c$ is $\ell'$, then $x = Z_a$ (that is, the zero-test is successful). Analogously with $x = Z_b$.
    \item If $\Isf_\ell = \mathsf{If}\ Z_0 = 0\ \mathsf{goto}\ \ell'\ \mathsf{else}\ \ell''$, and the line number of the unique successor configuration of $c$ is $\ell''$, then $x = \bar{Z}_a$ (that is, the zero-test fails). Analogously with $x = \bar{Z}_b$.
\end{itemize}

Let $w_c, w_{c'} \in \Sigma^*$ be two words encoding two configurations of~$\Mmc$. We call $w_c \cdot w_{c'}$ \emph{correct} if $c \vdash c'$. Hence, we can encode a unique infinite computations $c_0 c_1 c_2 \dots$ as an infinite word $w_{c_0} w_{c_1} w_{c_2} \dots$. We show that the $\omega$-language $L = \{w_{c_0} w_{c_1} w_{c_2} \dots \}$ can be written as the intersection of two deterministic strong reset~PA $\omega$-languages. Let
\begin{align*}
L_1 &= \{w_{c_0} w_{c_1} w_{c_2} \dots \mid w_{c_{2i}} w_{c_{2i+1}}\text{ is correct for every } i \geq 0\}, \text{ and} \\
L_2 &= \{w_{c_0} w_{c_1} w_{c_2} \dots \mid w_{c_{2i+1}} w_{c_{2i+2}} \text{ is correct for every } i \geq 0\}. 
\end{align*}

Observe that $L_1 \cap L_2 = L$, and $L$ is empty if and only if the unique computation of $\Mmc$ terminates.
Hence, it remains to show that $L_1$ and $L_2$ are recognized by deterministic strong reset~PA.
We argue for~$L_1$; the argument for $L_2$ is very similar.
The idea is as follows: We construct a deterministic strong reset~PA $\Amc_1$ with five counters that tests the correctness of two consecutive encodings of configurations, say $w_{c_{2i}} \cdot w_{c_{2i + 1}} = \ell_1 u_1 x_1 \cdot \ell_2 u_2 x_2$ with $\ell_1, \ell_2 \in \{1, 2, \dots, k\}$, $u_1 u_2 \in \{a,b\}^*$ and $x_1, x_2 \in \Sigma_\Isf$.
First observe that checking whether $\ell_2$ is indeed the correct line number (that is, the correct successor of~$\ell_1$) can be hard-coded into the state space of $\Amc_1$: if $x_1$ encodes an increment or decrement, we expect $\ell_2 = \ell_1 + 1$, and if $x_1$ encodes a successful or failing zero-test $\mathsf{If}\ Z_i = 0\ \mathsf{goto}\ \ell'\ \mathsf{else}\ \ell''$, we expect $\ell_2 = \ell'$ or $\ell_2 = \ell''$, respectively.
Four counters of $\Amc_1$ are used to count the numbers of $a$'s and $b$'s in~$u_1$ and~$u_2$, respectively. Then, if $x_1 = I_a$, we expect $|u_2|_a = |u_1|_a + 1$ and $|u_2|_b = |u_1|_b$, and so on.
To be able to perform the correct check, we also encode $x_1$ into the state space as well as the fifth counter by counting modulo~$|\Sigma_\Isf|$.
Observe that the guarded-decrement property ensures that decrements are handled correctly.
Hence, $\Amc_1$ has two sets of states counting the numbers of $a$'s and $b$'s of~$u_1$ and~$u_2$, accordingly, as well as a set of accepting states that is used to check the counter values. 

After such a check, $\Amc_1$ resets, and continues with the next two (encodings of) configurations. The automaton for $L_2$ works in the same way, but skips the first configuration. 
\end{proof}

%% file: 4.2_auto_det.tex
\section{Deterministic Parikh Automata on Infinite Words}
\label{sec:det}
In this section, we study the deterministic variants of Parikh automata on infinite words.
First, we study the closure properties of the newly introduced models, yielding the foundation in order to investigate the expressiveness of the models.
The main part of this section is devoted to the decision problems; in particular we focus on the core problems for model checking that, in contrast to the non-deterministic variants, are decidable for the deterministic variants and present algorithms for these problems.

\subsection{Closure Properties}
\label{sec:closure}

%\textcolor{red}{We can include the closure properties of the non-deterministic models.}

We now study the closure properties of the deterministic variants of the models introduced by Grobler et al.~\cite{grobler2023remarks}, that is, deterministic limit~PA, deterministic reachability-regular~PA, deterministic strong reset~PA, and deterministic weak reset~PA.

It is well known that semi-linear sets over $\Nbb^d$ are closed under complement~\cite{semilin}; see also~\cite{haase}.
Before we study deterministic limit automata we show that this is also true for semi-linear sets enriched with $\infty$. 

\begin{lemma}\label{lem:complement-semi-linear-with-infty}
	Let $C \subseteq \Nbbinfty^d$ be a semi-linear set. Then the complement $\bar{C} = \Nbbinfty^d \setminus C$ is semi-linear.
\end{lemma}
\begin{proof}
	Let $f : \Nbbinfty \rightarrow\Nbb$ be the bijection with $f(\infty) = 0$ and $f(i) = i+1$ for $i \in \Nbb$.
	We extend $f$ to vectors \mbox{$\vbf = (v_1, \dots, v_d) \in \Nbbinfty^d$} and sets of vectors $C \subseteq \Nbbinfty^d$ component-wise: $f(v_1, \dots, v_d) = (f(v_1), \dots, f(v_d))$ and $f(C) = \{f(\vbf) \mid \vbf \in C\}$. Note that \mbox{$f(C)\subseteq \Nbb^d$}. 
	
	Now we observe that $f(C)$ is semi-linear if and only if $C$ is semi-linear. First assume that $C$ is semi-linear. We may assume that $C = C(\bbf, P)$ is linear, as we can carry out the following procedure for every linear set individually. Assume $\bbf = (b_1, \dots, b_d)$. We define $D_\infty(\bbf) = \{i \mid b_i = \infty\}$.
	For a set $D \subseteq \{1, \dots, d\}$ with $D_\infty(\bbf) \subseteq D$ and a vector $\vbf = (v_1, \dots, v_d) \in \Nbbinfty^d$, let $\vbf^D = (v_1^D, \dots, v_d^D)$ with $v_i^D = 0$ if $i \in D$ and $v_i^D = v_i$ if $i \notin D$. 
	Furthermore, we call a subset $P' \subseteq P$ of period vectors \emph{$D$-compatible} if for every $i \in D \setminus D_\infty(\bbf)$, the set $P'$ contains at least one vector where the $i$th component is~$\infty$, and if for every $i \notin D$, the set $P'$ contains no vector where the $i$th component is $\infty$. 
	Observe that this definition ensures that for every vector $\vbf \in P'$ we have $\vbf^D \in \Nbb^d$, that is, no component in $\vbf^D$ is $\infty$. 
	Let $\Pmc^D_\infty \subseteq 2^P$ be the collection of $D$-compatible subsets of~$P$.
	%Then we have $f(C) = \bigcup_{D_\infty(\bbf) \subseteq D \subseteq \{1, \dots, d\}} \bigcup_{P \in \Pmc^D_\infty} \{\bbf^D + \1^D + \sum_{\pbf \in P} \pbf^D + \sum_{\pbf \in P} \pbf^D z_\pbf \mid z_\pbf \in \Nbb, \pbf\in P\}$ (where $\1$ denotes the all-one vector),
	Then we have 
	\[f(C) = \bigcup_{D_\infty(\bbf) \subseteq D \subseteq \{1, \dots, d\}} \bigcup_{P' \in \Pmc^D_\infty} C(\bbf^D + \1^D + \sum_{\pbf \in P'} \pbf^D, \{\pbf^D \mid \pbf \in P'\}),\]
	which is semi-linear by definition.
	
	For the other direction, we may again assume that $f(C) = C(\bbf, P)$ is linear.
	Similar to above, assume $\bbf = (b_1, \dots, b_d)$ and define $D_0(\bbf) = \{i \mid b_i = 0\}$.
	For a set $D \subseteq D_0(\bbf)$ we call a subset $P' \subseteq P$ of period vectors \emph{$D$-safe} if for every $i \in D$ the set $P'$ contains no vector where the $i$th component is not 0, and if for every $i \notin D$ the set $P'$ contains at least one vector where $i$th component is not 0.
	Let $\Pmc^D_0 \subseteq 2^P$ be the collection of $D$-safe subsets of $P$.
	%For every subset $D \subseteq D'$ we define $C_D$ to be the subset of $C$ such that all vectors in $C_D$ have $\infty$ at position $i$ for all $i \in D$.
	%Similar to above, for a subset $D \subseteq D'$ let $B^D_{>0}$ be the collection of sets of period vectors such that every set contains at least one period vector with a non-zero entry at position $i$ for every $i \in D$. Furthermore, let $B_{|D}$ be the subset of period vectors with only zero entries at position $i$ for every $i \in D$. 
	Let $\ibf_D = (v_1, \dots, v_d)$ with $v_i = \infty$ if $i \in D$ and $v_i = 0$ if $i \notin D$.
	Then we have 
	\[C = \bigcup_{D \subseteq D_0(\bbf)} \bigcup_{P' \in \mathcal{P}^D_0} C(\ibf_D + \bbf - \1 + \sum_{\pbf \in P'} \pbf, P'),\] 
	which is semi-linear by definition (we assume $\infty-1=\infty$). Observe that every component in $\ibf_D + \bbf + \sum_{\pbf \in P'} \pbf$ is strictly greater 0, hence we can subtract $\1$ without getting negative.
	
	As semi-linear sets over $\Nbb^d$ are closed under complement, we have 
	\begin{align*}
		\text{$C$ is semi-linear} & \Leftrightarrow \text{$f(C)$ is semi-linear} \\
		& \Leftrightarrow \text{$\Nbb^d \setminus f(C)$ is semi-linear}\\
		& \Leftrightarrow \text{$f^{-1}(\Nbb^d \setminus f(C)) = \overline{C}$ is semi-linear}.
	\end{align*}
\end{proof}

\begin{lemma}\label{lem:closure-limit}
	The class of deterministic limit~PA recognizable languages is closed under union, intersection and complement. 
\end{lemma}
\begin{proof}
	First observe that we can always assume that every state of a (deterministic) limit~PA is accepting, as we can check the existence of an accepting state being visited infinitely often in the semi-linear set. To achieve this, we introduce one new counter and increment it at every transition that points to an accepting state. In the semi-linear set we enforce that this counter is $\infty$.
	
	Hence, we can show the closure under union and intersection by a standard product construction. In case of union, we test whether at least one automaton has good counter values, and we can show the closure under intersection by testing whether both automata have good counter values.
	Closure under complement follows immediately from \Cref{lem:complement-semi-linear-with-infty}.
\end{proof}

Following the standard proof showing that $L_{a<\infty} = \{\alpha \in \{a,b\}^\omega \mid |\alpha|_a < \infty\}$ is not deterministic $\omega$-regular, we make 
%and make constantly use of 
the following observation. 

%The proof that these models do not recognize $L_{a<\infty}$ mimics the standard proof showing that this $\omega$-language is not deterministic $\omega$-regular, see \mbox{\eg~\cite{thomasinfinite}}.

\begin{observation}
	\label{lem:not_all_reg}
	There is no deterministic reachability-regular~PA, deterministic weak reset~PA or deterministic strong reset~PA recognizing the $\omega$-regular language $L_{a<\infty}$.
	%$ = \{\alpha \in \{a,b\}^\omega \mid |\alpha|_a < \infty\}$.
\end{observation}

Observe however, that the complement $L_{a=\infty} = \{\alpha \in \{a,b\}^\omega \mid |\alpha|_a = \infty\}$ of $L_{a<\infty}$ is recognized by all of these models.
%\pagebreak

\begin{lemma}\label{lem:non-closure-reach-reg}
	The class of deterministic reachability-regular~PA recognizable languages is not closed under union, intersection or complement. 
\end{lemma}
\begin{proof}
	First we show non-closure under union.
	Let 
	$$L_1 = \{u c \alpha \mid u \in \{a,b,c\}^*, |u|_a = |u|_b, |\alpha|_c = \infty\}$$ and 
	$$L_2 = \{v a \beta \mid v \in \{a,b,c\}^*, |v|_b = |v|_c, |\beta|_a = \infty\}.$$
	
	Both languages are deterministic reachability-regular~PA recognizable, as witnessed by the automaton in \Cref{fig:limitunion} and the fact that $L_2$ can be obtained from $L_1$ by shuffling symbols. 
	We show that the language $L_1 \mathop{\cup} L_2$ is not deterministic reachability-regular~PA recognizable.
	Assume there is an $n$-state deterministic reachability-regular~PA $\Amc$ recognizing $L_1 \cup L_2$.
	Then there is a unique accepting run $r_1 r_2 \dots$ of $\Amc$ on $abc^\omega$. In particular, at some point the automaton verifies the Parikh condition, say after using the transition $r_i$. Let $m = \max\{n+1, i\}$ and consider the infinite word $abc^m b^{m-1} a^\omega$
	with the unique accepting run $r'_1 r'_2 \dots$ of $\Amc$. Due to determinism, we have $r_j = r'_j$ for all $j \leq m+2$; in particular, the automaton verifies the Parikh condition within this run prefix. As $m-1 \geq n$, the automaton visits a state, say $q$ twice while reading~$b^{m-1}$. Hence, we can pump this $q$-cycle and obtain an accepting run of~$\Amc$ on $abc^m b^{m-1+k} a^\omega$ for some $k > 0$, a contradiction. 
	
	For the non-closure under intersection, define 
	\[L_1 = \{\alpha \mid |\alpha[1:i]|_a = |\alpha[1:i]|_b \text{ for some }i \}\] 
	and 
	\[L_2 = \{\alpha \mid |\alpha[1:i]|_a = |\alpha[1:i]|_c \text{ for some }i \}.\]
	Suppose there is an $n$-state deterministic reachability-regular~PA recognizing $L_1 \cap L_2$. Let $\alpha = a(a^n b^n)^{n+1} c^{n(n+1) + 1} a^\omega$. The unique run~$r$ of $\Amc$ on $\alpha$ is not accepting, as $\alpha$ has no balanced $a$-$b$ prefix.
	Observe that~$\Amc$ visits at least one state twice while reading a $b^n$-block.
	Furthermore, there is a state, say~$q$, such that~$\Amc$ visits~$q$ twice while reading two different $b^n$-blocks, as there are $n+1$ different $b^n$-blocks.
	Hence, we can swap the latter~$q$-cycle to the front and obtain an infinite word, say $\alpha'$ in $L_1 \cap L_2$, with a unique accepting run~$r'$ of $\Amc$. This run verifies the Parikh condition at some point. We distinguish two cases.
	If~$r'$ verifies the Parikh condition before reading the first $c$, we can depump the $c^{n(n+1)+1}$-block and obtain an accepting run on an infinite word without a balanced $a$-$c$-prefix, a contradiction.
	Hence assume that~$r'$ verifies the Parikh condition after reading at least one~$c$, say at position~$k$. However, then we have $\rho(r[1:k]) = \rho(r'[1:k])$, and hence $\Amc$ also accepts $\alpha$, a contradiction.
	
	The non-closure under complement follows immediately from \Cref{lem:not_all_reg}.
	\begin{figure}
		\centering
		\begin{tikzpicture}[->,>=stealth',shorten >=1pt,auto,node distance=3.5cm, semithick]
			\tikzstyle{every state}=[minimum size=1cm]
			
			\node[state, initial, initial text={}] (q0) {$q_0$};	
			\node[state, accepting] (q1) [right of=q0] {$q_1$};
			
			\path
			(q0) edge [loop above, align=center] node {$a, \begin{pmatrix}1\\0\end{pmatrix};$\ \ $b, \begin{pmatrix}0\\1\end{pmatrix}$} (q0)
			(q0) edge [bend left]  node {$c, \0$} (q1)
			(q1) edge [loop above] node {$c, \0$} (q1)
			(q1) edge [bend left, align=center] node {$a, \begin{pmatrix}1\\0\end{pmatrix}$;\ \ $b, \begin{pmatrix}0\\1\end{pmatrix}$} (q0)	
			;
		\end{tikzpicture}
		\caption{The deterministic reachability-regular~PA with $C = C(\0, \{\1\})$ \\ for $L_1 = \{u c \alpha \mid u \in \{a,b,c\}^*, |u|_a = |u|_b, |\alpha|_c = \infty\}$.}
		\label{fig:limitunion}
	\end{figure}
\end{proof}

\begin{lemma}\label{lem:non-closure-weak-reset}
	The class of deterministic weak reset~PA recognizable languages is not closed under union, intersection or complement. 
\end{lemma}
\begin{proof}
	We begin with the non-closure under union. Let 
	\[L_{a=b} = \{\alpha \mid |\alpha[1:i]|_a = |\alpha[1:i]|_c \text{ for $\infty$ many }i \}\] and similarly define 
	\[L_{a=c} = \{\alpha \mid |\alpha[1:i]|_a = |\alpha[1:i]|_c \text{ for $\infty$ many }i.\]
	We show that $L_{a=b} \cup L_{a=c}$ is not deterministic weak reset~PA recognizable.
	Assume there is a deterministic weak reset~PA $\Amc$ recognizing $L_{a=b} \cup L_{a=c}$. Consider the unique accepting run $r$ of $\Amc$ on $\alpha = (ab)^\omega$ and let $i,j$ be two positions with $i+1 < j$ such that~$r$ resets after reading $\alpha[1:i]$ and $\alpha[1:j]$ in the same state (such a pair of positions does always exists by the infinite pigeonhole principle).
	Now consider the infinite word $\alpha[1:i] c^{|\alpha[1:i]|_a} (ac)^\omega$, which is also accepted by $\Amc$. However, this implies that $\Amc$ also accepts $\alpha[1:j] c^{|\alpha[1:i]|_a} (ac)^\omega$, as~$\Amc$ is in the same (accepting) state after reading $\alpha[1:i]$ as well as $\alpha[1:j]$, but this infinite word is not contained in $L_{a=b} \cup L_{a=c}$, as $\alpha[1:j]$ contains at least one more $a$ than $\alpha[1:i]$, a contradiction.    
	
	The argument for the non-closure under intersection is the same as for the non-deterministic setting, as the $\omega$-languages considered in~\Cref{lem:intersectionNondet} are indeed deterministic weak reset~PA recognizable.
	
	The non-closure under complement is an immediate consequence of~\Cref{lem:not_all_reg}.
\end{proof}

\begin{lemma}\label{lem:non-closure-strong-reset}
	The class of deterministic strong reset~PA recognizable languages is not closed under union, intersection or complement. 
\end{lemma}
\begin{proof}
	We begin with the non-closure under union. Let $L = \{c^* a^n c^* b^n \mid n > 0\}^\omega$ and $L_{c=\infty} = \{\alpha \mid |\alpha|_c = \infty\}$.
	Observe that $a^n c^\omega \in L \cup L_{c=\infty}$ for every $n \geq 0$.
	Assume there is a deterministic strong reset~PA recognizing $L \cup L_{c=\infty}$. Let $n_1 \neq n_2$ be such that the unique accepting runs of $\Amc$ on $\alpha_1 = a^{n_1} c^\omega$ resp.\ $\alpha_2 = a^{n_2} c^\omega$ reset in the same state the first time they reset after reading at least one~$c$, say after reading $\alpha_1[1:i_1]$ resp.\ $\alpha_2[1:i_2]$ (with $i_1 > n_1$ and $i_2 > n_2$).
	As $a^{n_1} c^{i_1-n_1} b^{n_1} (ab)^\omega$ is also accepted by $\Amc$, the infinite word $a^{n_2} c^{i_2 - n_2} b^{n_1} (ab)^\omega$ is also accepted by $\Amc$, a contradiction.
	
	To show the non-closure under intersection, we use an argument similar to the non-deterministic setting. Let \mbox{$L_1 = \{a^n b^n \mid n > 0\}^\omega$} and $L_2 = \{a\} \{b^n a^{2n} \mid n > 0\}$. Then $L_1 \cap L_2$ contains only one infinite word, namely $ab a^2 b^2 a^4 b^4 \dots$. Hence $L_1 \cap L_2$ is not ultimately periodic and hence not accepted by any strong reset~PA~\cite{grobler2023remarks}.
	
	The non-closure under complement again follows from~\Cref{lem:not_all_reg}.
\end{proof}

\subsection{Expressiveness}
\label{sec:expressiveness}

In this section we study the expressiveness of deterministic~PA on infinite words for those models whose deterministic variants were not studied before in the literature.

First we remark that deterministic reachability-regular~PA, deterministic limit~PA, deterministic strong reset~PA and deterministic weak reset~PA are strictly weaker than their non-deterministic counterparts. This follows immediately from their different closure properties: reachability-regular~PA, weak reset~PA (and hence strong reset~PA) are closed under union (see \Cref{cor:union}), and limit~PA are not closed under complement (see \Cref{lem:intersectionNondet}).
Hence, from the results of the previous section we obtain the following corollary.

\begin{corollary}
	The following strict inclusions hold.
	\begin{itemize}
		\item Deterministic reachability-regular~PA $\subsetneq$ Reachability-regular~PA.
		\item Deterministic limit~PA $\subsetneq$ Limit~PA.
		\item Deterministic strong reset~PA $\subsetneq$ Reset~PA.
		\item Deterministic weak reset~PA $\subsetneq$ Reset~PA.
	\end{itemize}
\end{corollary}

In the following,
when we show non-inclusions, we always give the strongest separation, \eg, when showing that a deterministic strong reset~PA cannot simulate a deterministic weak reset~PA, we show that it can not even simulate a deterministic reachability~PA, which is weaker than a deterministic weak reset~PA. We refer to \Cref{fig:inclusions} for an overview of the results in this section. We simply write that a model is a strict/no subset of another model; by that we mean that the class of $\omega$-languages recognized by the first model is a strict/no subset of the class of $\omega$-languages recognized by the second model.

\begin{figure}
	\centering
	\begin{tikzpicture}[%
		node distance=27mm,>=Latex,
		initial text="", initial where=below left,
		every state/.style={rectangle,rounded corners,draw=black,thin,fill=black!5,inner sep=1mm,minimum size=6mm},
		every edge/.style={draw=black,thin}
		]
		\node[state] (reach) {rechability~PA};
		\node[state, right = 2cm of reach, align=center]  (reachreg) {reachability-regular~PA};
		\node[state, right= 2cm of reachreg] (wreset) {weak reset~PA};
		
		\node[state,below = 1cm of reach] (reg) {(non-det.) $\omega$-regular};
		\node[state,right = 2cm of reg] (limit) {limit~PA};
		\node[state,right = 2cm of limit] (sreset) {strong reset~PA};
		
		\node[state,below = 1cm of reg] (safety) {safety~PA};
		\node[state,right = 2cm of safety] (cobuchi) {co-Büchi~PA};
		\node[state,right = 2cm of cobuchi] (buchi) {Büchi~PA};
		
		\path[-{Latex}]
		(reach) edge (reachreg)
		(reachreg) edge (wreset)
		(sreset) edge (wreset)
		(reg) edge (limit)
		;

	\end{tikzpicture}        
	\caption{Inclusion diagram of the studied \emph{deterministic} models. Arrows indicate strict inclusions while no connections mean incomparability.}
	\label{fig:inclusions}
\end{figure}

\subsubsection{\boldmath\texorpdfstring{$\omega$}{omega}-regular languages}

We begin by showing that every $\omega$-regular language is deterministic limit~PA recognizable. 

\begin{lemma}
	%    The class of $\omega$-regular languages is a strict subclass of the class of deterministic limit~PA recognizable $\omega$-languages.
	$\omega$-regular $\subsetneq$ deterministic limit~PA.
\end{lemma}
\begin{proof}
	Let $L \subseteq \Sigma^\omega$ be $\omega$-regular and let $\Amc = (Q, \Sigma, q_0, \Delta, \Fmc)$ be a deterministic Muller automaton recognizing $L$.
	The idea is to construct an equivalent deterministic limit~PA $\Amc' = (Q, \Sigma, q_0, \Delta', Q, C)$ of dimension $|Q|$, where every state is accepting, while encoding the sets in $\Fmc$ into the semi-linear set $C$.
	Let $f : Q \rightarrow \{1, \dots, |Q|\}$ be a bijection associating every state with a counter. Hence, we define $\Delta' = \{(p, a, \ebf_{f(q)}^{|Q|}, q) \mid (p, a, q) \in \Delta\}$.
	%where $\ebf_{f(q)}^{|Q|}$ is the $|Q|$-dimensional vector where the $f(q)$-th component is $1$, and every other entry is 0. 
	%We define $C$ as follows.
	%For every $F \in \Fmc$, define $C_F = \{ \sum_{q \in F} \ibf_{f(q)} + \sum_{q \notin F} \ebf_{f(q)} z_q \mid z_q \in \Nbb, q \notin F\}$, where $\ibf_{f(q)}$ is the vector where the $f(q)$-th component is $\infty$, and every other entry is 0. 
	For every $F \in \Fmc$, we define $C_F = C(\sum_{q \in F} \ibf_{f(q)}, \{\ebf_{f(q)} \mid q \notin F\})$.
	That is, for every state in $F$ we expect its counter value to be $\infty$, while we expect every other counter value to be a finite number. We choose $C = \bigcup_{F \in \Fmc}C_F$ and hence obtain an equivalent deterministic limit~PA.
	
	The strictness is witnessed by the $\omega$-language $\{a^n b^n c^\omega \mid n > 0\}$, which is obviously deterministic limit~PA recognizable, but not $\omega$-regular.
\end{proof}

\Cref{lem:not_all_reg} immediately yields the following result.

\begin{cor}
	%    The class of $\omega$-regular languages is not a subclass of any of the the classes of $\omega$-languages recognized by deterministic reach-reg.~PA, deterministic Büchi~PA, deterministic strong reset~PA, or deterministic weak reset~PA.
	%$\omega$-regular $\not\subseteq$ deterministic reach-reg.~PA, deterministic Büchi~PA, deterministic strong reset~PA, deterministic weak reset~PA
	\mbox{}\\[1mm]
	\begin{tikzpicture}
		\node[anchor=west] at (-0.45,0) {$\omega$-regular};

		\node at (1.6, 0) {$\not\subseteq$};
		
		\draw [decorate, thick,
		decoration = {calligraphic brace, raise = 2pt, amplitude = 4pt,mirror}] (2.1,0.8) --  (2.1,-0.8);
		
		\node[anchor=west] at (2.2,0.65) {Deterministic reachability-regular~PA};
		\node[anchor=west] at (2.2,0.25) {Deterministic Büchi~PA};
		\node[anchor=west] at (2.2,-0.25) {Deterministic strong reset~PA};
		\node[anchor=west] at (2.2,-0.65) {Deterministic weak reset~PA};
	\end{tikzpicture}
\end{cor}

%\begin{proof}
%    None of these~PA recognize $L_{a<\infty} = \{\alpha \in \{a,b\}^\omega \mid |\alpha|_a < \infty\}$, which is obviously $\omega$-regular. The proof that these models do not recognize $L_{a<\infty}$ mimics the standard proof showing that this $\omega$-language is not deterministic $\omega$-regular, see \mbox{\eg~\cite{thomasinfinite}}.
%\end{proof}
Observe however that these models generalize deterministic Büchi automata. This is not true for deterministic reachability~PA, deterministic safety~PA nor deterministic \mbox{co-Büchi~PA}, as shown in the next lemma. 

\begin{lemma}\label{lem:safety-cobucki}
	%The class of deterministic $\omega$-regular languages is not a subclass of any of the the classes of $\omega$-languages recognized by deterministic reach~PA, deterministic safety~PA or deterministic \mbox{co-Büchi~PA}.
	%Deterministic $\omega$-regular $\not\subseteq$ deterministic reach~PA, deterministic safety~PA, deterministic \mbox{co-Büchi~PA}
	\mbox{}\\[1mm]
	\begin{tikzpicture}
		\node[anchor=west] at (-1,0) {Deterministic $\omega$-regular};

		\node at (3.4, 0) {$\not\subseteq$};
		
		\draw [decorate, thick,
		decoration = {calligraphic brace, raise = 2pt, amplitude = 4pt,mirror}] (3.9,0.66) --  (3.9,-0.66);
		
		\node[anchor=west] at (4,0.47) {Deterministic reachability~PA};
		\node[anchor=west] at (4,0) {Deterministic safety~PA};
		\node[anchor=west] at (4,-0.4) {Deterministic co-Büchi~PA};
	\end{tikzpicture}
\end{lemma}
\begin{proof}
	As an immediate consequence of \Cref{lem:char-det-reach} (proved below) we have that no deterministic reachability~PA recognizes the deterministic \mbox{$\omega$-regular} language $a^* b^\omega$.
	
	The two other claims follow from~\cite[Proof of Theorem 3]{infiniteZimmermann}, where the authors have shown that (even non-deterministic) safety~PA do not recognize the deterministic $\omega$-regular language \mbox{$\{a,b\}^\omega \setminus \{a\}^\omega$} and that no co-Büchi~PA recognizes the deterministic $\omega$-regular language $L_{a=\infty}=\{\alpha \in \{a,b\}^\omega \mid |\alpha|_a = \infty\}$.
\end{proof}

\subsubsection{Deterministic Safety~PA and co-Büchi~PA}

As a consequence of \Cref{lem:safety-cobucki} we obtain the following corollary. 

\begin{cor}
	\mbox{}\\[1mm]
	\begin{tikzpicture}
		\node[anchor=west] at (-.7,-.2) {Deterministic reachability-regular~PA};
		\node[anchor=west] at (-.7,-0.6) {Deterministic limit~PA};
		%\node[anchor=west] at (0,-0.8) {Deterministic Büchi~PA};
		\node[anchor=west] at (-.7,-1.1) {Deterministic strong reset~PA};
		\node[anchor=west] at (-.7,-1.5) {Deterministic weak reset~PA};
		
		\draw [decorate,thick, 
		decoration = {calligraphic brace, raise = 2pt, amplitude = 4pt}] (5.9,0) --  (5.9,-1.7);
		
		\node at (6.4, -0.85) {$\not\subseteq$};
		
		\draw [decorate, thick,
		decoration = {calligraphic brace, raise = 2pt, amplitude = 4pt,mirror}] (6.9,-0.4) --  (6.9,-1.3);
		
		\node[anchor=west] at (6.9,-0.6) {Deterministic safety~PA};
		\node[anchor=west] at (6.9,-1.1) {Deterministic co-Büchi~PA};
	\end{tikzpicture}
\end{cor}

As shown in \cite{infiniteZimmermann} also deterministic reachability~PA $\not\subseteq$ deterministic safety~PA and deterministic reachability~PA $\not\subseteq$ deterministic co-Büchi~PA as well as deterministic Büchi~PA $\not\subseteq$ deterministic safety~PA and deterministic Büchi~PA $\not\subseteq$ deterministic co-Büchi~PA. Furthermore, the classes of deterministic safety~PA and deterministic co-Büchi~PA are themselves incomparable as shown in~\cite{infiniteZimmermann}. 

Vice versa, deterministic safety~PA $\not\subseteq$ non-deterministic weak reset~PA and deterministic co-Büchi~PA $\not\subseteq$ non-deterministic weak reset~PA~\cite{grobler2023remarks}. Hence, these classes are no subclasses of any of the other studied classes. 

Overall, deterministic safety~PA and deterministic co-Büchi~PA are incomparable with all other studied models. 

% As a consequence of this result and the results in~\cite{grobler2023remarks}, namely, that not even non-deterministic strong (equivalently non-determin\-istic weak) reset~PA generalize safety~PA nor co-Büchi~PA, and the separations shown in~\cite{infiniteZimmermann}, we obtain that deterministic safety~PA are incomparable to every other model, and determinisic co-Büchi are also incomparable to every other model.

\subsubsection{Deterministic Reachability~PA}

We begin by characterizing the class of deterministic reachability~PA recognizable $\omega$-languages.
\begin{lemma}\label{lem:char-det-reach}
	An $\omega$-language $L$ is deterministic reachability~PA recognizable if and only if $L = U \Sigma^\omega$, where $U \subseteq \Sigma^*$ is recognized by a deterministic~PA.
\end{lemma}
\begin{proof}
	Let $\Amc$ be a deterministic reachability~PA recognizing $L$. Then we have $L(\Amc) = U$. Likewise, if $\Amc$ is a~PA recognizing~$U$, then \mbox{$R_\omega(\Amc) = L$} (recall that \Amc is complete by the definition of determinism). 
\end{proof}

We have the following strict inclusion. 

\begin{lemma}
	\label{lem:reachToReachReg}
	%    The class of deterministic reach~PA recognizable $\omega$-lan\-guages is a strict subclass of the class of deterministic reach-reg.~PA recognizable \mbox{$\omega$-languages}.
	Deterministic reachability~PA $\subsetneq$ deterministic reachability-regular~PA.
\end{lemma}
\begin{proof}
	Let $\Amc$ be a deterministic reachability~PA. We may assume that every state of $\Amc$ is accepting, as we can project the current state into the semi-linear set. To be precise, we introduce two new counters for each state of $\Amc$, counting the number of visits and exits. Then, the current state is the (unique) state with one more visit than exit, or in case that the number of visits and exits is the same for every state, then the current state is the initial state of $\Amc$.
	As these statements can be encoded into a semi-linear set, we can assume that every state is equipped with its own semi-linear set, and can hence make every state accepting (and assign the empty set if we want to simulate a non-accepting state).
	If every state is accepting, then $\Amc$ is an equivalent deterministic reachability-regular~PA.
	
	The strictness is witnessed \eg, by the $\omega$-language $\{a^nb^na^\omega\mid n>0\}$, which is deterministic reachability-regular~PA recognizable and by \Cref{lem:char-det-reach} not deterministic reachability~PA recognizable.
\end{proof}

It remains to show the following incomparability results. 

\begin{lemma}
	\label{lem:reach-no-limit}
	%    The class of deterministic reachability~PA recognizable $\omega$-languages is no subclass of the class of deterministic limit~PA recognizable $\omega$-languages.
	Deterministic reachability~PA $\not\subseteq$ deterministic limit~PA.
\end{lemma}
\begin{proof}
	We show that the deterministic reachability~PA recognizable $\omega$-language
    \[L = \{\alpha \mid |\alpha[1:i]|_a = |\alpha[1:i]|_b \text{ for some }i > 0\}\]
    is not deterministic limit~PA recognizable. The proof is similar \cite[Lemma 3]{infiniteZimmermann}.
	Assume there is an $n$-state deterministic limit~PA $\Amc$ recognizing~$L$. Consider the unique non-accepting run of $\Amc$ on $a (a^n b^n)^\omega$.
	Observe that $\Amc$ visits at least one state twice while reading a $b^n$-block, and there are at least two of the (infinitely many) \mbox{$b^n$-blocks} such that $\Amc$ visits the same state, say $q$, twice while reading them. Hence, we can shift one such \mbox{$q$-cycle} to the front and obtain the unique run on an infinite word that is in~$L$. However, this run is still non-accepting, as the extended Parikh image and number of visits of an accepting state do not change.
\end{proof}

\begin{lemma}
	\label{lem:reach-no-strong-reset}
	%    The class of deterministic reachability~PA recognizable $\omega$-languages is no subclass of the class of deterministic strong reset~PA recognizable $\omega$-languages.
	Deterministic reachability~PA $\not\subseteq$ deterministic strong reset~PA.
\end{lemma}
\begin{proof}
	We show that the deterministic reachability~PA recognizable $\omega$-language 
	\[L = \{a^n b^n \mid n \geq 1\} \cdot \{a,b\}^\omega\]
	is not deterministic strong reset~PA recognizable. Assume there is a deterministic strong reset~PA~$\Amc$ with $n$ states recognizing $L$. Let $\alpha = a^n b^\omega$ with unique accepting run $r = r_1 r_2 r_3 \dots$ of $\Amc$ on~$\alpha$.
	Let $f_1, f_2, \dots$ be the sequence of reset positions of $r$ and let $i > n$ be minimal with $i = f_{i'}$ for some $i' \geq 1$ (that is, $f_{i'}$ is the first reset position after reading a~$b$).
	
	First observe that $i < 2n$. Assume that this is not the case. As~$\Amc$ visits at least one state twice while reading $b^n$, say state~$q$, we observe that $\Amc$ is caught in a $q \dots q$ cycle while reading $b^\omega$ due to determinism. That is, every state that is visited while reading~$b^\omega$ is already visited while reading the first $n$ many $b$s. Hence we have $i < 2n$.
	Now let $j \geq 2n$ be minimal such that $j = f_{j'}$ for some $j' > i'$ is a reset position in $r$ such that the state at position $f_{j'}$ is the same state as the one at position $f_{i'}$ (which exists by the same argument).
	
	Now let $\alpha' = a^n b^{j-n}a^\omega$ with unique accepting run $r' = r'_1 r'_2 r'_3 \dots$ of~$\Amc$ on $\alpha'$. 
	Observe that $\alpha[1:j] = \alpha'[1:j]$, and hence $r[1:j] = r'[1:j]$. 
	As the partial runs $r[1:i]$ and $r[1:j]$ reachability the same accepting state, the run $r[1:i]r'_{j+1}r'_{j+2} \dots$ is an accepting run of~$\Amc$ on $a^n b^{i-n} a^\omega$.
	However, as $i-n < n$, this infinite word is not contained in $L$, a contradiction.
\end{proof}

\subsubsection{Deterministic Reachability-Regular~PA}

We begin by showing that every deterministic reachability-regular~PA (and hence every deterministic reachability~PA) can be translated into an equivalent deterministic weak reset~PA.
\begin{lemma}
	\label{lem:reachRegToWeak}
	%The class of deterministic reach-reg.~PA recognizable $\omega$-lan\-guages is a strict subclass of the class of deterministic weak reset~PA recognizable $\omega$-languages.
	Deterministic reachability-regular~PA $\subsetneq$ deterministic weak reset~PA.
\end{lemma}
\begin{proof}
	Let $\Amc = (Q, \Sigma, q_0, \Delta, F, C)$ be a deterministic reachability-regular~PA and let $\Amc' = (Q \cup \{q_0'\}, \Sigma, q_0', \Delta', F, C')$ be a copy of $\Amc$ with a new fresh initial state $q_0'$ inheriting all outgoing transitions of $q_0$ (observe that this modification preserves determinism). We add one new counter that is incremented at every transition leaving $q_0'$, and not modified otherwise, that is, 
	\[
		\Delta' = \{(p, a, \vbf \cdot 0, q) \mid (p, a, \vbf, q) \in \Delta\}
		\cup \{(q_0', a, \vbf \cdot 1, q) \mid (q_0, a, \vbf, q) \in \Delta\}.
	\]
	
	We choose $C' = C \cdot \{1\} \cup \Nbb^d \cdot \{0\}$ and obtain an equivalent weak reset~PA $\Amc'$.
	
	The strictness is witnessed by the $\omega$-language \mbox{$\{a^n b^n \mid n > 0\}^\omega$}, which is obviously deterministic weak reset~PA recognizable, but not even recognized by (non-deterministic) Büchi~PA~\cite{infiniteZimmermann}, which are more expressive than reachability-regular~PA.
\end{proof}

\subsubsection{Deterministic Strong Reset~PA}

\begin{lemma}
\label{lem:detStrongResetVsdetWeakReset}
	%    The class of deterministic strong reset~PA recognizable $\omega$-lan\-guages is a strict subclass of the class of deterministic weak reset~PA recognizable $\omega$-lan\-guages.
	Deterministic strong reset~PA $\subsetneq$ deterministic weak reset~PA.
\end{lemma}
\begin{proof}
	The inclusion follows from \Cref{lem:SPBAtoWPBA}, as the construction preserves determinism.
	The strictness follows from the fact that $\{a^n b^n \mid n \geq 1\} \cdot \{a,b\}^\omega$ is deterministic reachability~PA recognizable, and hence deterministic weak reset~PA recognizable (by \Cref{lem:reachToReachReg} and \Cref{lem:reachRegToWeak}), but not recognized by any deterministic strong reset~PA, as shown in \Cref{lem:reach-no-strong-reset}.
\end{proof}

\begin{lemma}
	%    The class of deterministic strong reset~PA recognizable $\omega$-languages is no subclass of the class of deterministic Büchi~PA recognizable $\omega$-languages.
	Deterministic strong reset~PA $\not\subseteq$ deterministic Büchi~PA.
\end{lemma}
\begin{proof}
	The argument is as in \Cref{lem:reachRegToWeak}. The $\omega$-language $\{a^n b^n \mid n > 0\}^\omega$ is deterministic strong reset~PA recognizable, but there is no Büchi~PA recognizing it~\cite{infiniteZimmermann}.
	%This is witnessed by the $\omega$-language $\{a^n b^n \mid n > 0\}^\omega$, which is deterministic strong reset~PA recognizable but not even non-deterministic Büchi~PA recognizable. The latter claim follows from the fact that this language is not recognized by a blind counter automaton~\cite[Lemma 3.3]{blindcounter} and their equivalence to Büchi~PA~\cite[Lemma 21, Theorem 22]{grobler2023remarks} \textcolor{red}{we basically have this in Lemma 4.7}
\end{proof}

\begin{lemma}
	\label{lem:strong-reset-no-limit}
	%   The class of deterministic strong reset~PA recognizable $\omega$-languages is no subclass of the class of deterministic limit~PA recognizable $\omega$-languages.
	Deterministic strong reset~PA $\not\subseteq$ deterministic limit~PA.
\end{lemma}
\begin{proof}
	This follows from the previous proof as limit~PA are less expressive than Büchi~PA.
\end{proof}

\subsubsection{Deterministic Büchi~PA}
We show that $\omega$-languages recognized by deterministic Büchi~PA can be characterized in a similar way as deterministic $\omega$-regular languages.
\begin{lemma}
	An $\omega$-language $L$ is deterministic Büchi~PA recognizable if and only of $L = \vec{P}$ where $P$ is recognized by a deterministic~PA.
\end{lemma}
\begin{proof}
	Let $\Amc$ be a deterministic Büchi~PA recognizing $L$ and let $\alpha \in B_\omega(\Amc)$ with accepting run $r$. As $r$ has infinitely many accepting hits by definition, we have $\alpha \in \vec{L(\Amc)}$.
	Similarly, let $\Amc$ be a deterministic~PA recognizing $P$ and let $\alpha \in \vec{P}$. As $\Amc$ is deterministic, the unique run of $\Amc$ on $\alpha$ has infinitely many accepting hits, hence we have $\alpha \in B_\omega(\Amc)$.
\end{proof}

\begin{lemma}
	%   The class of deterministic Büchi~PA recognizable $\omega$-languages is no subclass of the class of deterministic limit~PA recognizable $\omega$-languages.
	Deterministic Büchi~PA $\not\subseteq$ deterministic limit~PA.
\end{lemma}
\begin{proof}
	The proof is almost identical to the proof of \Cref{lem:reach-no-limit}, but this time we consider the $\omega$-language $L_{a=b} = \{\alpha \mid |\alpha[1:i]|_a = |\alpha[1:i]|_b \text{ for $\infty$ many }i \}$. Then we can re-use the same argument as the constructed infinite word has indeed infinitely many balanced $a$-$b$ prefixes. 
\end{proof}

\begin{lemma}
	\label{lem:bukkiweak}
	%  The class of deterministic Büchi~PA recognizable $\omega$-languages is no subclass of the class of deterministic weak reset~PA recognizable $\omega$-languages.
	Deterministic Büchi~PA $\not\subseteq$ deterministic weak reset~PA.
\end{lemma}
\begin{proof}
	As shown in \Cref{lem:non-closure-weak-reset}, the $\omega$-language $L_{a=b} \cup L_{a=c}$ is not deterministic weak reset~PA recognizable. However, it is obviously deterministic Büchi~PA recognizable.
\end{proof}

We note however that the class of $\omega$-languages recognized by deterministic Büchi~PA with a \emph{linear} set form a subclass of the class of $\omega$-languages recognized by deterministic weak reset~PA with a linear set, as clarified in the following lemma.

\begin{lemma}
	Let $\Amc$ be a deterministic Büchi~PA with a \emph{linear set}~$C(\bbf, P)$. Then there is an equivalent deterministic weak reset~PA.
\end{lemma}
\begin{proof}
	First we observe that if $\bbf = \0$, then we have $B_\omega(\Amc) = WR_\omega(\Amc)$. To see this, let $\alpha \in B_\omega(\Amc)$ with (unique) accepting run~$r$. By Dickson's Lemma~\cite{dickson}, the run $r$ contains an infinite monotone sequence $s_1 < s_2 < \dots$ of accepting hits, that is, for all $i \geq 0$ we have $\rho(r[1:s_i]) \in C(\bbf , P)$ and for all $j > i$ we have $\rho(r[s_i+1: s_j]) \in C(\bbf, P)$. Hence, the run~$r$ also satisfies the weak reset condition.
	For the other direction let $\alpha \in WR_\omega(\Amc)$ with (unique) accepting run $r$ and reset positions $0 = k_0 < k_1 < k_2 \dots$. As we assume $\bbf = \0$, it is immediate that $\rho(r[1:k_i]) \in C(\bbf, P)$ for all $i \geq 1$. Hence, the run~$r$ also satisfies the Büchi condition.
	
	Finally we argue that we can always assume that $\bbf = \0$. Indeed, we can always encode~$\bbf$ into the state space of $\Amc$.
\end{proof}

\subsection{Decision Problems and Model Checking}
\label{sec:decision}

%\subsection{Decision Problems}
In this section, we study classical decision problems as well as the core problems for model checking for the deterministic variants. We repeat the problems for convenience. For an overview of the results in this section we refer to \Cref{tab:decision} and \Cref{tab:mc}.
\begin{itemize}
	\item Emptiness: given a~PA $\Amc$, is the $\omega$-language of $\Amc$ empty?
	\item Membership: given a~PA $\Amc$ and finite words $u, v$, does $\Amc$ accept $uv^\omega$?
	\item Universality: given a~PA $\Amc$, does $\Amc$ accept every infinite word?
	\item Inclusion: given a safety automaton or a~PA $\Amc_1$, and a~PA $\Amc_2$, is the $\omega$-language of~$\Amc_1$ a subset of the $\omega$-language of $\Amc_2$?
	\item Intersection-emptiness: given a safety automaton or a~PA $\Amc_1$, and a~PA $\Amc_2$, is the $\omega$-language of $\Amc_1$ disjoint from the $\omega$-language of $\Amc_2$?
\end{itemize}

The techniques employed for showing \NP-completeness for non-emptiness and membership are identical to the non-deterministic case, see \Cref{lem:nonemptiness}, \Cref{lem:nonemptiness-limit} and \Cref{cor:membership}.
\begin{cor}
Emptiness for deterministic limit~PA, deterministic reachability-regular~PA, deterministic weak reset~PA and deterministic strong reset~PA is $\coNP$-complete.   
\end{cor}
\begin{cor}
	Membership for deterministic limit~PA, deterministic reachability-regular~PA, deterministic weak reset~PA and deterministic strong reset~PA is $\NP$-complete.   
\end{cor}

However, showing undecidability and completeness results for universality and the model checking problems require more sophisticated methods. Hence, we devote most of this section to the latter problems,  being the core of solving universal model checking.
Recall that we can always reduce universality to inclusion, as an automaton $\Amc$ is universal if and only if $\Sigma^\omega$ is a subset of the $\omega$-language of $\Amc$. Hence, we show all undecidability results and lower bounds for universality and all decidability results and upper bounds for inclusion. We begin with the undecidability results.

\begin{lemma}
	Universality for deterministic reachability-regular~PA and deterministic weak reset~PA is undecidable.   
\end{lemma}
\begin{proof}
	As shown in~\cite{infiniteZimmermann}, universality is already undecidable for deterministic reachability~PA. As we can effectively construct equivalent deterministic reachability-regular~PA and deterministic weak reset~PA from deterministic reachability~PA by \Cref{lem:reachToReachReg} and \Cref{lem:reachRegToWeak}, the lemma follows.
\end{proof}

Recall our strategy for showing that universality for finite word~PA is $\Pi_2^\P$-hard by a reduction from the irrelevance problem for~PA, see \Cref{cor:uniHardnessFinite}.
We conclude $\Pi_2^\P$-hardness for universality for deterministic limit and deterministic strong reset~PA using a simplified variant of this reduction.

\begin{cor}
	\label{cor:universalityLimit}
	Universality and inclusion for deterministic limit~PA and deterministic strong reset~PA is $\Pi_2^\P$-hard.
\end{cor}

Now we focus on the decidability results and upper bounds. Let $\Amc_1$ and $\Amc_2$ be two (deterministic limit)~PA. Note that $L_\omega(\Amc_1) \subseteq L_\omega(\Amc_2)$ holds if and only if \mbox{$L_\omega(\Amc_1) \cap \overline{L_\omega(\Amc_2)} =\varnothing$}. As deterministic limit~PA are effectively closed under complement and intersection, and emptiness is $\coNP$-complete (and hence decidable) for them, we obtain the following result.
\begin{cor}
	Universality and inclusion for deterministic limit~PA are decidable.    
\end{cor}
Unfortunately, we do not obtain tight bounds and conjecture that these problem are $\Pi_2^\P$-complete for them. 
However, we make the following observation.
The relatively high~$\Pi_2^\P$ lower bound of universality for deterministic limit~PA and deterministic strong reset~PA (and also of irrelevance for finite word~PA)
can be explained by the cost of (implicitly) complementing semi-linear sets. In fact, if we have the guarantee that the semi-linear set of the second~PA can be complemented in polynomial time, the universality and inclusion problems become $\coNP$-complete.
We start with deterministic limit~PA.

\begin{lemma}
	\label{lem:limitcoNP}
	Let $\Amc_1$ and $\Amc_2$ be deterministic limit~PA with the guarantee that the semi-linear set of $\Amc_2$ can be complemented in polynomial time. Then the following questions are $\coNP$-complete.
	\begin{enumerate}
		\item Is $L_\omega(\Amc_2) = \Sigma^\omega$?
		\item Is $L_\omega(\Amc_1) \subseteq L_\omega(\Amc_2)$?
	\end{enumerate}
\end{lemma}
\begin{proof}
	Containment in $\coNP$ of both questions follows immediately from a reduction to emptiness for deterministic limit~PA, as the guarantee allows us to complement deterministic limit~PA in polynomial time, and we can construct the product automaton of two deterministic limit~PA in polynomial time.
	
	The partition problem is defined as follows: given a multiset $M$ of positive integers, is there a subset $M'$ of $M$ such that $\sum_{n \in M} n = \sum_{n \in M \setminus M'} n$?. This problem is one of the classical \NP-complete problems~\cite{intract}. We reduce from its complement.
	
	Let $M = \{n_1, \dots n_k\}$ be a multiset of positive integers. We construct a deterministic limit~PA~$\Amc$ over the alphabet $\{a,b\}$ of dimension 2 as follows. The state set of $\Amc$ is $\{q_0, q_1, \dots q_k\}$ where~$q_0$ is the initial state and $q_k$ is the only accepting state. For every $1 \leq i \leq k$, there is an $a$-transition from $q_{i-1}$ to ${q_i}$ labeled with $(n_i, 0)$ as well as a $b$-transition labeled with~$(0, n_i)$. Finally, we add an $a$-loop and a $b$-loop to the accepting state $q_k$, both labeled with $\0$. The semi-linear set of $\Amc$ is 
	\[C = \{(z,z') \mid z \neq z' \} = C((1,0), \{(1,0), (1,1)\}) \cup C((0,1), \{(0,1), (1,1)\}),\]
	whose size does not depend on the size of $M$. Furthermore, the complement of $C$ is 
	\[\overline{C} = \{(z,z) \mid z \in \Nbb\} = C(\0, \{(1,1\}),\]
	and can hence be computed in polynomial time.
	
	Now we have that $\Amc$ is universal if and only if $M$ is a negative instance of partition. To see this, observe that every prefix $w_1 \dots w_k \in \{a,b\}^k$ of every word $\alpha \in L_\omega(\Amc)$ represents a subset $M'$ of $M$ with $n_i \in M'$ if and only if $w_i = a$, and hence $n_i \in M \setminus M'$ if and only if $w_i = b$. The semi-linear set $C$ of $\Amc$ states that $M'$ and $M \setminus M'$ are not a partition of $M$. Hence, the claim follows.
\end{proof}

Now we show that this is also the case for deterministic strong reset~PA, presenting in parallel the strategy that we will finally adapt to show that inclusion is $\Pi_2^\P$-complete in the general case.
We show that complements of deterministic strong reset~PA recognizable $\omega$-languages are reachability-regular~PA recognizable, and given that we can complement the semi-linear set in polynomial time, we can construct such a reachability-regular~PA in polynomial time. Subsequently, we show how to test intersection emptiness of a deterministic strong reset~PA and a reachability-regular~PA in $\coNP$. We begin by proving the first main ingredient.

\begin{lemma}
	\label{lem:complSR}
	Let $\Amc = (Q, \Sigma, q_0, \Delta, F, C)$ be a deterministic strong reset~PA. Then there is a (non-deterministic) reachability-regular~PA $\Amc'$ recognizing $\overline{SR_\omega(\Amc)}$. Furthermore, if $C$ can be complemented in polynomial time, then we can compute $\Amc'$ in polynomial time.
\end{lemma}
\begin{proof}
	We assume that every state of $\Amc$ is reachable from the initial state $q_0$ (as we can safely remove unreachable states in polynomial time).
	Observe that $\Amc$ rejects an infinite word $\alpha$ whenever one of the following two conditions is met:
	\begin{enumerate}
		\item The unique run of $\Amc$ on $\alpha$ visits every accepting state just finitely often.
		\item The unique run of $\Amc$ on $\alpha$ visits an accepting state with bad counter values.
	\end{enumerate}
	
	The first condition is $\omega$-regular, while the second one can be tested with a reachability~PA. 
	We begin by clarifying the first point. Testing whether there is an infinite word rejected by~$\Amc$ because its unique run visits every accepting only finitely often boils down to testing whether the complement $\omega$-language of the underlying Büchi automaton of $\Amc$ is non-empty.
	Hence, let $\Bmc$ be the Büchi automaton obtained from~$\Amc$ by forgetting all vectors and let $\overline{\Bmc}$ be a Büchi automaton recognizing the complement of~$L_\omega(\Bmc)$. As $\Bmc$ is deterministic, we can construct $\overline{\Bmc}$ in polynomial time~\cite{detBuchiCompl}. Then $\overline{\Bmc}$ accepts all words rejected by $\Amc$ due to the first condition.

	Second, we show how to construct a (non-deterministic) reachability~PA accepting all infinite words that are rejected by $\Amc$ due the second condition. Recall the following definition from \Cref{lem:nonemptiness}. For any two states $p, q \in Q$ let $\Amc_{p \Rightarrow q} = (Q \cup \{q_0'\}, \Sigma, q_0', \Delta_{p \Rightarrow q}, \{q\}, C)$ where 
	\[\Delta_{p \Rightarrow q} = \{(q_1, a, \vbf, q_2) \mid (q_1, a, \vbf, q_2) \in \Delta, q_1 \notin F) \cup \{(q_0', a, \vbf, q_2) \mid (p, a, \vbf, q_2) \in \Delta\}.\]

	For every pair of accepting states $p,q \in F$, we consider the~PA $\overline{\Amc}_{p \Rightarrow q}$, which is defined as $\Amc_{p \Rightarrow q}$ but with the complement semi-linear set $\overline{C}$ of $C$. Observe that we can compute~$\overline{\Amc}_{p \Rightarrow q}$ in polynomial time if $C$ can be complemented in polynomial time. Hence, $\overline{\Amc}_{p \Rightarrow q}$ accepts all finite words collecting a vector not in $C$ when starting in $p$, ending in $q$, and not visiting other accepting states in-between.
	Now let $\Amc^\circ_{p \Rightarrow q}$ the~PA obtained from $\Amc$ and $\overline{\Amc}_{p \Rightarrow q}$ as follows. We start with a copy of $\Amc$ but every state is not accepting and every vector is replaced by~$\0$. Then, we identify the state $p$ in $\Amc^\circ_{p \Rightarrow q}$ with the initial state of $\overline{\Amc}_{p \Rightarrow q}$, that is, we remove every outgoing transition of $p$ in $\Amc^\circ_{p \Rightarrow q}$ and replace them with the outgoing transitions of the initial state of $\overline{\Amc}_{p \Rightarrow q}$. Furthermore, the accepting state $q$ of $\overline{\Amc}_{p \Rightarrow q}$ is the only accepting state of $\Amc^\circ_{p \Rightarrow q}$. Observe that $q$ has no outgoing transitions by construction. Finally, we add a trivial self-loop to $q$ and choose $\overline{C}$ to be the semi-linear set of $\Amc^\circ_{p \Rightarrow q}$.
	Hence, $\Amc^\circ_{p \Rightarrow q}$ accepts all infinite words of the form $uv\beta$ with $\beta \in \Sigma^\omega$ such that $\Amc$ is in state $p$ after reading $u$, then in state $q$ after further reading~$v$, and collects a vector not in~$C$ while reading~$v$ and hence rejects every infinite word with prefix $uv$.
	We compute the~PA $\Amc^\circ_{p \Rightarrow q}$ for every $p, q \in F$ and let $\Amc^\circ$ be the reachability~PA recognizing the union of all these $\Amc^\circ_{p \Rightarrow q}$ by taking the disjoint union and connecting them with a fresh initial state (see \cite[Lemma 6]{infiniteZimmermann} for details). In contrast to an iterated product construction, the presented construction allows us to compute $\Amc^\circ$ in polynomial time, albeit not preserving determinism.
	
	Finally, let $\Amc'$ be a reachability-regular~PA accepting $L_\omega(\overline{\Bmc}) \cup R_\omega(\Amc^\circ)$. Note that $\Amc'$ can be computed in polynomial time in the sizes of $\overline{\Bmc}$ and $\Amc^\circ$ by turning both automata into reachability-regular~PA and again taking their disjoint union with a fresh initial state. Now we have $RR_\omega(\Amc') = \overline{SR_\omega(\Amc)}$.
\end{proof}

Before we proceed, we show two auxiliary lemmas.
\begin{lemma}
	Let $\Amc_1 = (Q_1, \Sigma, p_I, \Delta_1, F_1, C)$ be a deterministic strong reset~PA and $\Bmc = (Q_2, \Sigma, q_I, \Delta_2, F_2)$ be a Büchi automaton. Then $SR_\omega(\Amc_1) \cap L_\omega(\Bmc)$ is ultimately periodic.
\end{lemma}
\begin{proof}
	Assume $SR_\omega(\Amc_1) \cap L_\omega(\Bmc) \neq \varnothing$ and let $\alpha$ be an infinite word accepted by both automata, $\Amc_1$ and $\Bmc$.
	If $\alpha = uv^\omega$ for some $u, v \in \Sigma^*$, we are done. Hence assume that this is not the case.
	Let $\Amc = (Q_1 \times Q_2, \Sigma, (p_I, q_I), \Delta, F_1 \times Q_2, C)$ with $\Delta = \{((p, q), a, \vbf, (p', q')) \mid (p, a, \vbf, p') \in \Delta_1, (q, a, q') \in \Delta_2\}$ be the product automaton\footnote{We note that $\Amc$ interpreted as a strong reset~PA does not recognize $SR_\omega(\Amc_1) \cap L_\omega(\Bmc)$. Instead, it accepts all infinite words $\alpha \in SR_\omega(\Amc_1)$ such that $\Bmc$ has an infinite but not necessarily accepting run on~$\alpha$.} of $\Amc_1$ and $\Bmc$.
	As $\alpha \in SR_\omega(\Amc_1)$, the unique accepting run of $\Amc_1$ on $\alpha$, say $(p_0, \alpha_1, \vbf_1, p_1)$ $(p_1, \alpha_2, \vbf_2, p_2) \dots$ with $p_0 = p_I$, is accepting. Likewise, as $\alpha \in L_\omega(\Bmc)$, there is an accepting run $(q_0, \alpha_1, q_1) (q_1, \alpha_2, q_2) \dots$ with $q_0 = q_I$ of $\Bmc$ on $\alpha$.
	
	Hence, $r = ((p_0, q_0), \alpha_1, \vbf_1, (p_1, q_1)) ((p_1, q_1), \alpha_2, \vbf_2, (p_2, q_2)) \dots$ is a run of $\Amc$ on $\alpha$
	with the following properties:
	\begin{itemize}
		\item there is $p_f \in F_1$ such that for infinitely many $i$ we have $p_i = p_f$. Let $f_1, f_2, \dots$ denote the positions of all occurrences of a state of the form $(p_f, \cdot)$ in $r$.
		\item there is $q_f \in F_2$ such that for infinitely many $i$ we have $q_i = q_f$. 
		%Let $f'_1, f'_2, \dots$ denote the positions of all occurrences of a state of the form $(\cdot, q_f)$ in $r$.
	\end{itemize}
	Let $j \geq f_1$ be minimal such that $q_j = q_f$.
	Now let $k \leq j$ be maximal such that $k = f_\ell$ for some $\ell \geq 1$.
	Then $r[1:f_\ell]r[f_\ell+1 : f_{\ell+1}]^\omega$ is an accepting run of $\Amc$ on an ultimately periodic word, say $uv^\omega$, that visits an accepting state of $\Bmc$ infinitely often. Hence $uv^\omega \in SR_\omega(\Amc_1) \cap L_\omega(\Bmc)$.
\end{proof}
\begin{lemma}
	\label{lem:interSRBukki}
	Let $\Amc_1 = (Q_1, \Sigma, p_0, \Delta_1, F_1, C)$ be a deterministic strong reset~PA and let $\Bmc = (Q_2, \Sigma, q_0, \Delta_2, F_2)$ be a Büchi automaton. The question $SR_\omega(\Amc_1) \cap L_\omega(\Bmc) = \varnothing$ is \coNP-complete.
\end{lemma}
\begin{proof}
	The lower bound follows immediately from the $\coNP$-hardness of emptiness; hence, we focus on the containment in $\coNP$ by showing that testing intersection non-emptiness is in \NP.
	By the previous lemma it is sufficient to check the existence of an ultimately periodic word.
	%Let $\Amc_1 = (Q_1, \Sigma, p_0, \Delta_1, F_1, C)$ and $\Bmc = (Q_2, \Sigma, q_0, \Delta_2, F_2)$.
	Recall the algorithm in \ref{lem:nonemptiness} that decides the non-emptiness problem for reset~PA in~$\NP$ by exploiting that $\omega$-languages accepted by a strong reset~PA are ultimately periodic.
	We modify the algorithm as follows.
	
	First, let $\Amc$ be the product automaton of $\Amc_1$ and $\Bmc$ (as in the previous proof; however, the set of accepting states is not important for the algorithm).
	We guess state $p_f \in F_1$ and $q_f \in F_2$ that we expect to be seen infinitely often to satisfy the acceptance conditions of~$\Amc_1$ resp.\ $\Bmc$. Furthermore, similar to the algorithm above we guess a sequence of distinct states $(~p_1, q_1) (p_2, q_2) \dots (p_n, q_n)$ of $\Amc$ such that for some $\ell \leq n$ we have $q_\ell = q_f$, for some $k \leq \ell$ we have $p_k = p_f$, and for all $j \neq \ell$ we have $p_j \in F_1$.
	If $p_\ell \in F_1$, we proceed exactly as in the original $\NP$-algorithm, that is, for every $0 \leq i < n$, we test the finite word~PA $\Amc_{(p_i, q_i) \Rightarrow (p_{i+1}, q_{i+1})}$, as well as the finite word~PA $\Amc_{(p_n, q_n) \Rightarrow (p_{k}, q_{k})}$ for non-emptiness. 
	
	Now assume $p_\ell \notin F_1$.
	Let $\Amc_{p_{\ell-1} \Rightarrow q_\ell \Rightarrow p_{\ell + 1}}$ be the finite word~PA that accepts all finite infixes accepted by the product automaton $\Amc$ when starting in $(p_{\ell-1}, q_{\ell-1})$, visiting $(p_\ell, q_\ell)$ at some point, and ending in $(p_{\ell+1}, q_{\ell+1})$ such that for all internal states $(p_i, q_i)$ we have $p_i \notin F_1$.
	Two achieve this, we take two copies of $\Amc$, where all accepting states in the first copy are not reachable, all accepting states in the second copy have no outgoing transitions, and the second copy can only be reached from the first copy via $(p_\ell, q_\ell)$.
	
	Formally, let 
	$$\Amc_{p_{\ell-1} \Rightarrow q_\ell \Rightarrow p_{\ell + 1}} = (Q_1 \times Q_2 \times \{1,2\} \cup \{q_0'\}, \Sigma, q_0', \Delta', \{(p_{\ell+1}, q_{\ell+1})\}, C)$$ 
	with 
	\begin{align*}
		\Delta' =&\ \{((p,q,1), a, \vbf, (p', q', 1)) \mid ((p,q), a, \vbf, (p', q')) \in \Delta, p, p' \notin F_1\} \\
		\cup&\ \{((p,q,2), a, \vbf, (p', q', 2)) \mid ((p,q), a, \vbf, (p', q')) \in \Delta, p \notin F_1\} \\
		\cup&\ \{((p,q,1), a, \vbf, (p_\ell, q_\ell, 2) \mid ((p, q), a, \vbf, (p_\ell, q_\ell)) \in \Delta \} \\
		\cup&\ \{(q'_0, a, \vbf, (p', q', 1) \mid ((p_{\ell-1}, q_{\ell-1}), a, \vbf, (p', q' )) \in \Delta, p' \notin F_1 \} \\
		\cup&\ \{(q'_0, a, \vbf, (p_\ell, q_\ell, 2) \mid ((p_{\ell-1}, q_{\ell-1}), a, \vbf, (p_\ell, q_\ell)) \in \Delta \}.
	\end{align*}
	
	Now the algorithm is similar the first case with the addition that we also test this automaton for non-emptiness. Hence, for every $0 \leq j < \ell-1$ and $\ell+1 \leq j < n$ we test $\Amc_{(p_j, q_j) \Rightarrow (p_{j+1}, q_{j+1})}$ as well as $\Amc_{(p_n, q_n) \Rightarrow (p_k, q_k)}$ for non-emptiness. Furthermore, we test $\Amc_{p_{\ell-1} \Rightarrow q_\ell \Rightarrow p_{\ell + 1}}$ for non-emptiness. 
	
	If all these tests succeed, we conclude $SR_\omega(\Amc_1) \cap L_\omega(\Bmc) \neq \varnothing$, as they witness the existence of an ultimately periodic word $\alpha$ accepted by $\Amc_1$ with the property that there is a run an accepting run of $\Bmc$ on $\alpha$ that visits $q_f$ infinitely often.
\end{proof}

Let $\Amc$ be a deterministic strong reset~PA whose semi-linear set can by complemented in polynomial time. By \Cref{lem:complSR} we can compute a Büchi automaton $\overline{\Bmc}$ and a reachability~PA~$\Amc^\circ$ such that $\overline{SR_\omega}(\Amc) = L_\omega(\overline{\Bmc}) \cup R_\omega(\Amc^\circ)$.
By the previous lemma we can test $SR_\omega(\Amc) \cap L_\omega(\overline{\Bmc}) = \varnothing$ in \coNP.
Hence, it remains to show that testing intersection emptiness of a deterministic strong reset~PA and a reachability~PA is decidable in $\coNP$. To achieve that, we use the \NP-algorithm in~\cite{zvassnz} deciding the reachability problem for $\Zbb$-$\mathsf{VASS}$ with~$k$ nested zero-tests ($\Zbb$-$\mathsf{VASS}^\mathsf{nz}_k$).
A $\Zbb$-$\mathsf{VASS}^\mathsf{nz}_k$ (of dimension $d \geq 1$) is a tuple $V = (Q, Z, E)$ where $Q$ is a finite set of states, 
$Z \subseteq \{0, 1, \dots, d\}$ is its set of zero tests with $|Z \setminus \{0\}| = k$, and $E \subseteq Q \times \Zbb^d \times Z \times Q$ is a finite set of transitions. 
A configuration of $V$ is a pair $(p, \vbf) \in Q \times \Zbb^d$. Assume $\vbf = (v_1, \dots, v_d)$. We write $(p, \vbf) \vdash_V (p', \vbf')$ if there is a transition $(p, \ubf, \ell, p') \in E$ such that $\vbf' = \vbf + \ubf$ and $v_1 = \dots = v_\ell = 0$. Furthermore, we write $(p, \vbf) \vdash^*_V (p', \vbf')$ if there is a sequence $(p_1, \vbf_1) \vdash_V \dots \vdash_V (p_n, \vbf_n)$ for some $n \geq 1$ such that $(p, \vbf) = (p_1, \vbf_1)$ and $(p', \vbf') = (p_n, \vbf_n)$. 
Intuitively, a $\Zbb$-$\mathsf{VASS}^\mathsf{nz}_k$ is a counter machine with zero-tests that can only zero-test the top-most $\ell$ counters at once, for $k$ different values of $\ell$.

The \emph{reachability} problem for $\Zbb$-$\mathsf{VASS}^\mathsf{nz}_k$ is defined as follows: given a $\Zbb$-$\mathsf{VASS}^\mathsf{nz}_k$ $V$ and two configurations $(p, \vbf), (p', \vbf')$, does $(p, \vbf) \vdash_V^* (p', \vbf')$ hold? As shown in~\cite{zvassnz}, this problem is $\NP$-complete\footnote{We note that our definition of $\Zbb$-$\mathsf{VASS}^\mathsf{nz}_k$ differs slightly from the definition in~\cite{zvassnz}, as we allow $E\subseteq Q\times \Zbb^d\times Z\times Q$ instead of $E\subseteq Q\times \{-1,0,1\}^d\times Z \times Q$ only. However, this difference does not change the mentioned complexity for the reachability problem~\cite{zvassnz}, see also~\cite[Section A.1]{logspaceconversion}.} for any fixed~$k$.
We are now ready to show that universality and inclusion for deterministic strong reset~PA are $\coNP$-complete if we can complement their semi-linear sets in polynomial time.

\begin{lemma}
	Let $\Amc_1$ and $\Amc_2$ be deterministic strong reset~PA with the guarantee that the semi-linear set of $\Amc_2$ can be complemented in polynomial time. Then the following questions are $\coNP$-complete.
	\begin{enumerate}
		\item Is $SR_\omega(\Amc_2) = \Sigma^\omega$?
		\item Is $SR_\omega(\Amc_1) \subseteq SR_\omega(\Amc_2)$?
	\end{enumerate}
\end{lemma}
\begin{proof}
	Hardness follows again from a reduction from partition, very similar to the reduction in \Cref{lem:limitcoNP}. Hence, we focus on the containment in $\coNP$ of the second question.
	
	As mentioned above, we can compute a Büchi automaton $\overline{\Bmc}$ and a reachability~PA $\Amc^\circ$ such that $\overline{SR_\omega}(\Amc_2) = L_\omega(\overline{\Bmc}) \cup R_\omega(\Amc^\circ)$ by \Cref{lem:complSR}, and test $SR_\omega(\Amc_1) \cap L_\omega(\overline{\Bmc}) = \varnothing$ in $\coNP$ by \Cref{lem:interSRBukki}.
	
	It remains to show how to solve $SR_\omega(\Amc_1) \cap R_\omega(\Amc^\circ) = \varnothing$ in $\coNP$.
	In order to do so we use the \NP-algorithm in~\cite{zvassnz} solving reachability for $\Zbb$-$\mathsf{VASS}^\mathsf{nz}_2$ to decide $SR_\omega(\Amc) \cap R_\omega(\Amc^\circ) \neq \varnothing$ in $\NP$. 
	Let $\Amc_1 = (Q_1,\Sigma,p_0, \Delta_1, F_1, C_1)$ be of dimension $d_1$ and $\Amc^\circ = (Q_2,\Sigma, q_0, \Delta_2, F_2, C_2)$ be of dimension $d_2$.
	Assume $SR_\omega(\Amc) \cap R_\omega(\Amc^\circ) \neq \varnothing$ and let $\alpha$ be an infinite word accepted by both automata. In particular, there is a finite prefix $u$ of $\alpha$ and an accepting run of $\Amc^\circ$ on $\alpha$ satisfying the Parikh condition after processing $u$, say in the accepting state $q_f \in F_2$. Recall that all outgoing transitions of every accepting state of $\Amc^\circ$ are self-loops, hence we may assume that $\Amc^\circ$ does not leave $q_f$ anymore after processing $u$.
	Now let $p_f \in F_1$ be the first accepting state visited by the unique accepting run of $\Amc_1$ on $\alpha$ after processing $u$, say after processing the prefix $uv$. Note that $\Amc^\circ$ is still in $q_f$ after processing $uv$.
	
	Our strategy is as follows.
	First, we guess states $p_f$ and $q_f$ with the mentioned properties. We then verify these properties by building a product $\Zbb$-$\mathsf{VASS}^\mathsf{nz}_2$ with the property that $((p_0, q_0), \0) \vdash^*_V ((p_f, q_f), \0)$ if and only there is a finite prefix $uv$ as described above. Then we need to test whether $\Amc_1$ can continue a partial run from $p_f$ to an accepting run, that is, whether $\Amc_2$ with initial state $p_f$ accepts at least one infinite word, say $\beta$. If all these tests succeed, the infinite word $uv\beta$ witnesses non-emptiness.

	We build a product $\Zbb$-$\mathsf{VASS}^\mathsf{nz}_2$ $V$ with $d_1 + d_2$ many counters. The idea is as follows. We use the first $d_1$ counters with to simulate $\Amc_1$ ensuring that every visit of an accepting state of~$\Amc_1$ is with good counter values by zero-testing them.
	Likewise, we use a second set of $d_2$ counters to simulate $\Amc^\circ$.
	Let us give some more details on how to verify that every visit of an accepting state implies good counter values. Let $C_1 = C(\bbf_1, P_1) \cup \dots \cup C(\bbf_k, P_k)$ for some $k \geq 1$. 
	For every accepting state $f \in F_1$ we insert $k$ states, say $f^{(1)}, \dots, f^{(k)}$, and a copy of~$f$ itself. We connect~$f$ to $f^{(i)}$ with a transition subtracting $\bbf_i$ and no zero-test.
	Then, for every period vector $\pbf_i \in P_i$, we insert a loop on~$f^{(i)}$ subtracting~$\pbf_i$. 
	Finally, every outgoing transition of $f^{(i)}$ is equipped with a zero-test on the first $d_1$ counters. 
	This construction allows us to test membership of the current counter values in $C_1$, while resetting the counters in parallel.
	Finally, when reaching $(p_f, q_f)$, we use the same idea to check whether the vector induced by the last $d_2$ counters yields a vector in $C_2$. As $p_f$ is accepting, we expect all counters to be zero at this point, which we check by zero-testing them all.
	Hence, if $((p_0, q_0), \0) \vdash^*_V ((p_f, q_f), \0)$ holds (where we assume that $(p_f, q_f)$ is the state of~$V$ reached after the described zero-tests), this implies that there is a finite prefix $uv$ such that~$\Amc_2$ rejects every infinite word with prefix $uv$ due to bad counter values, while the unique partial run of~$\Amc_1$ on $uv$ respects the strong reset acceptance condition.
	Finally, we need to check whether the $\omega$-language of $\Amc_1$ with initial state $p_f$ is not-empty using the $\NP$-algorithm for testing emptiness for strong reset~PA (see \Cref{lem:nonemptiness})
	If these tests succeed, there is an infinite word $\beta$ such that $uv\beta \in SR_\omega(\Amc_1) \setminus SR_\omega(\Amc_2)$, witnessing non-inclusion.
\end{proof}

We are now ready to combine and generalize the results in the previous lemmas to arbitrary deterministic strong reset~PA.

\begin{lemma}
	Universality and inclusion for deterministic strong reset~PA are $\Pi_2^\P$-complete.
\end{lemma}
\begin{proof}
	Hardness follows from~\Cref{cor:universalityLimit}; hence, we show that non-inclusion is in $\Sigma_2^\P$, yielding the desired result.
	
	Let $\Amc_1 = (Q_1, \Sigma, p_0, \Delta_1, F_1, C_1)$ and $\Amc_2 = (Q_2, \Sigma, q_0, \Delta_2, F_2, C_2)$ be deterministic strong reset~PA of dimensions $d_1$ and $d_2$, resp. 
	As in the previous lemma, we can compute a Büchi automaton~$\overline{\Bmc}$ accepting all infinite words that are rejected by~$\Amc_2$ because every accepting state of $\Amc_2$ is visited only finitely often. 
	By \Cref{lem:interSRBukki}, we can test $SR_\omega(\Amc_1) \cap L_\omega(\overline{\Bmc}) \neq \varnothing$ in~$\NP$. 
	If the intersection is indeed non-empty, we conclude $SR_\omega(\Amc_1) \not\subseteq SR_\omega(\Amc_2)$.
	Otherwise, this non-inclusion might hold as there is an infinite word $\alpha \in SR_\omega(\Amc_1)$ rejected by~$\Amc_2$ due to a reset with bad counter values. To test this case, we proceed as follows.
	
	We guess two accepting states $q_2$, $q_3$ of $\Amc_2$ and
	utilize the non-irrelevance algorithm in \Cref{lem:irrelevance} to test the existence of such a finite infix with bad counter values, \ie, an infix whose unique partial run of $\Amc_2$ yields a vector $\vbf \notin C_2$. We then construct a $\Zbb$-$\mathsf{VASS}^\mathsf{nz}_2$ to test whether there is an infinite word accepted by $\Amc_1$ that is rejected by $\Amc_2$ because the partial run between $q_2$ and $q_3$ yields the bad vector $\vbf$, ensuring that $\Amc_1$ respects all resets, similar to the proof of the previous lemma.
	
	Let $\Amc = (Q_1 \times, Q_2, \Sigma, (p_0, q_0), \Delta, F_1 \times Q_2, C_1)$ with \[\Delta = \{((p,q), a, \vbf, (p',q')) \mid (p,a,\vbf,p') \in \Delta_1, (q,a,\cdot,q') \in \Delta_2\}\] be the product automaton of $\Amc_1$ and $\Amc_2$ where we only keep the vectors and accepting states from $\Amc_1$. We guess two states $(p_2, q_2), (p_3, q_3) \in Q_1 \times F_2$ such that there is $\alpha \in SR_\omega(\Amc_1) \setminus SR_\omega(\Amc_2)$ with the property that $\Amc_2$ rejects $\alpha$ because the unique rejecting run of $\Amc_2$ on $\alpha$ visits $q_2$ at some position, say $f_2$ (and hence resets), the next reset is in $q_3$ at position $f_3$, and the vector collected in this time is not a member of $C_2$. Furthermore, $p_2$ and $p_3$ are the corresponding states of $\Amc_1$ at positions $f_2$ and $f_3$.
	Additionally, we guess two states $(p_1, q_1), (p_4, q_4) \in F_1 \times Q_2$ such that the unique accepting run of $\Amc_1$ on $\alpha$ resets the last time before reaching position $f_2$ say at position $f_1 \leq f_2$, namely in $p_1$; and $\Amc_1$ resets the first time after reaching position $f_3$ say at position $f_4 \geq f_3$, namely in $p_4$. Furthermore,~$q_1$ and $q_4$ are the states of $\Amc_2$ at positions $f_1$ and $f_4$. Observe that $f_1 = f_2$ and $f_3 = f_4$ are possible.
	
	First, we test whether $(p_1, q_1)$ is reachable in $\Amc$ in the sense that there is a finite prefix~$u$ of~$\alpha$ such that $\Amc$ is in state $(p_1, q_1)$ after reading $v$ and $\Amc_1$ resets with good counter values whenever an accepting state of $\Amc_1$ is seen. In order to do so, we modify $\Amc$ such that $(p_1, q_1)$ is the only accepting state, and replace its outgoing transitions with trivial self-loops. Then we use the $\NP$-algorithm in~\Cref{lem:nonemptiness} to test non-emptiness of the resulting automaton.
	
	Second, we test whether it is possible to successfully continue the run from $(p_4, q_4)$ in the sense that there is an infinite word $\beta$ accepted by $\Amc$ when starting in $(p_4, q_4)$ (and hence by~$\Amc_1$ when starting in $p_4)$. In order to do so, we modify $\Amc$ such that $(p_4, q_4)$ is the initial state and test non-emptiness of the resulting automaton, again using the \NP-algorithm in~\Cref{lem:nonemptiness}.
	
	Let $\Amc'$ be defined as $\Amc$ but this time we only keep the vectors from $\Amc_2$ (instead of~$\Amc_1$) and its semi-linear set is $C_2$ (instead of $C_1$).
	Third, we test whether there is indeed a finite prefix of $w$ of $\alpha$ such that $\Amc'$ is in state $(p_3, q_3)$ after processing $w$ when starting in $(p_2, q_2)$, and the vector collected by the unique partial run of $\Amc'$ (and hence $\Amc_2$) is not a member of~$C_2$.
	To achieve that we compute the~PA $\Amc'_{(p_2, q_2) \Rightarrow (p_3, q_3)}$ and test for non-irrelevance using the algorithm in \Cref{lem:irrelevance} in $\Sigma_2^\P$ (observe that the construction of this~PA preserves determinism up to completeness, and we can always complete the~PA by adding a non-accepting sink). Recall that a part of this algorithm guesses a vector $\vbf$ not contained in~$C_2$ such that there is a run ending in the accepting state (here $(p_3, q_3)$) collecting $\vbf$. We need to keep this vector for the next step.
	
	Finally, we need to verify that there is a non-rejecting partial run of $\Amc$ on an infix~$vwx$, starting in $(p_1, q_1)$, visiting $(p_2, q_2)$ and $(p_3, q_3)$ in between, and ending in $(p_4, q_4)$ such that $w$ is processed between the visits of $(p_2, q_2)$ and $(p_3, q_3)$, witnessing that $\Amc_2$ is indeed rejecting. 
	To achieve that, we construct a product $\Zbb$-$\mathsf{VASS}^\mathsf{nz}_2$ $V$ of dimension $d_1 + d_2$ in a similar way as in the proof of the previous lemma. We give a high level description of $V$. The state set of~$V$ consists of three copies of the product $Q_1 \times Q_2$, and the transitions keep the vectors of the transitions of both automata, $\Amc_1$ and $\Amc_2$. 
	We can move from the first copy to the second copy upon reaching $(p_2, q_2)$, and from the second copy to the third copy  upon reaching~$(p_3, q_3)$. 
	The $d_2$ counters belonging to $\Amc_2$ are frozen (that is $\0$) in the first and third copy. Contrary, all counters are used in the second copy. Furthermore, we remove every state in $Q_1 \times F_2$ in the second copy besides $(p_2, q_2)$ and $(p_3, q_3)$.
	We verify that $\Amc_1$ always resets with good counter values using zero tests as in the proof of the previous lemma.
	Finally, upon reaching~$(p_4, v_4)$, we subtract $\0^{d_1} \cdot \vbf$ (where $\vbf$ is the vector guessed in the previous step) and test whether all counters are zero. As $p_4$ is accepting, we expect the first~$d_1$~counters to be zero. 
	Furthermore, as $\vbf$ is a vector witnessing bad counter values with respect to $\Amc_2$, all-zero counters imply that there is indeed such an infix $w$ of $\alpha$ breaking the run of $\Amc_2$. Let $(p_f, q_f)$ be the state of $V$ reached after this zero test.
	Then we use the $\NP$-algorithm in~\cite{zvassnz} to test $((p_1, q_1), \0) \vdash^*_V ((p_f, q_f), \0)$. If the answer is positive, we conclude that $SR_\omega(\Amc_1) \not\subseteq SR_\omega(\Amc_2)$, as witnessed by $\alpha = uvwxy\beta$.
	
	We conclude with a technical remark: the bit size of the vector~$\vbf$ guessed in the third step might polynomially depend on $C_1$ and $C_2$. In particular, the bit size of (a suitable encoding of) $C_1$ might be arbitrary larger than the bit size of $C_2$. Hence, when calling the irrelevance algorithm with $\Amc'$ (where only $C_2$ is present) we need to take in account that the bit size of~$\vbf$ might not be polynomial in $C_2$ but in $C_1$ and $C_2$.
\end{proof}

Finally, we study the intersection-emptiness problems, being the core of solving existential model checking. We observe that all results translate from the non-deterministic setting.
Hence, for deterministic limit~PA, deterministic reachability-regular~PA and deterministic Büchi~PA we conclude $\coNP$-completeness by \Cref{lem:interLimitReachreg} and \Cref{lem:interBuchi}.
\begin{cor}
Intersection-emptiness for deterministic limit~PA, deterministic reachability-regular~PA, and deterministic Büchi~PA is $\coNP$-complete.
\end{cor}

For deterministic strong reset~PA (and hence for deterministic weak reset~PA) we obtain undecidability, as the~PA constructed in the proof of \Cref{lem:interReset} is indeed deterministic.
\begin{cor}
Intersection-emptiness for deterministic strong and weak reset~PA is undecidable.
\end{cor}

%% file: 4.3_auto_concl.tex
\section{Conclusion and Future Work}
We have studied the expressiveness, closure properties and classical decision problems of the non-deterministic and determinstic variants of the newly introduced models of PA operating on infinite words, namely reachability-regular PA, limit~PA, weak reset~PA, and strong reset~PA. Notably, deterministic limit PA are closed under the Boolean operations and hence all common decision problems are decidable for them, including the classical model checking problems.
Additionally, (strong) reset PA, being a very expressive model, enjoy a decidable emptiness problems as well as a decidable existential safety model checling problem. 
Closely related to model checking problems are synthesis problems. Here, the problem is to generate a model from a system specification (which is correct by construction). Gale-Stewart games play a key role in solving such synthesis problems~\cite{GaleStewart}. However, these games are undecidable when winning conditions are specified by automata whose emptiness or universality problem is undecidable. 
Our decidability results for deterministic limit PA raise the interesting and important question whether Gale-Stewart games can be solved when their winning condition is expressed by these automata. 

In future work we further plan to study the regular separability problem for these models, that is, given two $\omega$-languages $L_1, L_2$ recognized by~PA operating on infinite words, is there an $\omega$-regular language~$L$ with $L_1 \subseteq L$ and $L_2 \cap L = \varnothing$. 
Solving this problem can be used as an alternative approach to solving existential PA model checking, as (regular) separability implies intersection emptiness. It has already been studied for some related models, \eg, PA on finite words~\cite{parikhsep} and Büchi $\mathsf{VASS}$~\cite{buechiVASSsep}.
Furthermore, it remains to classify the intersections of all incomparable models and thereby to provide a fine grained ``map of the universe'' for Parikh recognizable $\omega$-languages.

%% file: lit.bib
@inproceedings{FiliotGM19,
  author    = {Emmanuel Filiot and
               Shibashis Guha and
               Nicolas Mazzocchi},
  title     = {Two-Way {P}arikh Automata},
  booktitle = {39th {IARCS} Annual Conference on Foundations of Software Technology
               and Theoretical Computer Science ({FSTTCS} 2019)},
  series    = {Leibniz International Proceedings in Informatics (LIPIcs)},
  volume    = {150},
  pages     = {40:1--40:14},
  publisher = {Schloss Dagstuhl - Leibniz-Zentrum für Informatik},
  year      = {2019}}

@article{mcnaughton1966testing,
  title={Testing and generating infinite sequences by a finite automaton},
  author={McNaughton, Robert},
  journal={Information and control},
  volume={9},
  number={5},
  pages={521--530},
  year={1966},
  publisher={Elsevier}
}

@book{thomas2002automata,
	author = {Gr\"{a}del, Erich and Thomas, Wolfgang and Wilke, Thomas},
	title = {Automata logics, and infinite games: a guide to current research},
	year = {2002},
	isbn = {3540003886},
	publisher = {Springer}
}

@inbook{thomasinfinite,
author = {Thomas, Wolfgang},
title = {Automata on Infinite Objects},
year = {1991},
isbn = {0444880747},
publisher = {MIT Press},
booktitle = {Handbook of Theoretical Computer Science (Vol. B): Formal Models and Semantics},
pages = {133–191},
numpages = {59}
}

@book{intract,
  author    = {Michael R. Garey. and David S. Johnson},
  edition   = {1st},
  isbn      = {0716710455},
  publisher = {W. H. Freeman},
  title     = {Computers and Intractability: A Guide to the Theory of NP-Completeness},
  year      = {1979}
}

@article{haase,
  author    = {Haase, Christoph},
  title     = {A Survival Guide to Presburger Arithmetic},
  year      = {2018},
  publisher = {Association for Computing Machinery},
  volume    = {5},
  number    = {3},
  journal   = {ACM SIGLOG News},
  pages     = {67--82},
  numpages  = {16}
}

@article{blindcounter,
  author  = {Fernau, Henning and Stiebe, Ralf},
  year    = {2008},
  pages   = {51--64},
  title   = {Blind Counter Automata on omega-Words.},
  volume  = {83},
  journal = {Fundamenta Informaticae}
}

@inproceedings{emptynp,
  author    = {Figueira, Diego and Libkin, Leonid},
  title     = {Path Logics for Querying Graphs: Combining Expressiveness and Efficiency},
  year      = {2015},
  isbn      = {9781479988754},
  booktitle = {Proceedings of the 2015 30th Annual ACM/IEEE Symposium on Logic in Computer Science (LICS 2015)},
  pages     = {329--340}
}

@InProceedings{klaedtkeruess,
  author    = {Klaedtke, Felix and Ruess, Harald},
  title     = {Monadic Second-Order Logics with Cardinalities},
  booktitle = {Automata, Languages and Programming},
  year      = {2003},
  publisher = {Springer},
  pages     = {681--696},
  isbn      = {978-3-540-45061-0}
}

@article{DBLP:journals/ijfcs/CadilhacFM12,
  author    = {Micha{\"{e}}l Cadilhac and
               Alain Finkel and
               Pierre McKenzie},
  title     = {Bounded {P}arikh Automata},
  journal   = {International Journal of Foundations of Computer Science},
  volume    = {23},
  number    = {8},
  pages     = {1691--1710},
  year      = {2012}
}

@article{DBLP:journals/ita/CadilhacFM12,
  author    = {Micha{\"{e}}l Cadilhac and
               Alain Finkel and
               Pierre McKenzie},
  title     = {Affine {P}arikh automata},
  journal   = {{RAIRO} Theoretical Informatics and Applications},
  volume    = {46},
  number    = {4},
  pages     = {511--545},
  year      = {2012}
}

@inproceedings{DBLP:conf/fossacs/DartoisFT19,
  author    = {Luc Dartois and
               Emmanuel Filiot and
               Jean{-}Marc Talbot},
  title     = {Two-Way {P}arikh Automata with a Visibly Pushdown Stack},
  booktitle = {Foundations of Software Science and Computation Structures -- 22nd
               International Conference ({FOSSACS} 2019)},
  series    = {Lecture Notes in Computer Science},
  volume    = {11425},
  pages     = {189--206},
  publisher = {Springer},
  year      = {2019}
}

@article{buechi,
author = {Büchi, Julius R.},
title = {Weak Second-Order Arithmetic and Finite Automata},
journal = {Mathematical Logic Quarterly},
volume = {6},
number = {1‐6},
pages = {66--92},
year = {1960}
}

@InProceedings{infiniteZimmermann,
  author =	{Guha, Shibashis and Jecker, Isma\"{e}l and Lehtinen, Karoliina and Zimmermann, Martin},
  title =	{{{P}arikh Automata over Infinite Words}},
  booktitle =	{42nd IARCS Annual Conference on Foundations of Software Technology and Theoretical Computer Science (FSTTCS 2022)},
  pages =	{40:1--40:20},
  series =	{Leibniz International Proceedings in Informatics (LIPIcs)},
  ISBN =	{978-3-95977-261-7},
  ISSN =	{1868-8969},
  year =	{2022},
  volume =	{250},
  publisher =	{Schloss Dagstuhl -- Leibniz-Zentrum für Informatik}
}

@article{infiniteOurs,
  title={{P}arikh Automata on Infinite Words},
  author={Grobler, Mario and Sabellek, Leif and Siebertz, Sebastian},
  journal={arXiv preprint arXiv:2301.08969},
  year={2023}
}

@inproceedings{cadilhac2011expressiveness,
  author       = {Micha{\"{e}}l Cadilhac and
                  Alain Finkel and
                  Pierre McKenzie},
  title        = {On the Expressiveness of {P}arikh Automata and Related Models},
  booktitle    = {Third Workshop on Non-Classical Models for Automata and Applications ({NCMA} 2011)},
  volume       = {282},
  pages        = {103--119},
  publisher    = {Austrian Computer Society},
  year         = {2011}
}

@InProceedings{erlich2022history,
	author =	{Erlich, Enzo and Guha, Shibashis and Jecker, Isma\"{e}l and Lehtinen, Karoliina and Zimmermann, Martin},
	title =	{{History-Deterministic Parikh Automata}},
	booktitle =	{34th International Conference on Concurrency Theory (CONCUR 2023)},
	pages =	{31:1--31:16},
	series =	{Leibniz International Proceedings in Informatics (LIPIcs)},
	ISBN =	{978-3-95977-299-0},
	ISSN =	{1868-8969},
	year =	{2023},
	volume =	{279},
	publisher =	{Schloss Dagstuhl -- Leibniz-Zentrum für Informatik}
}

@article{karianto2004parikh,
  title={{P}arikh automata with pushdown stack},
  author={Karianto, Wong},
  journal={Diplomarbeit, RWTH Aachen},
  year={2004}
}

@article{greibach,
  title = {Remarks on blind and partially blind one-way multicounter machines},
  journal = {Theoretical Computer Science},
  volume = {7},
  number = {3},
  pages = {311--324},
  year = {1978},
  issn = {0304-3975},
  author = {Sheila A. Greibach},
}

@article{groupautomata,
  title = {Extended finite automata over groups},
  journal = {Discrete Applied Mathematics},
  volume = {108},
  number = {3},
  pages = {287--300},
  year = {2001},
  issn = {0166-218X},
  author = {Victor Mitrana and Ralf Stiebe},
}

@InProceedings{valenceRevisited,
  author="Hoogeboom, Hendrik J.",
  title="Context-Free Valence Grammars - Revisited",
  booktitle="Developments in Language Theory",
  year="2002",
  publisher="Springer",
  pages="293--303",
  isbn="978-3-540-46011-4"
}

@article{valence,
  title = {Sequential grammars and automata with valences},
  journal = {Theoretical Computer Science},
  volume = {276},
  number = {1},
  pages = {377--405},
  year = {2002},
  issn = {0304-3975},
  author = {Henning Fernau and Ralf Stiebe}
}

@inproceedings{tgba,
  title={From States to Transitions: Improving Translation of {LTL} Formulae to {B\"u}chi Automata},
  author={Dimitra Giannakopoulou and Flavio Lerda},
  booktitle={Formal Techniques for (Networked and) Distributed Systems},
  year={2002}
}

@book{Clarke,
  author = {Clarke, Edmund M. and Grumberg, Orna and Peled, Doron A.},
  title = {Model checking},
  year = 1999,
  publisher = {The MIT Press}
}

@book{Baier2008,
  author = {Baier, Christel and Katoen, Joost-Pieter},
  publisher = {The MIT Press},
  title = {Principles of Model Checking},
  year = 2008
}

@book{ClarkeHandbook,
  author = {Clarke, Edmund M. and Henzinger, Thomas A. and Veith, Helmut and Bloem, Roderick},
  title = {Handbook of Model Checking},
  year = {2018},
  publisher = {Springer},
  edition = {1st}
}

@article{dickson,
 author = {Leonard E. Dickson},
 journal = {American Journal of Mathematics},
 number = {4},
 pages = {413--422},
 publisher = {Johns Hopkins University Press},
 title = {Finiteness of the Odd Perfect and Primitive Abundant Numbers with n Distinct Prime Factors},
 volume = {35},
 year = {1913}
}

@InProceedings{georgBlindCounter,
  author="Zetzsche, Georg",
  title="Silent Transitions in Automata with Storage",
  booktitle="Automata, Languages, and Programming",
  year="2013",
  publisher="Springer",
  pages="434--445"
}

@article{ibarra,
author = {Ibarra, Oscar H.},
title = {Reversal-Bounded Multicounter Machines and Their Decision Problems},
year = {1978},
issue_date = {Jan. 1978},
publisher = {Association for Computing Machinery},
volume = {25},
number = {1},
issn = {0004-5411},
journal = {J. ACM},
pages = {116-–133}
}

@article{bakerbook,
title = {Reversal-bounded multipushdown machines},
journal = {Journal of Computer and System Sciences},
volume = {8},
number = {3},
pages = {315--332},
year = {1974},
issn = {0022-0000},
author = {Brenda S. Baker and Ronald V. Book}
}

@article{latteux,
title = {Cônes rationnels commutatifs},
journal = {Journal of Computer and System Sciences},
volume = {18},
number = {3},
pages = {307--333},
year = {1979},
author = {Michel Latteux}
}

@InProceedings{logspaceconversion,
author={Baumann, Pascal and D'Alessandro, Flavio and Ganardi, Moses and Ibarra, Oscar and McQuillan, Ian and Sch{\"u}tze, Lia and Zetzsche, Georg},
title={Unboundedness Problems for Machines with Reversal-Bounded Counters},
booktitle={Foundations of Software Science and Computation Structures -- 26th
International Conference ({FOSSACS} 2023)},
year={2023},
publisher={Springer},
pages={240--264}
}

@article{parikh1966context,
  title={On context-free languages},
  author={Parikh, Rohit J.},
  journal={Journal of the ACM (JACM)},
  volume={13},
  number={4},
  pages={570--581},
  year={1966},
  publisher={ACM New York, NY, USA}
}

@book{minsky,
  author = {Minsky, Marvin L.},
  publisher = {Prentice-Hall},
  series = {Prentice-Hall Series in Automatic Computation},
  title = {Computation: Finite and Infinite Machines},
  year = 1967
}

@InProceedings{grobler2023remarks,
	author =	{Grobler, Mario and Sabellek, Leif and Siebertz, Sebastian},
	title =	{{Remarks on Parikh-Recognizable Omega-languages}},
	booktitle =	{32nd EACSL Annual Conference on Computer Science Logic (CSL 2024)},
	pages =	{31:1--31:21},
	series =	{Leibniz International Proceedings in Informatics (LIPIcs)},
	ISBN =	{978-3-95977-310-2},
	ISSN =	{1868-8969},
	year =	{2024},
	volume =	{288},
	publisher =	{Schloss Dagstuhl -- Leibniz-Zentrum für Informatik},
}

@article{ramsey,
author = {Bergstr\"{a}\ss{}er, Pascal and Ganardi, Moses and Lin, Anthony W. and Zetzsche, Georg},
title = {Ramsey Quantifiers in Linear Arithmetics},
year = {2024},
publisher = {Association for Computing Machinery},
volume = {8},
doi = {10.1145/3632843},
journal = {Proceedings of the {ACM} on Programming Languages}
}

@InProceedings{buechiVASSsep,
	author = {Baumann, Pascal and Meyer, Roland and Zetzsche, Georg},
	title =	{{Regular Separability in B\"{u}chi VASS}},
	booktitle =	{40th International Symposium on Theoretical Aspects of Computer Science (STACS 2023)},
	pages =	{9:1--9:19},
	series =	{Leibniz International Proceedings in Informatics (LIPIcs)},
	ISBN =	{978-3-95977-266-2},
	ISSN =	{1868-8969},
	year =	{2023},
	volume =	{254},
	publisher =	{Schloss Dagstuhl -- Leibniz-Zentrum für Informatik}
}

@inproceedings{zvassnz,
author = {Haase, Christoph and Zetzsche, Georg},
title = {Presburger arithmetic with stars, rational subsets of graph groups, and nested zero tests},
year = {2021},
publisher = {IEEE Press},
booktitle = {Proceedings of the 34th Annual ACM/IEEE Symposium on Logic in Computer Science (LICS 2019)},
articleno = {58},
numpages = {14}
}

@article{landweber,
  title={Decision problems for $\omega$-automata},
  author={Lawrence H. Landweber},
  journal={Mathematical systems theory},
  year={1969},
  volume={3},
  pages={376-384},
}

@InProceedings{parikhsep,
	author =	{Clemente, Lorenzo and Czerwinski, Wojciech and Lasota, Slawomir and Paperman, Charles},
	title =	{{Regular Separability of Parikh Automata}},
	booktitle =	{44th International Colloquium on Automata, Languages, and Programming (ICALP 2017)},
	pages =	{117:1--117:13},
	series =	{Leibniz International Proceedings in Informatics (LIPIcs)},
	year =	{2017},
	volume =	{80},
	publisher =	{Schloss Dagstuhl -- Leibniz-Zentrum für Informatik}
}

@InProceedings{GaleStewart,
title = {Infinite Games with Perfect Information},
booktitle = {Contributions to the Theory of Games (AM-28), Volume II},
author = {David Gale and Frank M. Stewart},
publisher = {Princeton University Press},
pages = {245--266},
isbn = {9781400881970},
year = {1953}
}

@misc{parikhComplexity,
	title={Parikh Images of Regular Languages: Complexity and Applications}, 
	author={Anthony W. To},
	year={2010},
	eprint={1002.1464},
	archivePrefix={arXiv},
	primaryClass={cs.LO}
}

@article{semilin,
	ISSN = {00029947},
	author = {Seymour Ginsburg and Edwin H. Spanier},
	journal = {Transactions of the American Mathematical Society},
	number = {2},
	pages = {333--368},
	publisher = {American Mathematical Society},
	title = {Bounded Algol-Like Languages},
	volume = {113},
	year = {1964}
}

@article{semilinUpper1,
	author = {Itshak Borosh and Leon B. Treybig},
	journal = {Proceedings of the American Mathematical Society},
	number = {2},
	pages = {299--304},
	publisher = {American Mathematical Society},
	title = {Bounds on Positive Integral Solutions of Linear Diophantine Equations},
	urldate = {2024-04-10},
	volume = {55},
	year = {1976}
}

@article{semilinUpper2,
	author = {Joachim von zur Gathen and Malte Sieveking},
	journal = {Proceedings of the American Mathematical Society},
	number = {1},
	pages = {155--158},
	publisher = {American Mathematical Society},
	title = {A Bound on Solutions of Linear Integer Equalities and Inequalities},
	volume = {72},
	year = {1978}
}

@incollection{karp,
	author = {Karp, Richard},
	booktitle = {Complexity of Computer Computations},
	pages = {85--103},
	publisher = {Plenum Press},
	title = {Reducibility among combinatorial problems},
	year = 1972
}

@inproceedings{schwentikHorn,
	author = {Verma, Kumar N. and Seidl, Helmut and Schwentick, Thomas},
	title = {On the complexity of equational horn clauses},
	year = {2005},
	publisher = {Springer},
	booktitle = {Proceedings of the 20th International Conference on Automated Deduction (CADE 2020)},
	pages = {337--352},
	numpages = {16}
}

@inproceedings{huynhUpper,
	author = {Dung T. Huynh},
	title = {The Complexity of Semilinear Sets},
	year = {1980},
	publisher = {Springer},
	booktitle = {Proceedings of the 7th Colloquium on Automata, Languages and Programming (ICALP 1980)},
	pages = {324--337},
	numpages = {14}
}

@article{huynhUpperSimplified,
	author       = {Dung T. Huynh},
	title        = {{A Simple Proof for the $\Sigma_2^\P$ Upper
	Bound of the Inequivalence Problem for Semilinear Sets}},
	journal      = {Journal of Information Processing and Cybernetics},
	volume       = {22},
	number       = {4},
	pages        = {147--156},
	year         = {1986}
}

@inproceedings{intexpr,
	author = {Stockmeyer, Larry J. and Meyer, Albert R.},
	title = {Word problems requiring exponential time (Preliminary Report)},
	year = {1973},
	publisher = {Association for Computing Machinery},
	url = {https://doi.org/10.1145/800125.804029},
	booktitle = {Proceedings of the Fifth Annual ACM Symposium on Theory of Computing (STOC 1973)},
	pages = {1--9},
	numpages = {9},
}

@article{polyhierarchy,
	title = {The polynomial-time hierarchy},
	journal = {Theoretical Computer Science},
	volume = {3},
	number = {1},
	pages = {1--22},
	year = {1976},
	author = {Larry J. Stockmeyer}
}

@article{detBuchiCompl,
	title = {Complementing deterministic Büchi automata in polynomial time},
	journal = {Journal of Computer and System Sciences},
	volume = {35},
	number = {1},
	pages = {59-71},
	year = {1987},
	author = {Robert P. Kurshan}
}

@inproceedings{haaseNexp,
	author = {Haase, Christoph},
	title = {Subclasses of presburger arithmetic and the weak EXP hierarchy},
	year = {2014},
	isbn = {9781450328869},
	publisher = {Association for Computing Machinery},
	booktitle = {Proceedings of the Joint Meeting of the Twenty-Third EACSL Annual Conference on Computer Science Logic (CSL 2014) and the Twenty-Ninth Annual ACM/IEEE Symposium on Logic in Computer Science (LICS 2014)},
	articleno = {47},
	numpages = {10}
}

@phdthesis{cadilhac2013automates,
    author = {Micha{\"e}l Cadilhac},
    title = {Automata with a semilinear constraint},
    school = {Universit{\'e} de Montr{\'e}al},
    year = {2013}
}

@book{libkinElements,
  author = {Libkin, Leonid},
  isbn = {3540212027},
  publisher = {Springer},
  title = {Elements of Finite Model Theory},
  year = 2004
}

@InProceedings{haaseTaming,
  author =	{Chistikov, Dmitry and Haase, Christoph},
  title =	{{The Taming of the Semi-Linear Set}},
  booktitle =	{43rd International Colloquium on Automata, Languages, and Programming (ICALP 2016)},
  pages =	{128:1--128:13},
  series =	{Leibniz International Proceedings in Informatics (LIPIcs)},
  ISBN =	{978-3-95977-013-2},
  ISSN =	{1868-8969},
  year =	{2016},
  volume =	{55},
  publisher =	{Schloss Dagstuhl -- Leibniz-Zentrum für Informatik}
}

@book{chang2013model,
  title={Model Theory: Third Edition},
  author={Chang, Chen C. and J. Keisler, Howard},
  isbn={9780486310954},
  series={Dover Books on Mathematics},
  url={https://books.google.de/books?id=sRi0AAAAQBAJ},
  year={2013},
  publisher={Dover Publications}
}
